\documentclass[journal]{IEEEtran}
\usepackage{amssymb}
\usepackage{mathrsfs}
\usepackage{amsthm}
\usepackage{amsmath,amsfonts}
\usepackage{algorithmic}
\usepackage{algorithm}
\usepackage{array}
\usepackage{textcomp}
\usepackage{stfloats}
\usepackage{url}
\usepackage{verbatim}
\usepackage{graphicx}
\usepackage{subfigure}
\usepackage{cite}

\newtheorem{theorem}{Theorem}
\newtheorem{corollary}{Corollary}
\newtheorem{lemma}{Lemma}
\newtheorem{proposition}{Proposition}
\newtheorem{definition}{Definition}

\newtheorem{remark}{Remark}

\ifCLASSINFOpdf
\else
\fi
\begin{document}
%
\title{Robust prescribed-time coordination control of cooperative-antagonistic networks with disturbances}
%
%
%

\author{Zhen-Hua~Zhu, Huaiyu~Wu, Zhi-Hong~Guan, Zhi-Wei~Liu, Yang Chen, and Xiujuan Zheng
\thanks{This work was partially supported by the National Natural Science Foundation of China under grant 62073250. \textit{(Corresponding author: Huaiyu~Wu.)}}
\thanks{Z.-H.~Zhu, H.~Wu, Y. Chen and X. Zheng are with the Engineering Research Center for Metallurgical Automation and Measurement Technology of Ministry of Education, Wuhan University of Science and Technology, Wuhan 430081, China, and also with the Institute of Robotics and Intelligent Systems, Wuhan University of Science and Technology, Wuhan 430081, China (e-mail: zhuzhenhua@wust.edu.cn; wuhy@wust.edu.cn; chenyag@wust.edu.cn; zhengxj@wust.edu.cn)}
\thanks{Z.-H. Guan and Z.-W. Liu are with the School of Artificial Intelligence and Automation, Huazhong University of Science and Technology, Wuhan 430074, China (e-mail: zhguan@mail.hust.edu.cn; zwliu@hust.edu.cn).}}

\maketitle

\begin{abstract}
This article targets at addressing the robust prescribed-time coordination control (PTCC) problems for single-integrator cooperative-antagonistic networks (CANs) with external disturbances under arbitrary fixed signed digraphs without any structural constraints. Toward this end, the PTCC problems for nominal single-integrator CANs without disturbances are first investigated and a fully distributed control protocol with a time-varying gain, which grows to infinity as the time approaches the settling time, is proposed utilizing the relative states of neighboring agents. Then, based on the proposed control protocol for the nominal single-integrator CANs, a new second-order prescribed-time sliding mode control protocol is constructed to achieve accurate PTCC for single-integrator CANs in the presence of external disturbances. Using Lyapunov based analysis, sufficient conditions to guarantee the prescribed-time stability, bipartite consensus, interval bipartite consensus, and bipartite containment of single-integrator CANs without or with disturbances are, respectively, derived. In the end, numerical simulations are given to confirm the derived results.
\end{abstract}


%
\IEEEpeerreviewmaketitle

\section{Introduction}\label{sec1}
%
%
%
%
%
%

\IEEEPARstart{C}{ooperative-antagonistic} networks (CANs) that consist of a team of agents cooperatively or antagonistically interacting with each other find applications in diverse realms, such as robot competitions, biological systems, social networks, etc.\cite{meng2019extended} One fundamental feature of CANs is that interactions among agents are represented by signed graphs, with positive and negative edges characterizing respectively the antagonistic and cooperative interactions between agents. As an important subject in the research of CANs, coordination control, which aims at enforcing all agents to achieve a desired collective behavior via designing an appropriate control protocol based merely on local information, has gained great attention over the past decade\cite{li2022bipartite,chen2022secure,zhu2021collective}. Thus far, numerous types of coordination control issues have been studied for CANs, including stability, bipartite consensus, interval bipartite consensus, bipartite containment control, and so forth\cite{wu2022disagreement}.

For coordination control of CANs, an issue of practical importance is how to ensure that all agents accomplish desired collective behavior within a finite time. This is because finite-time control not only offers fast convergence speed, but also guarantees good system performances, such as robustness against uncertainty and disturbance rejection. Consequently, finite-time coordination control of CANs has emerged as an attractive topic in recent years and fruitful results have been developed\cite{meng2015finite,meng2015nonlinear,zhao2017adaptive,hu2019finite,lu2019finite,wang2018finite,liu2022neural}. Particularly, the finite-time bipartite consensus and stability issues for CANs with single-integrator dynamics over undirected connected signed graphs were considered in \cite{meng2015finite,meng2015nonlinear}. Two nonlinear protocols were proposed in \cite{lu2019finite} to copy with the finite-time bipartite consensus problem of single-integrator CANs under detail-balanced signed digraphs. The finite-time bipartite consensus and stability problems for first- and second-order CANs in the presence of external bounded disturbances over strongly connected signed digraphs were discussed in \cite{wang2018finite}. In \cite{liu2022neural}, the finite-time bipartite containment control problem was tackled for a class of nonaffine fractional-order CANs with disturbances and an adaptive neural network control scheme was presented.


It is worthy of noticing that the settling time for finite-time control protocols relies upon the initial conditions. This may seriously constrain the practical application of the aforementioned finite-time coordination control results since the information of initial conditions is often hard to obtain exactly. To overcome this defect, based on the fixed-time control approach proposed by Polyakov\cite{polyakov2011nonlinear}, a distributed nonlinear protocol was initially designed in \cite{meng2016signed} to settle the fixed-time stability and signed-average consensus  problems of single-integrator CANs under undirected connected signed graphs. It was shown that the settling time of the designed protocol is globally upper-bounded by a positive scalar irrelevant to the agents' initial states. Subsequently, increasing attention has been paid towards the fixed-time coordination control problems of CANs\cite{gong2019fixed,xu2022fixed,zhang2019bipartite,liu2019finite,guo2021command,zhu2022finite}. Specifically, a distributed control law with heterogeneous coupling gains was constructed to achieve fixed-time bipartite consensus tracking for fractional-order CANs with structurally balanced signed digraph topology in \cite{gong2019fixed}. In \cite{xu2022fixed}, the problems of fixed-time bipartite consensus and bipartite consensus tracking were addressed for nonlinear CANs with external disturbances. The authors in \cite{zhang2019bipartite} discussed the fixed-time output bipartite consensus tracking issue for heterogeneous linear CANs. The finite- and fixed-time bipartite consensus problems for single-integrator CANs over strongly connected and detail-balanced signed digraphs were examined in \cite{liu2019finite}. The problem of fixed-time bipartite containment control of nonlinear stochastic CANs with structurally balanced signed digraph topology which contain a spanning forest was treated in \cite{guo2021command}. In  \cite{zhu2022finite}, a unifying framework was proposed for finite- and fixed-time bipartite containment control of first-order CANs over arbitrary weakly connected signed digraphs.

It is worth noticing, however, that, though irrelevant to initial conditions, the settling time for fixed-time control protocols depends upon design parameters and cannot be arbitrarily preassigned. To circumvent these drawbacks, some efforts recently have been made to study the prescribed-time coordination control (PTCC) problems of CANs by applying the newly developed prescribed-time control methods in the literature, where the settling time can be uniformly preset without dependence upon design parameters and initial conditions. For instance, an event-triggered control law was constructed to solve the prescribed-time bipartite consensus problem of first-order CANs under structurally balanced undirected connected signed graphs in \cite{chen2020prescribed}. In \cite{gong2020distributed}, the prescribed-time bipartite consensus and interval bipartite consensus problems of first-order CANs under signed digraphs were addressed. Based on the Pontryagin's principle, prescribed-time bipartite consensus protocols with varying gain were proposed in \cite{zhao2021fixed} for single- and double-integrator CANs under structurally balanced signed digraphs. By utilizing the multi-step motion planning technique, the authors in \cite{zhou2020prescribed} investigated the predefined-time bipartite formation control problem for general linear CANs under structurally balanced signed digraph topology containing a spanning tree. Furthermore, the prescribed-time multi-scale bipartite consensus and stability problems for continuous- and discrete-time CANs with single-integrator dynamics under strongly connected signed digraphs were discussed in \cite{guo2021prescribed}. In \cite{li2021output} and \cite{ren2022predefined}, the predefined-time bipartite consensus tracking problem was, respectively, studied for second-order CANs with matched disturbances and high-order uncertain nonlinear CANs, where the topology subgraph among followers is supposed to be strongly connected and structurally balanced. In addition, the problem of predefined-time bipartite consensus tracking of multiple Euler-Lagrange systems was treated in \cite{tao2022predefined} with the assumption that the topology subgraph among followers is undirected and structurally balanced.

It should be pointed out that all of the PTCC results above are obtained under somewhat restrictive requirements on the topology graph or the topology subgraph among followers, such as sign-symmetry\cite{gong2020distributed,zhao2021fixed,zhou2020prescribed,guo2021prescribed}, undirected connectivity\cite{chen2020prescribed,tao2022predefined}, or strong connectivity\cite{li2021output,ren2022predefined}. Until now, the PTCC problems of CANs under general strongly, quasi-strongly, and weakly connected signed digraphs without any hypotheses on their sign patterns remain open to our best knowledge. On the other hand, it is recognized that in practical applications, the agents are inevitably affected by various external disturbances, which could cause performance degradation or failure of coordination control if not well dealt with. Therefore, it is necessary to further take external disturbances into consideration.

Motivated by the foregoing discussions, this research deals with the PTCC problems for single-integrator CANs with external disturbances under arbitrary static strongly, quasi-strongly, and weakly connected signed digraphs, without imposing any structural constraints. Inspired partly by the prescribed-time distributed control approach introduced in Reference \cite{wang2018prescribed}, a class of fully distributed control strategy is first proposed for PTCC of disturbance-free normal single-integrator CANs under any fixed signed digraphs. Then, a novel second-order prescribed-time sliding mode control protocol is developed to achieve robust PTCC for single-integrator CANs subject to external disturbances on the basis of the proposed control protocol for the nominal single-integrator CANs. The key novelty and contributions of this article are as follows.
\begin{enumerate}
  \item A novel unified design and analysis framework is provided for PTCC of single-integrator CANs with any fixed topologies. Furthermore, a second-order sliding mode based design framework is presented to address the robust PTCC problems with bounded external disturbances. To our best knowledge, it is the first effort to handle the coordination control problems with uniformly assignable convergence time for single-integrator CANs in the presence or absence of exogenous perturbations under arbitrary fixed signed digraphs, without imposing any structural constraints.
  \item A fully distributed continuous control protocol for PTCC of nominal single-integrator CANs is proposed without using any global information. In addition, a novel second-order sliding mode control protocol is devised for robust PTCC of single-integrator CANs with bounded external perturbations. Moreover, the settling times of presented control protocols can be explicitly specified a priori without relying on initial conditions, design parameters, and interaction topology among agents.
  \item Sufficient conditions are given for prescribed-time stability, signed-average consensus, bipartite consensus, interval bipartite consensus, and bipartite containment of single-integrator CANs with or without external disturbances. Our results significantly generalize the existing infinite/finite/fixed-time coordination control results reported in \cite{meng2016interval,meng2017bipartite,meng2015finite,meng2018convergence,wang2018finite,meng2019extended} and the classical prescribed-time average consensus, consensus tracking, and containment control results in  \cite{wang2018prescribed} for conventional cooperative networks under unsigned digraphs. In comparison to the existing results on PTCC of CANs \cite{gong2020distributed,zhao2021fixed,zhou2020prescribed,guo2021prescribed,chen2020prescribed,
      tao2022predefined,li2021output,ren2022predefined}, we remove the assumptions on (strong) connectivity, sign-symmetry, or structural balance of signed digraphs.
\end{enumerate}

The remainder of this article is outlined as follows. In Section \ref{sec2}, we present some preliminaries required for subsequent development and formally state the problems under consideration. In Section \ref{sec3}, the main results on PTCC of single-integrator CANs without and with perturbations are included. Numerical examples for demonstrating the validity of the results proposed are provided in Sections \ref{sec4} and Section \ref{sec5} concludes this article.

\textit{Notations}: $\mathbb{R}{^k}\left( {\mathbb{C}{^k}} \right)$, ${{\mathbb{R}}^{k\times r}}\left({{\mathbb{C}}^{k\times r}} \right)$, and ${\mathbb{R}} \left( {{\mathbb{R}}_{ \ge 0}}\right)$ denote respectively the sets of $k$-dimensional real (complex) vectors, $k\times r$ real (complex) matrices, and (nonnegative) real numbers. $I_k$ indicates the $k\times k$ unit matrix, $0$ denotes the zero matrix with compatible sizes, ${1_k} $ stands for the $k$-dimensional all-one vector, and $\rm{sign}\left(\cdot\right)$ refers to the sign function. We denote by $\mathscr{I}_{N} =\left\{ {1, 2,\ldots, N} \right\}$, by ${\mathbb{G}_N} = \left\{ {G = {\rm{diag}}\left\{ {{g_1},\ldots ,{g_N}} \right\}:{{g_k} \in \left\{\pm 1\right\},{k\in\mathscr{I}_{N}}}} \right\}$, by ${\rm{diag}}\left\{ {{\it{\Gamma}}_1,\ldots, {\it{\Gamma}}_N} \right\}$ a diagonal block matrix whose $k$th diagonal block is ${\it{\Gamma}}_k$, by $|\cdot|$ the absolute value of a scalar or matrix, and by ${\left\|\cdot \right\|}$ the Euclidian norm of a vector. For any Hermitian matrix $A\in {{\mathbb{C}}^{k\times k}}$, ${\lambda_{\max}}\left(A\right)$ (${\lambda_{\min}}\left(A\right)$) represents its largest (smallest) eigenvalue, and $A\succ0$($\succeq 0$) indicates that $A$ is positive (semi-)definite. Given any matrix $B =\left[ {{b_{ij}}} \right] \in {\mathbb{R}^{r \times r}}$, the comparison matrix of $B$ is defined by ${\mathscr{M}}\left( B \right) = \left[ {{m_{ij}}} \right]\in {\mathbb{R}^{r \times r}}$ with ${m_{ii}} = \left|{{b_{ii}}}\right|$ and ${m_{ij}} = -\left| {{b_{ij}}} \right|$, $i \ne j$. $B$ is termed a {\it{Z}}-matrix when ${{b_{ij}}}\leq 0, i\ne j$, and a singular (nonsingular) {\it{M}}-matrix if further its eigenvalues all possess nonnegative (positive) real parts. It is termed an {\it{H}}-matrix provided its comparison matrix ${\mathscr{M}}\left(B \right)$ is a singular or nonsingular {\it{M}}-matrix.

\section{Preliminaries and Problem Statement}\label{sec2}
In this section, we first introduce basic notions about signed graphs and some useful lemmas, and then state the problems under investigation.

\subsection{Basic Concepts of Signed Graphs}
A signed digraph (direct graph) $\mathscr{G}$ (of order $N$) is a triple $\left( {\mathscr{V},\mathscr{E},\mathscr{W}} \right)$, where $\mathscr{V} = \left\{ {v_k}: {k\in\mathscr{I}_{N}}\right\}$ and $\mathscr{E} \subseteq \left\{{\left({v_k,v_l}\right):{k,l \in \mathscr{I}_{N}}} \right\}$ represent, respectively, the node and edge sets, and $\mathscr{W} = \left[ {{w_{kl}}} \right] \in {\mathbb{R}^{N \times N}}$ is the weighted adjacency matrix such that ${w_{kl}}= 0$ iff $\left( {v_l,v_k} \right)\notin \mathscr{E}$ and ${w_{kl}}\ne 0$ otherwise. Clearly, $\mathscr{G}$ reduces to a conventional unsigned digraph when ${w_{kl}}\geq0, \forall k,l\in {\mathscr{I}_N}$. Suppose $\left( {v_k,v_k}\right)\notin \mathscr{E}$, $\forall k \in \mathscr{I}_{N}$, i.e., no self-loops exist in $\mathscr{G}$. Denote the index set of in-neighbors of node $v_k$  as ${\mathscr{N}_k} = \left\{ {l: {\left( {v_l,v_k} \right) \in \mathscr{E}}} \right\}$ . A signed digraph ${\mathscr{G}}^{\ast} = ({\mathscr{V}^{\ast}},{\mathscr{E}^{\ast}},{\mathscr{W}^{\ast}})$ is termed a subgraph of $\mathscr{G}$ provided ${\mathscr{V}^{\ast}} \subseteq \mathscr{V}$ and ${\mathscr{E}^{\ast}} \subseteq \mathscr{E}$. The Laplacian matrix of $\mathscr{G}$ is defined as $\mathscr{L}=\mathscr{D}-\mathscr{W}$ with $\mathscr{D} = {\rm{diag}}\left\{ {{{\mathscr{D}}_1}, \ldots, {{\mathscr{D}}_N}} \right\}$, where ${\mathscr{D}_k} = \sum\nolimits_{l \in {\mathscr{N}_k}} {\left| {{w_{kl}}} \right|}, k\in {\mathscr{I}_N}$. It is clear that ${\mathscr{M}}\left( \mathscr{L} \right){1_N} = {0}$. $\mathscr{G}$ is said undirected if ${w_{lk}} = {w_{kl}}$, $\forall k,l \in {\mathscr{I}_N}$. $\mathscr{G}$ is said sign-symmetric when ${w_{lk}}{w_{kl}} \ge 0$, $\forall l ,k \in {\mathscr{I}_N}$, and sign-asymmetric otherwise.\cite{altafini2012consensus} $\mathscr{G}$ is structurally balanced provided a bipartition $\left\{ {{\mathscr{V}^1}, {\mathscr{V}^2}} \right\}$ of $\mathscr{V}$ satisfying ${\mathscr{V}^1} \cap {\mathscr{V}^2} = \emptyset $ and ${\mathscr{V}^1} \cup {\mathscr{V}^2} = \mathscr{V}$ exists, so that ${w_{kl}} \ge 0$, $\forall {v_k},{v_l}\in {\mathscr{V}}^\imath$ and ${w_{kl}} \le 0$, $\forall {v_k} \in {\mathscr{V}}^\imath$, $\forall {v_l} \in {\mathscr{V}}^{3 - \imath}$, where $\imath \in \left\{ {1,2} \right\}$; and otherwise, it is structurally unbalanced. Notice that structurally balanced signed digraphs contain the conventional unsigned digraphs as a trivial case.

In $\mathscr{G}$, a collection of pairwise distinct nodes $v_{k_{0}},v_{k_{1}},\ldots,v_{k_{m}}$ so that $\left({v_{k_{\imath-1}},v_{k_{\imath}}}\right)\in \mathscr{E}, \imath = 1,\ldots,m$ is called a (directed) path from node $v_{k_0}$ to node $v_{k_m}$. If $\mathscr{G}$ admits paths between every two distinct nodes, we say that it is strongly connected. Note that strong connectivity degenerates into connectivity when $\mathscr{G}$ is undirected. $\mathscr{G}$ is termed quasi-strongly connected, if there is at least one vertex, called root, having paths to all other vertices. $\mathscr{G}$ is weakly connected provided the undirected graph, induced by replacing each edge of $\mathscr{G}$ with an undirected edge, is connected. Notice that the weak connectivity is a rather general connectivity condition, including the strong and quasi-strong connectivity as special cases. A maximal strongly connected subgraph of $\mathscr{G}$ with no incoming edges from nodes outside is called a closed strong component (CSC) of $\mathscr{G}$.


\subsection{Useful Lemmas}
\begin{lemma}\cite{meng2020convergence}\label{lemma1}
Given any signed digraph $\mathscr{G}$, it is structurally unbalanced (resp., balanced) iff there does not exist (resp., there exists) $G\in {\mathbb{G}_N}$ satisfying $G\mathscr{L}G=\mathscr{M}\left(\mathscr{L} \right)$.
\end{lemma}

\begin{lemma}\cite{zhu2020observer}\label{lemma2}
For any strongly connected signed digraph $\mathscr{G}$ whose Laplacian matrix is denoted as $\mathscr{L}$, it holds that:
\begin{enumerate}
\item $\mathscr{G}$ is structurally unbalanced iff $\mathscr{L}$ has every eigenvalue with positive real part;
\item $\mathscr{G}$ is structurally balanced iff 0 is a simple eigenvalue of $\mathscr{L}$ with $G{1_N}$ as corresponding eigenvector and all other eigenvalues have positive real parts, in which $G\in {\mathbb{G}_N}$ satisfies $G{\mathscr{L}}G = \mathscr{M}\left({\mathscr{L}}\right)$;
\item there exists some positive vector $p \in\mathbb{R}{^N} $ satisfying ${p^{\rm{T}}}\mathscr{M}\left(\mathscr{L}\right)= 0$ and ${p^{\rm{T}}}{1_N} = 1$.
\end{enumerate}
\end{lemma}

\begin{lemma}\cite{zhu2022finite}\label{lemma3}
Let $\mathscr{L}$ be the Laplacian matrix of a strongly connected signed digraph $\mathscr{G}$, and $p = {\left[ {{p_1},\ldots, {p_N}} \right]^{\rm{T}}}$ be a positive vector satisfying ${p^{\rm{T}}}\mathscr{M}\left(\mathscr{L}\right)= 0$ and ${p^{\rm{T}}}{1_N} = 1$. Then the following hold.
\begin{enumerate}
\item If $\mathscr{G}$ is structurally balanced, then $\bar{\mathscr{L}} \succeq 0$ with zero as a simple eigenvalue, in which $\bar{\mathscr{L}} = {{\left({P\mathscr{L} + {\mathscr{L}^{\rm{T}}}P} \right)} \mathord{\left/{\vphantom {{\left( {PL + {L^T}P} \right)} 2}} \right.
 \kern-\nulldelimiterspace} 2}$ with $P = {\rm{diag}}\{ {p_1}, \ldots ,{p_N}\}$. Moreover, the null space of $\bar{\mathscr{L}}$ is spanned by $G{1_N}$, in which $G\in {\mathbb{G}_N}$ fulfills $G{\mathscr{L}}G = \mathscr{M}\left({\mathscr{L}}\right)$. Furthermore, ${\xi^{\rm{T}}}{\bar{\mathscr{L}}} \xi \ge a\left({\mathscr{L}}\right){\xi^{\rm{T}}}P\xi$ holds for arbitrary $\xi\in {\mathbb{R}^N}$ fulfilling ${\xi^{\rm{T}}}Gp = 0$, in which $a\left(\mathscr{L}\right)= \mathop {\min_{{\xi^{\rm{T}}}Gp = 0,\xi \ne 0}} \frac{{{\xi^{\rm{T}}}\bar{\mathscr{L}}\xi}}{{{\xi^{\rm{T}}}P\xi}} > 0$.
\item If $\mathscr{G}$ is structurally unbalanced, there exists a diagonal matrix ${\it{\Omega}} = {\rm{diag}}\{ {\omega_1}, \ldots ,{\omega_N}\}$ with ${\omega_k}>0, k\in {\mathscr{I}_N}$ so that $\tilde{\mathscr{L}} \triangleq {{\it{\Omega}}\mathscr{L} + {\mathscr{L}^{\rm{T}}}{\it{\Omega}}}\succ0$.
\end{enumerate}
\end{lemma}

\begin{lemma}\cite{arik2000sufficient}\label{lemma4}
Given any nonsingular {\it{H}}-matrix $A \in {\mathbb{R}^{n \times n}}$, there always exists some positive diagonal matrix ${\it{\Sigma}}\in {\mathbb{R}^{n \times n}}$ satisfying ${{\it{\Sigma}} A + {A^{\rm{T}}}{\it{\Sigma}}}\succ0$.
\end{lemma}

\begin{lemma}\cite{boyd1994linear}\label{lemma5}
Given matrices ${S_1}\in {{\mathbb{R}}^{n\times n}}$, ${S_2}\in {{\mathbb{R}}^{m\times n}}$, and ${S_3}\in {{\mathbb{R}}^{m\times m}}$, if ${S_1} = S_1^{\rm{T}}$ and ${S_3} = S_3^{\rm{T}} \succ 0$, then ${S_1} - S_2^{\rm{T}}S_3^{ - 1}{S_2} \succ 0$ iff
$\left[ {\begin{array}{*{20}{c}}
{{S_1}}&{S_2^{\rm{T}}}\\
{{S_2}}&{{S_3}}
\end{array}} \right] \succ 0$ or $\left[ {\begin{array}{*{20}{c}}
{{S_3}}&{{S_2}}\\
{S_2^{\rm{T}}}&{{S_1}}
\end{array}} \right] \succ 0$.
\end{lemma}

\begin{definition}\cite{wang2018prescribed}
Consider the following system
\begin{equation}\label{eq1}
\dot \vartheta\left({t}\right) = g\left(t ,\vartheta\left({t}\right) \right), t \ge 0, \vartheta\left( 0 \right) = \vartheta_0,
\end{equation}
in which $\vartheta\left({t}\right)\in {\mathbb{R}^r}$ and $g:{{\mathbb{R}}_{\ge 0}}\times {\mathbb{R}^r}\to{\mathbb{R}^r}$ is a continuous function satisfying $g\left(t,0\right)=0$. Let $\vartheta \left(t,\vartheta_0 \right)$ be the solution to \eqref{eq1} with the initial value $\vartheta\left( 0 \right) = \vartheta_0$. Then the system \eqref{eq1} is called
\begin{enumerate}
\item globally finite-time stable if it is Lyapunov stable and for all $\vartheta_0\in {\mathbb{R}^r}$, there is a function $T:{\mathbb{R}^r}\to {{\mathbb{R}}_{\ge 0}}$, called the settling time function, so that ${\lim_{t\to T({\vartheta_0})}}\vartheta\left( {t,{\vartheta_0}} \right) = 0$ and $\vartheta\left(t, \vartheta_0\right) = 0, \forall t \geq T(\vartheta_0)$.
\item globally prescribed-time stable if it is globally finite-time stable and the settling time function $T$ is a constant which can be assigned arbitrarily.
\end{enumerate}
\end{definition}

The Lemma below plays an essential role in deriving the main results in Section 3.

\begin{lemma}\label{lemma6}
If there is a continuously differentiable positive definite Lyapunov function $V\left(\vartheta \left({t}\right) \right) :{\mathbb{R}^r} \to {\mathbb{R}_{\ge 0}}$ for system \eqref{eq1} so that
\[\dot V\left(\vartheta\left({t}\right)\right)\leq -{a}V{\left(\vartheta\left({t}\right)\right)} - b\frac{\dot{\varphi}\left( t,T \right)}{\varphi\left(t, T\right)} V{\left(\vartheta \left({t}\right)\right)},\;\forall t\ge 0, \]
with $a> 0$, $b > 0$, $\varphi\left(t, T\right) =
\begin{cases}
\frac{{{T^{\kappa}}}}{{{{\left( {T - t} \right)}^{\kappa}}}}, & t \in \left[ {0,T} \right)\\
{1,} & t \in \left[ {T, + \infty } \right)
\end{cases}$, and $\dot \varphi\left(t, T\right) = \begin{cases}
\frac{{\kappa}}{T}{\varphi\left(t, T\right)^{1 + \frac{1}{{\kappa}}}},& t \in \left[ {0,T} \right)\\
0,& t \in \left[ {T, + \infty } \right)
\end{cases}$, where $\kappa$ is an arbitrary real number fulfilling $\kappa > 2$, $ T > 0$ is a finite positive scalar that can be arbitrarily selected, and the derivative of $\varphi\left(t, T\right)$ at $t = T$ is the right-hand one, then the system \eqref{eq1} is globally prescribed-time stable with the settling time being $T$. In addition, there hold $V\left(\vartheta\left({t}\right)\right) \le {{\varphi\left(t, T\right)}^{-b}}{\mathrm{exp}^{ - at}}V\left(\vartheta_0\right),\forall t \in \left[ {0,T} \right)$ and $V\left(\vartheta\left({t}\right)\right) = 0, \forall t \in \left[ {T, + \infty } \right)$.
\end{lemma}
\begin{proof}
The proof could be done by following an analogous procedure to that of Lemma 1 in Reference \cite{wang2018prescribed}.
\end{proof}

\begin{remark}
The above Lemma \ref{lemma6} is a generalization of the Lemma 1 of Reference \cite{wang2018prescribed}, which is  for the case of $b = 2$ only. Moreover, it is worth mentioning that Lemma \ref{lemma6} is valid also for $\kappa > 0$.
\end{remark}

\subsection{Problem Statement}
Consider a CAN consisting of $N \left(N\geq2 \right)$ agents with dynamic given by
\begin{equation}\label{eq2}
{\dot{x}_k}\left({t}\right) = {u_k}\left( t \right) + {d_k}\left( t \right), \;\;k \in {\mathscr{I}_N},
\end{equation}
in which ${u_k}\left( {t} \right)\in {\mathbb{R}}$, ${x_k}\left({t}\right)\in \mathbb{R}$, and ${d_k}\left( t \right)\in {\mathbb{R}}$ represent, respectively, the control input, the state, and the external disturbance of the $k$th agent. Suppose there is a known positive scalar $\delta\le +\infty$ so that $|{d_k}\left({t}\right)|\le \delta,\forall k \in {\mathscr{I}_N}$. Letting ${d_k} = 0, k \in {\mathscr{I}_N}$, then the nominal CAN corresponding to \eqref{eq2} is obtained as
\begin{equation}\label{eq3}
{\dot{x}_k}\left({t}\right) = {{u}_k}\left({t}\right), \;\;k \in {\mathscr{I}_N}.
\end{equation}

The interaction topology among the $N$ agents is assumed to be fixed and modeled via a signed digraph $\mathscr{G}$, with each node corresponding to an agent. Let the Laplacian and adjacency matrices of $\mathscr{G}$ be denoted as $\mathscr{L}$ and $\mathscr{W}$, respectively. Notice that the results to be established for \eqref{eq2} and \eqref{eq3} can be easily extended to arbitrary high dimension by exploiting the Kronecker product. Hereinafter, the time variable $t$ will be omitted whenever no confusions occur.

Following, e.g., References \cite{meng2019extended,zhu2020observer}, we call agent $k$ a leader if its corresponding node $v_k$ lies within some CSC of $\mathscr{G}$, and call it a follower otherwise. Notice that this notion of leader contains the classical notion of leader referring to an isolated agent with no neighbors as a trivial case. Let us denote $\mathcal{L}$ and $\mathcal{F}$ as the sets of leaders and followers, respectively. Obviously, $ \mathcal{L} \cap \mathcal{F} = \emptyset$ and $\mathcal{L} \cup \mathcal{F} = {\mathscr{I}_N}$. Without loosing generality, suppose the agents indexed by $1, \ldots, K$ are leaders and the remaining agents are followers. It is easy to see that $1\leq K \leq N$ and $K = N$ iff $\mathscr{G}$ is strongly connected. Further, $K$ is equal to the number of roots in $\mathscr{G}$ and satisfies $1\leq K < N$ if $\mathscr{G}$ is quasi-strongly connected.

The primary purpose of this article is to deal with the robust PTCC problems of the CAN \eqref{eq2}. To this end, we are first devoted to proposing a distributed control protocol to address the PTCC problems of the disturbance-free nominal CAN described by \eqref{eq3}. We then contribute to developing a prescribed-time sliding mode control protocol to tackle the robust PTCC problems of the CAN \eqref{eq2} based on the control protocol for the nominal CAN \eqref{eq3}.


\begin{definition}(Prescribed-Time Stability)\label{defn2}
We say the CAN \eqref{eq2} attains prescribed-time stability in a preassigned finite time $T \in {{\mathbb{R}}_{ \ge 0}}$ if there is a suitable control protocol
${u_k}, k \in {\mathscr{I}_N}$ so that
${\lim}_{t \to T}{{x_k}\left( t \right)} = 0$ and ${{x_k}\left( t \right)}= 0,\forall t \ge T $
hold for $\forall k\in{\mathscr{I}_N}$.
\end{definition}

\begin{definition}(Prescribed-Time Bipartite Consensus)\label{defn3}
We say the CAN \eqref{eq2} accomplishes prescribed-time bipartite consensus in a preassigned finite time $T \in {{\mathbb{R}}_{ \ge 0}}$ if there exist a suitable control protocol
${u_k}, k \in {\mathscr{I}_N}$ so that
${\lim}_{t \to T} \left| {{x_k}\left( t \right)} \right| = {x^*}$ and $\left| {{x_k}\left( t \right)} \right| = {x^*},\forall t \ge T $
hold for $\forall k\in{\mathscr{I}_N}$, where ${x^*}>0$.
\end{definition}

\begin{definition}\label{defn4}(Prescribed-Time Interval Bipartite Consensus)
The CAN \eqref{eq2} under a quasi-strongly connected signed digraph $\mathscr{G}$ is said to attain prescribed-time interval bipartite consensus within a preassigned finite time $T \in {{\mathbb{R}}_{ \ge 0}}$ if there exist ${x_f}>0$ and a proper control protocol ${u_k}, k \in {\mathscr{I}_N}$ such that the following hold:
\begin{enumerate}
  \item $\mathop {\lim}\limits_{t \to T} \left| {{x_k}\left( t \right)} \right| = {x_f}$ and $\left| {{x_k}\left( t \right)} \right| = {x_f},\;\forall t \ge T,\; k \in \mathcal{L}$;
  \item $\mathop {\lim}\limits_{t \to T}\left| {{x_l}\left( t \right)} \right|\leq {x_f}$ and $\left| {{x_l}\left( t \right)} \right|\leq {x_f},\forall t\geq T, l\in \mathcal{F}$.
\end{enumerate}
\end{definition}

Notice that Definition \ref{defn4} contains Definition \ref{defn3} as a special case. Apparently, Definition \ref{defn4} reduces to prescribed-time bipartite consensus of Definition \ref{defn3} when $\mathop {\lim}_{t \to T}\left| {{x_l}\left( t \right)} \right|= {x_f}$ and $\left| {{x_l}\left( t \right)} \right|= {x_f},\forall t\geq T$ for $\forall l \in \mathcal{F}$.

\begin{definition}\label{defn5}(Prescribed-Time Bipartite Containment)
We say the CAN \eqref{eq2} under a weakly connected signed digraph $\mathscr{G}$ reaches prescribed-time bipartite containment within a preset finite time $T \in {{\mathbb{R}}_{ \ge 0}}$ if there is a suitable control protocol ${u_k}, k \in {\mathscr{I}_N}$ such that:
\begin{enumerate}
  \item the leaders in structurally balanced and unbalanced CSCs, respectively, reach prescribed-time bipartite consensus and stability in the preset finite time $T$;
  \item the followers converge toward the convex hull spanned by the symmetric converged states of all leaders within the preset finite time $T$, i.e., $\mathop {\lim}_{t \to T} \left[ {\mathop {\max }\limits_{k \in \mathcal{L}} \left| {{x_{k}}\left( t \right)} \right| -\left| {{x_{l}}\left( t \right)} \right|} \right]\geq 0$ and ${{x_{l}}\left( t \right)}\in \left[- \mathop {\max}\limits_{k \in \mathcal{L}} \left| {{x_{k}}\left( t \right)} \right|,\mathop {\max }\limits_{k \in \mathcal{L}} \left| {{x_{k}}\left( t \right)} \right|\right],\forall t\geq T$ hold for $\forall l \in \mathcal{F}$.
\end{enumerate}
\end{definition}

\section{Main results}\label{sec3}

\subsection{Prescribed-time coordination control without disturbances}
In this subsection, the PTCC problems of the nominal CAN \eqref{eq3} are addressed. The control protocol for each agent $k$, $k \in {\mathscr{I}_N}$ is designed as
\begin{equation}\label{eq4}
{{u}_k} = \left( {{\rho_1} + {\rho_2}\frac{{\dot \varphi\left(t, T\right)}}{\varphi\left(t, T\right)}} \right)\sum\limits_{l \in {{\mathcal{N}}_k}} {{w_{kl}}\left[ {{x_l} - {\rm{sign}}\left( {{w_{kl}}} \right){x_k}} \right]},
\end{equation}
where ${\rho_1}>0$, ${\rho_2}>0$, ${\varphi\left(t, T\right)}$ and ${\dot \varphi\left(t, T\right)}$ are defined in Lemma \ref{lemma6} with $T$ specified as $T = T_1 > 0$, and ${w_{kl}}$ is the $\left({k,l} \right){\rm{th}}$ element of the adjacency matrix $\mathscr{W}$. It is worth highlighting that the protocol \eqref{eq4} is completely distributed and scalable since it relies upon merely the relative states of neighboring agents, without exploiting any global information.

Note that, in the special case when ${\rho_1} =1$ and ${\rho_2}=0$, the above protocol \eqref{eq4} reduces to the typical Laplacian-type protocol studied in References \cite{altafini2012consensus,meng2016interval,meng2018convergence,meng2019extended}, which was proved to solve the asymptotic stability, bipartite consensus, interval bipartite consensus, and bipartite containment control problems of the nominal CAN \eqref{eq3} under suitable topology conditions. It is noteworthy that the convergence analysis approaches presented in References \cite{altafini2012consensus,meng2016interval,meng2018convergence,meng2019extended} are no longer applicable here.


Let ${e_k} = \sum\nolimits_{l \in {{{\cal N}}_k}} {{w_{kl}}\left[ {{x_l} - {\rm{sign}}\left( {{w_{kl}}} \right){x_k}} \right]}, k \in {\mathscr{I}_N}$. Further, denote $E = {\left[ {{e_1},\ldots,{e_N}} \right]^{\rm{T}}}$ and $X = {\left[ {{x_1},\ldots,{x_N}} \right]^{\rm{T}}}$. Then we have
\begin{equation}\label{eq5}
E = - \mathscr{L}X.
\end{equation}
In view of \eqref{eq3} and \eqref{eq4}, we clearly have
\begin{equation*}
{\dot{x}_k} = \left( {{\rho_1} + {\rho_2}\frac{{\dot \varphi\left(t, T_1\right)}}{\varphi\left(t, T_1\right)}} \right)\sum\limits_{l \in {{\mathcal{N}}_k}} {{w_{kl}}\left[ {{x_l} - {\rm{sign}}\left( {{w_{kl}}} \right){x_k}} \right]},
\end{equation*}
$k \in {\mathscr{I}_N}$, which can be represented compactly as
\begin{equation}\label{eq7}
\dot{X} = -\left( {{\rho_1} + {\rho_2}\frac{{\dot \varphi\left(t, T_1\right)}}{\varphi\left(t, T_1\right) }} \right)\mathscr{L}X.
\end{equation}
This, together with \eqref{eq5}, yields
\begin{equation}\label{eq8}
\dot{E} = -\left( {{\rho_1} + {\rho_2}\frac{{\dot \varphi\left(t, T_1\right)}}{\varphi\left(t, T_1\right) }} \right)\mathscr{L}E.
\end{equation}


First we consider the scenario where $\mathscr{G}$ is strongly connected. The theorem below reveals that with the control protocol \eqref{eq4}, the nominal CAN \eqref{eq3} under arbitrary strongly connected signed digraph $\mathscr{G}$  can be ensured to reach prescribed-time bipartite consensus (resp., stability) in the pre-specified finite time $T_1$, if $\mathscr{G}$ is structurally balanced (resp., unbalanced).

\begin{theorem}\label{thm1}
Consider the nominal CAN \eqref{eq3} under the control protocol \eqref{eq4}, and let $\mathscr{G}$ be strongly connected. Then, the nominal CAN \eqref{eq3} achieves
\begin{enumerate}
  \item prescribed-time bipartite consensus in the pre-specified finite time $T_1$, if $\mathscr{G}$ is structurally balanced.
  \item prescribed-time stability in the pre-specified finite time $T_1$, if $\mathscr{G}$ is structurally unbalanced.
\end{enumerate}
\end{theorem}
\begin{proof}
1) Because $\mathscr{G}$ is structurally balanced, there is $G = {\rm{diag}}\left\{ {{g_1}, \ldots ,{g_N}} \right\}\in {\mathbb{G}_N}$ such that $G\mathscr{L}G=\mathscr{M}\left(\mathscr{L} \right)$ by Lemma \ref{lemma1}.
Then, utilizing $\mathscr{M}\left(\mathscr{L} \right){1_N} = 0$ and $G = {G^{-1}}$, we have $\mathscr{L}G{1_N} = 0$. Because $\mathscr{G}$ is strongly connected, it follows from Lemma \ref{lemma2} that there is a positive vector $p = {\left[ {{p_1}, \ldots, {p_N}} \right]^{\rm{T}}}\in\mathbb{R}{^N}$ fulfilling ${p^{\rm{T}}}{1_N} = 1$ such that ${p^{\rm{T}}}\mathscr{M}\left(\mathscr{L}\right)= 0$. This, together with $G\mathscr{L}G=\mathscr{M}\left(\mathscr{L} \right)$ and $G = {G^{-1}}$, implies ${p^{\rm{T}}}G\mathscr{L}=0$. Let ${\epsilon_k} = {x_k} - {g_k}\sum\nolimits_{j = 1}^N {{p_j}{g_j}} {x_j}, k \in {\mathscr{I}_N}$, and $\epsilon  = {\left[ {{\epsilon _1}, \ldots ,{\epsilon _N}} \right]^{\rm{T}}}$. It is obvious that $\epsilon  = X - \left( {{p^{\rm{T}}}GX} \right)G{1_N}$. Using \eqref{eq7} and ${p^{\rm{T}}}G\mathscr{L}=0$, we can obtain
${p^{\rm{T}}}G\dot X = 0$, which implies $\dot \epsilon  = \dot X$. Noting $\mathscr{L}G{1_N} = 0$, we have
$\mathscr{L}\epsilon  = \mathscr{L}X - \left( {{p^{\rm{T}}}GX} \right)\mathscr{L}G{1_N} = \mathscr{L} X$.

Choose the Lyapunov function candidate as:
$$\bar{{V_1}} = {\epsilon^{\rm{T}}}P\epsilon,$$
where $P = {\rm{diag}}\{{p_1}, \ldots, {p_N}\}$. Obviously, ${\bar{V_1}}$ is positive definite. In view of \eqref{eq3}, we know that $X$ is continuous w.r.t. $t$ on $\left[ {0, + \infty } \right)$. Since $\epsilon  = X - \left( {{p^{\rm{T}}}GX} \right)G{1_N}$, it follows that $\epsilon$ is also continuous w.r.t. $t$ on $\left[ {0, + \infty } \right)$. Therefore, ${\bar{V_1}}$ is continuous w.r.t. $t$ on $\left[ {0, + \infty } \right)$. Using \eqref{eq7}, together with $\dot \epsilon = \dot X$ and $\mathscr{L}X = \mathscr{L}\epsilon $, we can obtain
\begin{equation}\label{eq9}
\begin{split}
\dot{\bar{V_1}}& = 2{\epsilon ^{\rm{T}}}P\dot \epsilon  = 2{\epsilon ^{\rm{T}}}P\dot X\\
& =  - 2\left( {{\rho_1}  + {\rho_2} \frac{{\dot \varphi\left(t, T_1\right)}}{\varphi\left(t, T_1\right) }} \right){\epsilon ^{\rm{T}}}P\mathscr{L}X\\
& =  - 2{\rho_1} {\epsilon ^{\rm{T}}}P\mathscr{L}\epsilon  - 2{\rho_2} \frac{{\dot \varphi\left(t, T_1\right) }}{\varphi\left(t, T_1\right)}{\epsilon ^{\rm{T}}}P\mathscr{L}\epsilon \\
& =  - 2{\rho_1} {\epsilon ^{\rm{T}}}\bar{\mathscr{L}}\epsilon  - 2{\rho_2} \frac{{\dot \varphi\left(t, T_1\right)}}{\varphi\left(t, T_1\right)}{\epsilon ^{\rm{T}}}\bar{\mathscr{L}}\epsilon,
\end{split}
\end{equation}
where $\bar{\mathscr{L}}$ is as defined in Lemma \ref{lemma3}. Noting that ${p^{\rm{T}}}{1_N} = 1$ and ${G^2} = I_N$, we have ${p^{\rm{T}}}G\epsilon  = {p^{\rm{T}}}G X - {p^{\rm{T}}}GG{1_N}\left( {{p^{\rm{T}}}GX} \right) = 0$,
which, by Lemma \ref{lemma3}, implies that ${\epsilon^{\rm{T}}}\bar{\mathscr{L}}\epsilon \ge a\left(\mathscr{L}\right){\epsilon^{\rm{T}}}P\epsilon$, where $a\left(\mathscr{L}\right)>0$ is defined as in Lemma \ref{lemma3}. It then follows from \eqref{eq9} that
\begin{equation}\label{eq10}
\begin{split}
\dot{\bar{V_1}}\le &- 2{\rho_1} a\left(\mathscr{L}\right){\epsilon^{\rm{T}}}P\epsilon - 2{\rho_2} a\left(\mathscr{L}\right)\frac{{\dot \varphi\left(t, T_1\right) }}{\varphi\left(t, T_1\right) }{\epsilon^{\rm{T}}}P\epsilon \\
&= - 2{\rho_1} a\left(\mathscr{L}\right){{\bar V}_1} - 2{\rho_2} a\left(\mathscr{L}\right)\frac{{\dot \varphi\left(t, T_1\right)}}{\varphi\left(t, T_1\right) }{{\bar V}_1}.
\end{split}
\end{equation}
By using Lemma \ref{lemma6}, one can obtain from \eqref{eq10} that
\begin{equation*}
{\bar V_1}\left( t \right)\left\{ \begin{array}{l}
 \le {\varphi\left(t, T_1\right) ^{ - 2{\rho_2} a\left(\mathscr{L}\right)}}{\exp ^{ - 2{\rho_1} a\left(\mathscr{L}\right)t}}{{\bar V}_1}\left( 0 \right), t \in \left[ {0,T_1} \right)\\
 \equiv 0, t \in \left[ {T_1, + \infty } \right).
\end{array} \right.
\end{equation*}
This implies that ${\left\|\epsilon\right\|}\le \sqrt {\frac{{{\lambda_{{\max}}}\left( P \right)}}{{{\lambda _{{\min}}}\left( P \right)}}} {\varphi\left(t, T_1\right) ^{ - {\rho_2} a\left(\mathscr{L}\right)}}{\exp ^{ - {\rho_1} a\left(\mathscr{L}\right)t}}{\left\| {\epsilon\left( 0 \right)} \right\|}$ on $\left[ {0,T_1} \right)$, and $\epsilon \equiv 0$ on $\left[ {T_1, + \infty } \right)$. Because ${\lim _{t \to {T_1^ - }}}{\varphi\left(t, T_1\right) ^{ - {\rho_2} a\left(\mathscr{L}\right)}} = 0$, we can get
${\lim _{t \to {T_1^ - }}}{\left\| \epsilon \right\|} = 0$, which gives ${\lim_{t \to {T_1^-}}}{\epsilon} = 0$. Thus, it follows that ${\lim _{t \to {T_1}}}{\epsilon} = 0$ and $\epsilon = 0$,$\forall t \geq T_1$. Therefore,  prescribed-time bipartite consensus of the nominal CAN \eqref{eq3} is achieved within the prescribed finite time $T_1$.

2) Choose the Lyapunov function candidate as
$${\tilde{V_1}} = {E^{\rm{T}}}{\it{\Omega}} E,$$
with ${\it{\Omega}}$ given in Lemma \ref{lemma3}. Since $E = - \mathscr{L}X$ and $X$ is continuous w.r.t. $t$ on $\left[ {0, + \infty } \right)$, we know that ${\tilde{V_1}}$ is continuous w.r.t. $t$ on $\left[ {0, + \infty } \right)$. Differentiating $\tilde{V_1}$ along \eqref{eq8} yields
\begin{equation}\label{eq11}
\begin{split}
{{\dot{\tilde V}_1}} &= 2{E^{\rm{T}}}{\it{\Omega}}\dot E\\
& = - 2{\rho_1} {E^{\rm{T}}}{\it{\Omega}}{\mathscr{L}}E - 2{\rho_2} \frac{{\dot \varphi\left(t, T_1\right) }}{\varphi\left(t, T_1\right) }{E^{\rm{T}}}{\it{\Omega}}\mathscr{L}E\\
& = -{\rho_1} {E^{\rm{T}}}\tilde{\mathscr{L}}E - {\rho_2} \frac{{\dot \varphi\left(t, T_1\right) }}{\varphi\left(t, T_1\right)}{E^{\rm{T}}}\tilde{\mathscr{L}}E\\
&\leq -{\rho_1} {{\lambda_{{\min}}}\left(\tilde{\mathscr{L}}\right)} {E^{\rm{T}}}E - {\rho_2} \frac{{\dot \varphi\left(t, T_1\right) }}{\varphi\left(t, T_1\right)}{{\lambda_{{\min}}}\left(\tilde{\mathscr{L}}\right)}{E^{\rm{T}}}E\\
&\le -{\rho_1} \frac{{{\lambda_{{\min}}}\left( {\tilde{\mathscr{L}}} \right)}}{{{\lambda_{{\max}}}\left( {\it{\Omega}} \right)}}{\tilde V_1} - {\rho_2} \frac{{{\lambda_{{\min}}}\left( {\tilde{\mathscr{L}}} \right)}}{{{\lambda_{{\max}}}\left( {\it{\Omega}}\right)}}\frac{{\dot \varphi\left(t, T_1\right) }}{\varphi\left(t, T_1\right)}{\tilde V_1},
\end{split}
\end{equation}
where $\tilde{\mathscr{L}}$ is defined as in Lemma \ref{lemma3}. By using Lemma \ref{lemma6}, we can get
from \eqref{eq11} that
\begin{equation*}
{\tilde V_1}\left( t \right)\left\{ \begin{array}{l}
 \le {\varphi\left(t, T_1\right) ^{-{\rho_2} \frac{{{\lambda_{{\min}}}\left( {\tilde{\mathscr{L}}} \right)}}{{{\lambda_{{\max}}}\left({\it{\Omega}}\right)}}}}{\exp ^{-{\rho_1}{\frac{{{\lambda_{{\min}}}\left( {\tilde{\mathscr{L}}} \right)}}{{{\lambda_{{\max}}}\left({\it{\Omega}}\right)}}}t}}{{\tilde V}_1}\left( 0 \right),{\rm{ }}t \in \left[ {0,T_1} \right)\\
 \equiv 0,  \;\;\;\;\;\;\;\;\;\;\;\;\;\;\;\;\;\;\;\;\;\;\;\;\;\;\;\;\;\;\; t \in \left[ {T_1, + \infty } \right).
\end{array} \right.
\end{equation*}
This implies that $\left\| E \right\|^2 \le \frac{{{\lambda_{{\max}}}\left({\it{\Omega}}\right)}}{{{\lambda_{{\min}}}\left({\it{\Omega}}  \right)}}{\varphi\left(t, T_1\right) ^{ - {\rho_2} \frac{{{\lambda_{{\min}}}\left( {\tilde{\mathscr{L}}} \right)}}{{{\lambda_{{\max}}}\left({\it{\Omega}}\right)}}}}{\exp ^{ - {\rho_1} \frac{{{\lambda_{{\min}}}\left( {\tilde{\mathscr{L}}} \right)}}{{{\lambda_{{\max}}}\left({\it{\Omega}}\right)}}t}}\left\| {E\left( 0 \right)} \right\|^2$ on $\left[ {0,T_1} \right)$, and $E \equiv 0$ on $\left[ {T_1, + \infty } \right)$. Due to ${\lim _{t \to {T_1^ - }}}{\varphi\left(t, T_1\right) ^{ - {\rho_2} \frac{{{\lambda_{{\min}}}\left( {\tilde{\mathscr{L}}} \right)}}{{{\lambda_{{\max}}}\left({\it{\Omega}}\right)}}}} = 0$, we can obtain
${\lim _{t \to {T_1^ - }}}{\left\| E \right\|^2} = 0$, which yields ${\lim_{t \to {T_1^-}}}{E} = 0$. Thus, it follows that
\begin{equation}\label{eq12}
\left\{ \begin{array}{l}
{\lim_{t\to {T_1}}}{E} = 0 \\
{{E}} = 0,\forall t \ge T_1.
\end{array} \right.
\end{equation}
Since ${\mathscr{L}}$ is nonsingular by the first conclusion of Lemma \ref{lemma2}, it follows from \eqref{eq5} that $E = 0$ iff $X = 0$. This, together with \eqref{eq12}, clearly gives that ${\lim_{t\to {T_1}}}{X} = 0 \; {\rm{and}}\; X = 0, \forall t\geq T_1$. The proof is done.
\end{proof}

The Theorem \ref{thm1} above has the following corollaries.

\begin{corollary}\label{corol1}
Consider the nominal CAN \eqref{eq3} under the control protocol \eqref{eq4}, and let $\mathscr{G}$ be undirected and connected. Then, the nominal CAN \eqref{eq3} achieves
\begin{enumerate}
  \item  prescribed-time signed-average consensus in the pre-specified finite time $T_1$, if $\mathscr{G}$ is structurally balanced.
  \item  prescribed-time stability in the pre-specified finite time $T_1$, if $\mathscr{G}$ is structurally unbalanced.
\end{enumerate}
\end{corollary}
\begin{proof}
(1) Since $\mathscr{G}$ is undirected and connected, one has from Lemma \ref{lemma2} that $p =\frac{1}{N}{1_N}$. This leads to
$\varepsilon  = X - \frac{{1_N^{\rm{T}}GX}}{N}G{1_N}$. The rest proof follows the same line as that of the first conclusion of Theorem \ref{thm1}.

(2) It follows readily from the second conclusion of Theorem \ref{thm1} and the fact that undirected graphs are included as a special case of directed graphs.
\end{proof}

\begin{remark}
Corollary \ref{corol1} significantly extends the results in References \cite{meng2015nonlinear} and \cite{meng2016signed}, where the finite-time and fixed-time signed-average consensus and stability problems are, respectively, addressed for first-order CANs with connected undirected signed graph topology. Furthermore, the prescribed-time signed-average consensus result of Corollary \ref{corol1} includes the classical prescribed-time average consensus result of Theorem 1 in Reference \cite{wang2018prescribed} as a special case, since traditional cooperative networks is a trivial case of structurally balanced CANs.
\end{remark}

\begin{corollary}\label{corol2}
Consider the nominal CAN \eqref{eq3} under the control protocol \eqref{eq4}. Let $\mathscr{G}$ be strongly connected with its edge weights all positive. Then, the nominal CAN \eqref{eq3} reaches prescribed-time consensus in the pre-specified finite time $T_1$.
\end{corollary}

\begin{proof}
Since ${w_{kl}}\geq0,\forall k,l\in {\mathscr{I}_N}$, we have that $\mathscr{G}$ is structurally balanced with $G = I$ or $G = - I$. This implies $\epsilon = X - \left( {{p^{\rm{T}}}X} \right){1_N}$. The rest of the proof is the same as that of the first conclusion of Theorem \ref{thm1} and is thus omitted.
\end{proof}

Next, we consider the scenario where $\mathscr{G}$ is quasi-strongly connected. Since no directed path exists from arbitrary follower $l \in \mathcal{F}$ to arbitrary leader $k\in \mathcal{L}$, the adjacency matrix $\mathscr{W}$ can be expressed as
\begin{equation*}
\mathscr{W} = \left[ {\begin{array}{*{20}{c}}
{{\mathscr{W}_{L}}}&0\\
{{\mathscr{W}_{{FL}}}}&{{\mathscr{W}_{F}}}
\end{array}} \right],
\end{equation*}
where ${\mathscr{W}_{L}}\in {\mathbb{R}^{K \times K}}$ (resp., ${\mathscr{W}_{F}}\in {\mathbb{R}^{\left(N-K\right) \times \left(N-K\right)}}$) is the adjacency matrix of the subgraph $\mathscr{G}_{L}$ (resp., $\mathscr{G}_{F}$) of $\mathscr{G}$ by removing all the nodes and edges related to the followers (resp., leaders), and ${\mathscr{W}_{FL}}\in {\mathbb{R}^{\left(N-K\right)\times K}}$. Correspondingly, the Laplacian matrix $\mathscr{L}$ can be decomposed as
\begin{equation}\label{eq13}
\mathscr{L} = \left[ {\begin{array}{*{20}{c}}
{{\mathscr{L}_{L}}}&0\\
{{\mathscr{L}_{{FL}}}}&{{\mathscr{L}_{F}}}
\end{array}} \right],
\end{equation}
with ${\mathscr{L}_{L}}\in {\mathbb{R}^{K \times K}}$, ${\mathscr{L}_{F}}\in {\mathbb{R}^{\left(N-K\right) \times \left(N-K\right)}}$, and ${\mathscr{L}_{FL}}\in {\mathbb{R}^{\left(N-K\right)\times K}}$. Denote ${E_{L}} = \left[ {{e_1}, \ldots ,{e_K}} \right]^{\rm{T}}$ and ${E_{F}} = \left[ {{e_{K + 1}}, \ldots ,{e_N}} \right]^{\rm{T}}$. Clearly, $E = {\left[ {E_{L}^{\rm{T}},E_{F}^{\rm{T}}} \right]^{\rm{T}}}$. Let ${X_{L}} = \left[ {{x_1}, \ldots ,{x_K}} \right]^{\rm{T}}$ and ${X_{F}} = \left[ {{x_{K + 1}}, \ldots,{x_N}} \right]^{\rm{T}}$. Then, using \eqref{eq5} and \eqref{eq13}, we can obtain
\begin{align}
{{E}_{L}}& =  - {\mathscr{L}_{L}}{X_{L}},\label{eq14} \\
{{E}_{F}}& = - {\mathscr{L}_{F}}{X_{F}} - {\mathscr{L}_{{FL}}}{X_{L}}.\label{eq15}
\end{align}
In view of \eqref{eq7} and \eqref{eq13}, we have
\begin{align}
{{\dot X}_{L}} &=  - \left( {{\rho_1}  + {\rho_2} \frac{{\dot \varphi\left(t, T_1\right) }}{\varphi\left(t, T_1\right) }} \right){\mathscr{L}_{L}}{X_{L}},\label{eq16} \\
{{\dot X}_{F}} &= - \left( {{\rho_1}  + {\rho_2} \frac{{\dot \varphi\left(t, T_1\right) }}{\varphi\left(t, T_1\right) }} \right)\left({\mathscr{L}_{F}}{X_{F}} +{\mathscr{L}_{{FL}}}{X_{L}}\right).\nonumber
\end{align}
By considering  \eqref{eq8} and \eqref{eq13}, we have
\begin{align}
{{\dot E}_{L}} &=  - \left( {{\rho_1}  + {\rho_2} \frac{{\dot \varphi\left(t, T_1\right) }}{\varphi\left(t, T_1\right) }} \right){\mathscr{L}_{L}}{E_{L}},\label{eq17} \\
{{\dot E}_{F}} &= - \left( {{\rho_1}  + {\rho_2} \frac{{\dot \varphi\left(t, T_1\right) }}{\varphi\left(t, T_1\right) }} \right)\left({\mathscr{L}_{F}}{E_{F}}+{\mathscr{L}_{{FL}}}{E_{L}}\right).\label{eq18}
\end{align}

Before proceeding, the lemmas below are introduced.

\begin{lemma}\label{lemma7}
Let $\mathscr{G}$ be quasi-strongly connected with its Laplacian matrix $\mathscr{L}$ given as \eqref{eq13}. Then the following hold.
\begin{enumerate}
  \item ${\mathscr{L}_{F}}$ is a nonsingular {\it{H}}-matrix.
  \item There exists some positive diagonal matrix ${\it{\Xi}}_{F}\in {\mathbb{R}^{\left(N-K\right) \times \left(N-K\right)}}$ such that ${{\it{\Xi}}_{F}{\mathscr{L}_{F}} + {\mathscr{L}_{F}^{\rm{T}}}{{\it{\Xi}}}_{F}}\succ0$.
\end{enumerate}
\end{lemma}
\begin{proof}
1) Since $\mathscr{G}$ is quasi-strongly connected, the eigenvalues of ${\mathscr{L}_{F}}$ all possess positive real parts by Lemma 5.3 of Reference \cite{meng2018convergence}. Hence, ${\mathscr{L}_{F}}$ is nonsingular. Further, in view of Lemma 5.4  of Reference \cite{meng2018convergence}, we have that ${\mathscr{M}}\left({\mathscr{L}_{F}}\right)$ is an {\it{M}}-matrix. This indicates that ${\mathscr{L}_{F}}$ is an {\it{H}}-matrix. Therefore, ${\mathscr{L}_{F}}$ is a nonsingular {\it{H}}-matrix.

2) It straightforwardly follows from Lemma \ref{lemma4} and the fact that ${\mathscr{L}_{F}}$ is a nonsingular {\it{H}}-matrix.
\end{proof}

\begin{lemma}\cite{meng2018convergence}\label{lemma8}
Let $\mathscr{G}$ be quasi-strongly connected with its Laplacian matrix $\mathscr{L}$ given in \eqref{eq13}. Suppose $\mathscr{G}_{L}$ is structurally balanced, and define $\zeta  = -{\mathscr{L}_{F}^{-1}}{\mathscr{L}_{FL}}{{G}_{L}}{1_{{K}}}$, where ${G_{L}} \in {\mathbb{G}_{K}}$ such that ${{G}_{L}}{\mathscr{L}_{L}}{{G}_{L}}=\mathscr{M}\left({\mathscr{L}_{L}}\right)$, then $\left|\zeta \right|\le 1_{N-K}$. Moreover, if $\mathscr{G}$ is structurally balanced, then $\left|\zeta \right|= 1_{N-K}$ with $\zeta = {\mathbb{G}_{F}}{1_{N-K}}$, where ${\mathbb{G}_{F}}\in {\mathbb{G}_{N-K}}$ such that ${{\rm{diag}}\left\{ {{G_L},{G_F}} \right\}}{\mathscr{L}}{{\rm{diag}}\left\{ {{G_L},{G_F}} \right\}}=\mathscr{M}\left({\mathscr{L}}\right)$.
\end{lemma}

Based on Lemmas \ref{lemma7} and \ref{lemma8}, the following theorem can be established, which shows that for the CAN \eqref{eq3} under a quasi-strongly signed graph $\mathscr{G}$, the protocol \eqref{eq4} can ensure the CAN to reach prescribed-time interval bipartite consensus (resp., stability) in the pre-specified finite time $T_1$, if $\mathscr{G}_{L}$ is structurally balanced (resp., unbalanced).

\begin{theorem}\label{thm2}
Consider the nominal CAN \eqref{eq3} under the control protocol \eqref{eq4}, and let $\mathscr{G}$ be quasi-strongly connected. Then,  the nominal CAN \eqref{eq3} reaches
\begin{enumerate}
  \item  prescribed-time interval bipartite consensus in the pre-specified finite time $T_1$, if the subgraph $\mathscr{G}_{L}$ of $\mathscr{G}$ is structurally balanced.
  \item  prescribed-time stability in the pre-specified finite time $T_1$, if the subgraph $\mathscr{G}_{L}$ of $\mathscr{G}$ is structurally unbalanced.
\end{enumerate}
\end{theorem}
\begin{proof}
1) The proof is divided into two cases. We first prove the general case where $\mathscr{G}_{L}$ has more than one node, namely, $K \geq2$. Since $\mathscr{G}_{L}$ is structurally balanced, we have ${{G}_{L}}{\mathscr{L}_{L}}{{G}_{L}}=\mathscr{M}\left({\mathscr{L}_{L}}\right)$, where ${G_{L}}$ is as stated previously. Then, using $\mathscr{M}\left({\mathscr{L}_{L}}\right)1_{K}= 0$ and $G_{L} = {G_{L}^{-1}}$, we can get ${\mathscr{L}_{L}}{G_{L}}{1_K} = 0$. Noting that $\mathscr{G}_{L}$ is strongly connected, we know from Lemma \ref{lemma2} that there is a positive vector ${\nu_{L}}= {\left[ {{\nu_1}, \ldots ,{\nu_K}} \right]^{\rm{T}}} \in {\mathbb{R}^K}$ such that ${{\nu_{L}^{\rm{T}}}}{1_K} = 1$ and ${{\nu_{L}^{\rm{T}}}}{\mathscr{M}}\left( {{\mathscr{L}_{L}}} \right) = {0}$. Thus, we have ${{\nu_{L}^{\rm{T}}}}{{G}_{L}}{\mathscr{L}_{L}}{{G}_{L}} = {0}$. Since $ {{G}_{L}} = {{G}_{L}^{-1}}$, we can obtain ${{\nu_{L}^{\rm{T}}}}{{G}_{L}}{\mathscr{L}_{L}} = {0}$. This, together with \eqref{eq14}, implies ${\nu_{L}^{\rm{T}}}{G_{L}}{E}_{L}= 0$. Denote ${\bar{{\mathscr{L}_{L}}}} = {{\left( {{\it{\Xi}}_{L}{\mathscr{L}_{L}} + \mathscr{L}_{L}^{\rm{T}}{{\it{\Xi}}_{L}}} \right)} \mathord{\left/
 {\vphantom {{\left( {{\it{\Xi}}_L^{\rm{T}}{\mathscr{L}_L} + \mathscr{L}_L^{\rm{T}}{{\it{\Xi}}_{L}}} \right)} 2}} \right.\kern-\nulldelimiterspace} 2}$, where ${{\it{\Xi}}_{L}}= {\rm{diag}}\{ {\nu_1}, \ldots ,{\nu_K}\}$.
By Lemma \ref{lemma3}, we then have ${{E}_{L}^{\rm{T}}}{{\bar{\mathscr{L}}}_{L}}{{E}_{L}}\ge a\left({\mathscr{L}}_{L}\right){{E}_{L}^{\rm{T}}}{{\it{\Xi}}_{L}}{{E}_{L}}$,
where $a\left(\mathscr{L}_{L}\right)= \mathop {\min_{{{E}_{L}^{\rm{T}}}{G}_{L}{\nu_{L}} = 0,{{E}_{L}} \ne 0}} \frac{{{{E}_{L}^{\rm{T}}}{\bar{\mathscr{L}_{L}}}{{E}_{L}}}}{{{{E}_{L}^{\rm{T}}}{{\it{\Xi}}_{L}}{{E}_{L}}}} > 0$.
Let ${\it{\Xi}} = \rm{diag}\left\{{{{\it{\Xi}}_L},{\it{\varrho}} {{\it{\Xi}}_F}}\right\}$, where ${\it{\Xi}}_{F}$ is defined as in Lemma \ref{lemma7} and $0 < \varrho  < \frac{{2a\left( {{{\mathscr{L}}_L}} \right){\lambda_{\min}}\left( {{{\it{\Xi}}_L}} \right)}}{{{\lambda_{{\max}}}\left( {{\mathscr{L}}_{FL}^{\rm{T}}{{\it{\Xi}}_F}\left({{{\it{\Xi}}_{F}{\mathscr{L}_{F}} + {\mathscr{L}_{F}^{\rm{T}}}{\it{\Xi}}_{F}}}\right)^{-1}{{\it{\Xi}}_F}{{\mathscr{L}}_{FL}}}\right)}}$. Define ${\it{\Phi}}  = \left[ {\begin{array}{*{20}{c}}
{2a\left( {{{\mathscr{L}}_L}} \right){{\it{\Xi}}_{L}}}&{\vartheta{\mathscr{L}}_{FL}^{\rm{T}}{{\it{\Xi}}_F}}\\
{\vartheta{{\it{\Xi}}_F}{{\mathscr{L}}_{FL}}}&{\varrho\left({{{\it{\Xi}}_{F}{\mathscr{L}_{F}} + {\mathscr{L}_{F}^{\rm{T}}}{\it{\Xi}}_{F}}}\right)}
\end{array}} \right]$. Due to $\varrho  < \frac{{2a\left( {{{\mathscr{L}}_L}} \right){\lambda_{{\min}}}\left( {{{\it{\Xi}}_L}} \right)}}{{{\lambda_{{\max}}}\left( {{\mathscr{L}}_{FL}^{\rm{T}}{\it{\Xi}_F}\left({{{\it{\Xi}}_{F}{\mathscr{L}_{F}} + {\mathscr{L}_{F}^{\rm{T}}}{\it{\Xi}}_{F}}}\right)^{-1}{\it{\Xi}_F}{{\mathscr{L}}_{FL}}}\right)}}$, we have $2a\left( {{{\mathscr{L}}_L}} \right){\it{\Xi}_L} - \varrho {{\mathscr{L}}_{FL}^{\rm{T}}{\it{\Xi}_F}\left({{{\it{\Xi}}_{F}{\mathscr{L}_{F}} + {\mathscr{L}_{F}^{\rm{T}}}{\it{\Xi}}_{F}}}\right)^{-1}{\it{\Xi}_F}{{\mathscr{L}}_{FL}}} \succ 0$, which, by Lemma \ref{lemma5}, implies ${\it{\Phi}}\succ 0$.

Consider the Lyapunov function candidate
$${\bar{V_2}} = {{E}^{\rm{T}}}{\it{\Xi}}{E}.$$
Using \eqref{eq17} and \eqref{eq18}, we can derive
\begin{equation*}
\begin{split}
{\dot{\bar V}_2}=& 2{{E}^{\rm{T}}}{\it{\Xi}}\dot{E} = 2 E_L^{\rm{T}} {{\it{\Xi}}_L} {{\dot{E}}_L} + 2\varrho E_F^{\rm{T}}{{\it{\Xi}}_F}{{\dot E}_F}\\
= &- 2\left( {{\rho_1}  + {\rho_2} \frac{{\dot \varphi\left(t, T_1\right) }}{\varphi\left(t, T_1\right) }} \right) E_L^{\rm{T}}{{\bar{\mathscr{L}}}_L}{{ E}_L}\\
&-2\left( {{\rho_1}  + {\rho_2} \frac{{\dot \varphi\left(t, T_1\right) }}{\varphi\left(t, T_1\right) }} \right)\varrho E_F^{\rm{T}}{{\it{\Xi}}_F}\left( {{{\mathscr{L}}_F}{E_F} + {{\mathscr{L}}_{FL}}{E_L}} \right)\\
\le & - 2\left( {{\rho_1}  + {\rho_2} \frac{{\dot \varphi\left(t, T_1\right) }}{\varphi\left(t, T_1\right) }} \right)a\left( {{{\mathscr{L}}_L}} \right) E_L^{\rm{T}}{{\it{\Xi}}_L}{{E}_L}\\
& -\left( {{\rho_1}  + {\rho_2} \frac{{\dot \varphi\left(t, T_1\right) }}{\varphi\left(t, T_1\right) }} \right)\varrho E_F^{\rm{T}}\left( {{{\it{\Xi}}_F}{{\mathscr{L}}_F} + {\mathscr{L}}_F^{\rm{T}}{\it{\Xi}_F}} \right){E_F}\\
&- \left( {{\rho_1}  + {\rho_2} \frac{{\dot \varphi\left(t, T_1\right) }}{\varphi\left(t, T_1\right) }} \right)\varrho\left( {E_F^{\rm{T}}{{\it{\Xi}}_F}{{\mathscr{L}}_{FL}}{E_L} + E_L^{\rm{T}}\mathscr{L}_{FL}^{\rm{T}}{{\it{\Xi}}_F}{E_F}} \right)\\
= & - \left( {{\rho_1}  + {\rho_2} \frac{{\dot \varphi\left(t, T_1\right) }}{\varphi\left(t, T_1\right) }} \right){{ E}^{\rm{T}}}{\it{\Phi}}{E}\\
\le & - {\lambda_{{\min}}}\left({\it{\Phi}}\right)\left( {{\rho_1}  + {\rho_2} \frac{{\dot \varphi\left(t, T_1\right) }}{\varphi\left(t, T_1\right) }} \right){{E}^{\rm{T}}}{E}\\
\le & - \frac{{{\lambda_{{\min}}}\left({\it{\Phi}}\right)}}{{{\lambda_{{\max}}}\left({\it{\Xi}}\right)}}\left( {{\rho_1}  + {\rho_2} \frac{{\dot\varphi\left(t, T_1\right)}}{\varphi\left(t, T_1\right)}} \right){{\bar V}_2}\\
=&  - {\rho_1} \frac{{{\lambda_{{\min}}}\left({\it{\Phi}}\right)}}{{{\lambda_{{\max}}}\left({\it{\Xi}} \right)}}{{\bar V}_2} - {\rho_2} \frac{{{\lambda_{{\min}}}\left({\it{\Phi}}\right)}}{{{\lambda_{{\max}}}\left({\it{\Xi}}\right)}}\frac{{\dot\varphi\left(t, T_1\right)}}{\varphi\left(t, T_1\right)}{{\bar V}_2}.
\end{split}
\end{equation*}
 It then follows from Lemma \ref{lemma6} that
\begin{equation*}
{\bar V_2}\left( t \right)\left\{ \begin{array}{l}
 \le {\varphi\left(t, T_1\right) ^{ - {\rho_2} \frac{{{\lambda_{{\min}}}\left({\it{\Phi}}\right)}}{{{\lambda_{{\max}}}\left({\it{\Xi}}\right)}}}}{\exp ^{ - {\rho_1} \frac{{{\lambda_{{\min}}}\left({\it{\Phi}}\right)}}{{{\lambda_{{\max}}}\left({\it{\Xi}}\right)}}t}}{{\bar V}_2}\left( 0 \right),{\rm{ }}t \in \left[ {0,T_1} \right)\\
 \equiv 0,  \;\;\;\;\;\;\;\;\;\;\;\;\;\;\;\;\;\;\;\;\;\;\;\;\;\;\;\;\;\;\; t \in \left[ {T_1, + \infty } \right).
\end{array} \right.
\end{equation*}
 This implies that
$\left\|{E}\right\|^2 \le \frac{{{\lambda_{{\max}}}\left({\it{\Xi}}\right)}}{{{\lambda_{{\min}}}\left({\it{\Xi}}  \right)}}{\varphi\left(t, T_1\right) ^{-{\rho_2} \frac{{{\lambda_{{\min}}}\left({\it{\Phi}}\right)}}{{{\lambda_{{\max}}}\left({\it{\Xi}}\right)}}}}{\exp ^{-{\rho_1} \frac{{{\lambda_{{\min}}}\left({\it{\Phi}}\right)}}{{{\lambda_{{\max}}}\left({\it{\Xi}}\right)}}t}}\left\| {{E}\left( 0 \right)} \right\|^2$ on $\left[{0,T_1} \right)$, and ${E} \equiv 0$ on $\left[ {T_1, + \infty }\right)$. Noting that ${\lim_{t \to {T_1^-}}}{\varphi\left(t, T_1\right)^{-{\rho_2}\frac{{{\lambda_{{\min}}}\left({\it{\Phi}}\right)}}{{{\lambda_{{\max}}}\left( {\it{\Xi}}\right)}}}} = 0$, we can derive
${\lim_{t\to {T_1^-}}}{\left\|{E} \right\|^2} = 0$, which yields ${\lim_{t\to {T_1^-}}}{{E}} = 0$. Thus, \eqref{eq12} holds. This implies that
\begin{equation}\label{eq19}
\left\{ \begin{array}{l}
{\lim_{t \to T_1}}{{E}_L} = 0\\
{{E}_L} = 0,\forall t \ge T_1,
\end{array} \right.
\end{equation}
and
\begin{equation}\label{eq20}
\left\{ \begin{array}{l}
{\lim _{t \to T_1}}{E_F} = 0\\
{E_F} = 0,\forall t \ge T_1.
\end{array} \right.
\end{equation}

Since ${\mathscr{G}}_{L}$ is structurally balanced, according to Lemma \ref{lemma2}, we have that the null space of
${{\mathscr{L}}_L}$ is span by ${G_L}{1_K}$ . Consequently, we can derive from \eqref{eq14} and \eqref{eq19} that
\begin{equation}\label{eq21}
\left\{ \begin{array}{l}
{\lim_{t \to {T_1}}}{X_L} = {G_L}{1_K}c\\
{X_L} = {G_L}{1_K}c, \forall t \ge T_1,
\end{array} \right.
\end{equation}
 with $c \in \mathbb{R}$. Therefore, for any $k\in \mathcal{L}$, there hold ${\lim_{t \to {T_1}}} \left| {{x_k}} \right| = \left| c \right|$ and $\left| {{x_k}} \right| = \left| c \right|,\forall t \ge T_1$.
Furthermore, since ${\mathscr{L}_{F}}$ is nonsingular by Lemma \ref{lemma7}, it follows from \eqref{eq15} and \eqref{eq20} that
\begin{equation}\label{eq22}
\left\{ \begin{array}{l}
{\lim _{t \to T_1}}{X_F} = -{\mathscr{L}_{F}^{-1}}{\mathscr{L}_{FL}}{\lim _{t \to T_1}}{X_L}\\
{X_F} =  -{\mathscr{L}_{F}^{-1}}{\mathscr{L}_{FL}}{X_L},\forall t \ge T_1,
\end{array} \right.
\end{equation}
which, together with \eqref{eq21}, yields
\begin{equation}\label{eq23}
\left\{ \begin{array}{l}
{\lim _{t \to T_1}}{X_F} = -{\mathscr{L}_{F}^{-1}}{\mathscr{L}_{FL}}{G_L}{1_K}c\\
{X_F} =  -{\mathscr{L}_{F}^{-1}}{\mathscr{L}_{FL}}{G_L}{1_K}c,\forall t \ge T_1.
\end{array} \right.
\end{equation}
This, together with Lemma \ref{lemma8}, implies that ${\lim_{t \to {T_1}}}\left| {{x_l}} \right| = \left| {{{\lim }_{t \to {T_1}}}{x_l}} \right| \le \left| c \right|$ and $\left| {{x_l}} \right| \le \left| c \right|,\forall t \ge T_1$ for $\forall l\in \mathcal{F}$.

Now we consider the trivial case where $\mathscr{G}_{L}$ has exactly one node, namely, $K=1$. In this case, we have ${\mathscr{L}_{L}}= 0$. Then it follows from \eqref{eq16} that ${{\dot X}_{L}}= 0$, which implies ${X_{L}}\equiv X_{L}\left(0\right), \forall t {\ge 0}$. Noting ${\mathscr{L}_{L}}= 0$ and using \eqref{eq14}, we have ${E_{L}}= 0$. This, together with \eqref{eq18}, yields
\begin{equation}\label{eq24}
{{\dot E}_{F}} = - \left( {{\rho_1}  + {\rho_2} \frac{{\dot \varphi\left(t, T_1\right) }}{\varphi\left(t, T_1\right) }} \right){\mathscr{L}_{F}}{E_{F}}.
\end{equation}
Consider the following Lyapunov function:
$${\bar{V_2}^\ast} = {{E}^{\rm{T}}_{F}}{{\it{\Xi}}_{F}}{{E}_{F}},$$
where ${{\it{\Xi}}_{F}}$ is given in Lemma \ref{lemma7}. Using \eqref{eq24}, we can obtain
\begin{equation*}
\begin{split}
{\dot{\bar{V_2}}}^* &= 2E_F^{\rm{T}}{\it{\Xi}_F}{{\dot E}_F}\\
& =  - 2\left( {{\rho_1}  + {\rho_2} \frac{{\dot \varphi\left(t, T_1\right) }}{\varphi\left(t, T_1\right) }} \right)E_F^{\rm{T}}{\it{\Xi}_F}{\mathscr{L}_F}{E_F}\\
& =  - \left( {{\rho_1}  + {\rho_2} \frac{{\dot \varphi\left(t, T_1\right) }}{\varphi\left(t, T_1\right) }} \right)E_F^{\rm{T}}\left( {{\it{\Xi}_F}{\mathscr{L}_F} + \mathscr{L}_F^{\rm{T}}{\it{\Xi}_F}} \right){E_F}\\
& \le  - {\lambda_{{\min}}}\left( \Upsilon  \right)\left( {{\rho_1}  + {\rho_2} \frac{{\dot \varphi\left(t, T_1\right) }}{\varphi\left(t, T_1\right) }} \right)E_F^{\rm{T}}{E_F}\\
& \le  -{\rho_1} \frac{{{\lambda_{{\min}}}\left(\Upsilon\right)}}{{{\lambda_{\max}}\left( {{\it{\Xi}_F}}\right)}} \bar V_2^ * - {\rho_2}\frac{{{\lambda_{\min}}\left(\Upsilon \right)}}{{{\lambda_{\max}}\left( {{\it{\Xi}_F}} \right)}}\frac{{\dot\varphi\left(t, T_1\right)}}{\varphi\left(t, T_1\right)}\bar V_2^*,
\end{split}
\end{equation*}
where $\Upsilon = {\it{\Xi}_F}{\mathscr{L}_F} + \mathscr{L}_F^{\rm{T}}{\it{\Xi}_F}$. This, together with Lemma \ref{lemma6}, implies that ${\bar V_2}^* \left( t \right)\le {\varphi\left(t, T_1\right) ^{ - {\rho_2} \frac{{{\lambda_{{\min}}}\left({\Upsilon}\right)}}{{{\lambda_{{\max}}}\left({\it{\Xi}_F}\right)}}}}{\exp ^{-{\rho_1} \frac{{{\lambda_{{\min}}}\left({\Upsilon}\right)}}{{{\lambda_{{\max}}}\left({\it{\Xi}_F}\right)}}t}}{{\bar V}_2}^*\left( 0 \right)$ for all $ t \in \left[ {0,T_1} \right)$ and ${\bar V_2}^* \left( t \right)\equiv 0$, for all $t \in \left[ {T_1, + \infty } \right)$.
Then it follows that
$\left\|{{{E}_{F}}}\right\|^2 \le \frac{{{\lambda_{{\max}}}\left( {\it{\Xi_F}} \right)}}{{{\lambda_{{\min}}}\left( {\it{\Xi_F}}  \right)}}{\varphi\left(t, T_1\right) ^{-{\rho_2} \frac{{{\lambda_{{\min}}}\left( \Upsilon  \right)}}{{{\lambda_{{\max} }}\left( {{\it{\Xi}_F}} \right)}}}}{\exp ^{-{\rho_1} \frac{{{\lambda_{{\min}}}\left( \Upsilon  \right)}}{{{\lambda_{{\max} }}\left( {{\it{\Xi}_F}} \right)}}t}}\left\| {{{{E}_{F}}}\left( 0 \right)} \right\|^2$ on $\left[{0,T_1} \right)$ and ${E_F} \equiv 0$ on $\left[ {T_1, + \infty }\right)$. Due to ${\lim_{t \to {T_1^-}}}{\varphi\left(t, T_1\right)^{-{\rho_2}\frac{{{\lambda _{{\min}}}\left( \Upsilon  \right)}}{{{\lambda_{{\max} }}\left( {{\it{\Xi}_F}} \right)}}}} = 0$, we have
${\lim_{t\to {T_1^-}}}{\left\|{{{E}_{F}}} \right\|^2} = 0$, which gives ${\lim_{t\to {T_1^-}}}{{{{E}_{F}}}} = 0$. Thus, \eqref{eq20} follows. Further, noting that ${\mathscr{L}_{F}}$ is nonsingular and  ${X_{L}}\equiv X_{L}\left(0\right)$, we have from \eqref{eq15} and \eqref{eq20} that
\begin{equation}\label{eq25}
\left\{ \begin{array}{l}
{\lim _{t \to T_1}}{X_F} =  - \mathscr{L}_F^{ - 1}{\mathscr{L}_{FL}}{X_L}\left( 0 \right),\\
{X_F} =  - \mathscr{L}_F^{ - 1}{\mathscr{L}_{FL}}{X_L}\left( 0 \right),\forall t \geq T_1,
\end{array} \right.
\end{equation}
which, by Lemma \ref{lemma8}, implies that ${\lim_{t \to {T_1}}}\left| {{x_l}} \right| = \left| {{{\lim }_{t \to {T}}}{x_l}} \right| \le \left|{X_L}\left( 0 \right)\right|$ and $\left| {{x_l}} \right| \le \left|{X_L}\left( 0 \right) \right|,\forall t \ge T_1$ for all $l\in \mathcal{F}$.

Based on the above proof, we can see that the two conditions in Definition \ref{defn3} hold for both cases. Therefore, the conclusion holds.

2) Since $\mathscr{G}_{L}$ is structurally unbalanced, we have $K\geq 2$. Then it follows from Lemma \ref{lemma3} that there is a positive diagonal matrix $\tilde{{\it{\Xi}}}_{L}$ such that $\tilde{{\mathscr{L}}_{L}} \triangleq {\tilde{{\it{\Xi}}}_{L}{\mathscr{L}_{L}} + {\mathscr{L}_{L}^{\rm{T}}}\tilde{{\it{\Xi}}}_{L}}\succ0$. Let $\tilde{\it{\Xi}} = {\rm{diag}}\left\{{{\tilde{{\it{\Xi}}}}_{L},\tilde{\it{\varrho}} {\it{\Xi}_F}}\right\}$, where ${\it{\Xi}}_{F}$ is as stated in Lemma \ref{lemma7} and $0 < \tilde \varrho  < \frac{{{\lambda_{{\min}}}\left( {{{\tilde {\mathscr{L}}}_L}} \right)}}{{{\lambda_{{\max}}}\left( {{\mathscr{L}}_{FL}^{\rm{T}}{\it{\Xi}_F}{{\left( {{\it{\Xi} _F}{{\mathscr{L}}_F} + {\mathscr{L}}_F^{\rm{T}}{\it{\Xi}_F}} \right)}^{-1}}{\it{\Xi}_F}{{\mathscr{L}}_{FL}}} \right)}}$. Choose the following Lyapunov function candidate:
$${\tilde{V_2}} = {E^{\rm{T}}}\tilde{\it{\Xi}}E.$$
Its time derivative is given by
\begin{equation}\label{eq26}
\begin{split}
{\dot{\tilde{V}}_2} &= 2{E^{\rm{T}}}{\tilde{\it{\Xi}}} \dot E\\
&= - \left({{\rho_1}  + {\rho_2} \frac{{\dot \varphi\left(t, T_1\right) }}{\varphi\left(t, T_1\right) }} \right){E^{\rm{T}}}\left( {{\tilde{\it{\Xi}}}\mathscr{L} + {\mathscr{L}^{\rm{T}}}{\tilde{\it{\Xi}}}} \right)E\\
&=  - \left( {{\rho_1}  + {\rho_2} \frac{{\dot \varphi\left(t, T_1\right) }}{\varphi\left(t, T_1\right)}} \right){E^{\rm{T}}}{\it{\Psi}} E,
\end{split}
\end{equation}
where ${\it{\Psi}}= \left[ {\begin{array}{*{20}{c}}
{{{\tilde{\mathscr{L}}}_L}}&{\tilde \varrho {\mathscr{L}}_{FL}^{\rm{T}}{{\it{\Xi}}_F}}\\
{\tilde{\varrho}{{\it{\Xi}}_F}{{\mathscr{L}}_{FL}}}&{\tilde{\varrho} \left({{{\it{\Xi}}_F}{{\mathscr{L}}_F} + {\mathscr{L}}_F^{\rm{T}}{{\it{\Xi}}_F}} \right)}
\end{array}} \right]$. Owing to $\tilde \varrho  < \frac{{{\lambda_{{\min}}}\left( {{{\tilde {\mathscr{L}}}_L}} \right)}}{{{\lambda_{{\max}}}\left( {{\mathscr{L}}_{FL}^{\rm{T}}{\it{\Xi}_F}{{\left( {{\it{\Xi}_F}{{\mathscr{L}}_F} + {\mathscr{L}}_F^{\rm{T}}{\it{\Xi}_F}} \right)}^{-1}}{\it{\Xi}_F}{{\mathscr{L}}_{FL}}} \right)}}$, we can obtain
${{\tilde{\mathscr{L}}}_L} - {\tilde{\varrho}}{\mathscr{L}}_{FL}^{\rm{T}}{\it{\Xi}_F}{\left( {{\it{\Xi} _F}{{\mathscr{L}}_F} + {\mathscr{L}}_F^{\rm{T}}{\it{\Xi}_F}} \right)^{-1}}{\it{\Xi}_F}{{\mathscr{L}}_{FL}} \succ 0$, which, by Lemma \ref{lemma5}, implies ${\it{\Psi}}\succ 0$. It then can be derived from \eqref{eq26} that
\begin{equation*}
\begin{split}
{\dot{\tilde{V}}_2} &\le  - {\lambda_{{\min}}}\left({\it{\Psi}}\right)\left( {{\rho_1} + {\rho_2} \frac{{\dot \varphi\left(t, T_1\right) }}{\varphi\left(t, T_1\right) }} \right){E^{\rm{T}}}E\\
&\le  - {\rho_1} \frac{{{\lambda_{{\min}}}\left({\it{\Psi}}\right)}}{{{\lambda_{{\max}}}\left( {{\tilde{\it{\Xi}}}} \right)}} {{\tilde V}_2} - {\rho_2}\frac{{{\lambda_{{\min}}}\left({\it{\Psi}}\right)}}{{{\lambda_{{\max}}}\left( {\tilde{\it{\Xi}}} \right)}} \frac{{\dot \varphi\left(t, T_1\right)}}{\varphi\left(t, T_1\right)}{{\tilde V}_2},
\end{split}
\end{equation*}
which, together with Lemma \ref{lemma6}, gives
\begin{equation*}
{\tilde V_2}\left( t \right)\left\{ \begin{array}{l}
 \le {\varphi\left(t, T_1\right) ^{-{\rho_2} \frac{{{\lambda_{{\min}}}\left({\it{\Psi}}\right)}}{{{\lambda_{{\max}}}\left({\it{\tilde{\Xi}}}\right)}}}}{\exp ^{-{\rho_1}{\frac{{{\lambda_{{\min}}}\left({\it{\Psi}}\right)}}{{{\lambda_{{\max}}}\left({\it{\tilde{\Xi}}}\right)}}}t}}{{\tilde V}_2}\left( 0 \right),{\rm{ }}t \in \left[ {0,T_1} \right)\\
 \equiv 0,  \;\;\;\;\;\;\;\;\;\;\;\;\;\;\;\;\;\;\;\;\;\;\;\;\;\;\;\;\;\;\; t \in \left[ {T_1, + \infty } \right).
\end{array} \right.
\end{equation*}
This implies that $\left\|E\right\|^2 \le \frac{{{\lambda_{{\max}}}\left({\it{\tilde{\Xi}}} \right)}}{{{\lambda_{{\min}}}\left({\it{\tilde{\Xi}}}\right)}}{\varphi\left(t, T_1\right)^{-{\rho_2} \frac{{{\lambda_{{\min}}}\left({\it{\Psi}}\right)}}{{{\lambda_{{\max}}}\left({\it{\tilde{\Xi}}} \right)}}}}{\exp ^{-{\rho_1} \frac{{{\lambda_{{\min}}}\left({\it{\Psi}}\right)}}{{{\lambda_{{\max}}}\left( {\it{\tilde{\Xi}}}\right)}}t}}\left\| {E\left( 0 \right)} \right\|^2$ on $\left[{0,T_1} \right)$, and $E \equiv 0$ on $\left[ {T_1, + \infty }\right)$. In view of ${\lim_{t \to {T_1^-}}}{\varphi\left(t, T_1\right)^{-{\rho_2}\frac{{{\lambda_{{\min}}}\left({\it{\Psi}} \right)}}{{{\lambda_{{\max}}}\left( {\it{\tilde{\Xi}}}\right)}}}} = 0$, we can obtain
${\lim_{t\to {T_1^-}}}{\left\|E \right\|^2} = 0$, which implies ${\lim_{t\to {T_1^-}}}{E} = 0$. Thus, we have \eqref{eq12}. Therefore, \eqref{eq19} and \eqref{eq20} hold. Since $\mathscr{G}_{L}$ is strongly connected and structurally unbalanced, we know from Lemma \ref{lemma2} that $\mathscr{L}_{L}$ is nonsingular. This, together with \eqref{eq14} and \eqref{eq19}, clearly gives that ${\lim_{t \to {T_1}}}{X_L} = 0$ and ${X_L} = 0, \forall t \geq T_1$. By noting that ${\mathscr{L}_{F}}$ is nonsingular and using \eqref{eq20}, we can deduce from \eqref{eq15} that ${\lim_{t \to {T_1}}}{X_F} = 0$ and ${X_F} = 0, \forall t \geq T_1$. Hence, it follows that ${\lim_{t \to {T_1}}}{{X}}= 0$ and $ {{X}}= 0, \forall t \ge T_1$. This completes the proof.
\end{proof}

For particular cases where $\mathscr{G}$ is a structurally balanced quasi-strongly connected signed digraph or a traditional quasi-strongly connected digraph with positive edge weights, the following corollaries of the above theorem can be obtained.

\begin{corollary}\label{corol3}
Consider the nominal CAN \eqref{eq3} under the control protocol \eqref{eq4}. Let $\mathscr{G}$ be quasi-strongly connected and structurally balanced. Then, the nominal CAN \eqref{eq3} achieves prescribed-time bipartite consensus in the pre-specified finite time $T_1$.
\end{corollary}

\begin{proof}
Since any subgraph of a structurally balanced signed digraph is structurally balanced, one knows that $\mathscr{G}_{L}$ is structurally balanced by the structural balance of $\mathscr{G}$. When $K\geq 2$, it follows from the above proof of the first conclusion of Theorem \ref{thm2} that \eqref{eq21} and \eqref{eq23} hold.  Since $\mathscr{G}$ is quasi-strongly connected and structurally balanced, one has from Lemma \ref{lemma8} that $\zeta = {\mathbb{G}_{F}}{1_{N-K}}$. This, together with \eqref{eq23}, leads to
\begin{equation}\label{eq27}
\left\{ \begin{array}{l}
{\lim _{t \to T_1}}{X_F} = {\mathbb{G}_{F}}{1_{N-K}} c\\
{X_F} = {\mathbb{G}_{F}}{1_{N-K}} c,\forall t \ge T_1.
\end{array} \right.
\end{equation}
Further, using \eqref{eq21} and \eqref{eq27}, we can obtain that
${\lim_{t \to {T_1}}}\left| {{x_k}} \right| = \left| c \right|$ and $\left| {{x_k}} \right|= \left| c \right|,\forall t \ge T_1$ hold for all $k \in {\mathscr{I}_N}$.

When $K = 1$, we have from the above proof of the first conclusion of Theorem \ref{thm2} that ${X_{L}}\equiv X_{L}\left(0\right),\forall t \ge 0$ and \eqref{eq25} hold. Using Lemma \ref{lemma8}, we can obtain $- {\mathscr{L}}_F^{ - 1}{{\mathscr{L}}_{FL}}={\mathbb{G}_F}{1_{N - 1}}{\mathbb{G}_L}$. This, together with \eqref{eq25}, implies that
${\lim_{t \to {T_1}}}\left| {{x_l}} \right| = \left|{X_L}\left( 0 \right)\right|$ and $\left| {{x_l}} \right| = \left|{X_L}\left( 0 \right) \right|,\forall t \ge T_1$ hold for all $l\in \mathcal{F}$. Therefore, ${\lim_{t \to {T_1}}}\left| {{x_k}} \right| = \left|{x_1}\left( 0 \right)\right|$ and $\left| {{x_k}} \right| = \left|{x_1}\left( 0 \right) \right|,\forall t \ge T_1$ hold for all $k \in {\mathscr{I}_N}$.

Combining the above two aspects, we can see that prescribed-time bipartite consensus of the nominal CAN \eqref{eq3} is reached in the pre-specified finite time $T_1$. The proof is done.
\end{proof}

\begin{corollary}\label{corol4}
Consider the nominal CAN \eqref{eq3} under the control protocol \eqref{eq4}. Let $\mathscr{G}$ be quasi-strongly connected with its edge weights all positive. Then, the nominal CAN \eqref{eq3} achieves prescribed-time consensus in the pre-specified finite time $T_1$.
\end{corollary}

\begin{proof}
Since the edge weights of $\mathscr{G}$ are all positive, it is clear that $\mathscr{L}= {\mathscr{M}}\left( \mathscr{L} \right)$ and $\mathscr{G}_{L}$ is structurally balanced. Noting that $\mathscr{L}= {\mathscr{M}}\left( \mathscr{L} \right)$ and ${\mathscr{M}}\left( \mathscr{L} \right){1_N} = 0$, we have from \eqref{eq13} that ${\mathscr{L}_{FL}}{1_K} + {\mathscr{L}_F}{1_{N - K}} = 0$, which, by the first conclusion of Lemma \ref{lemma7}, implies
\begin{equation}\label{eq28}
\mathscr{L}_F^{ - 1}{\mathscr{L}_{FL}}{1_K} = -{1_{N - K}}.
\end{equation}

When $K\geq 2$, it follows from the above proof of the first conclusion of Theorem \ref{thm2} that \eqref{eq21} and \eqref{eq23} hold. Recalling that $\mathscr{G}_{L}$ is structurally balanced, we know that ${G_{L}}= I_{K}$ or ${G_{L}}= -{I_{K}}$. This, together with \eqref{eq21}, \eqref{eq23}, and \eqref{eq28}, implies that
$$\left\{ \begin{array}{l}
{\lim _{t \to T_1}}{X} ={1_N}c\\
{X} = {1_N}c,\forall t \ge T_1,
\end{array} \right.
{\rm{or}} \left\{ \begin{array}{l}
{\lim _{t \to T_1}}{X} =-{1_N}c\\
{X} = -{1_N}c,\forall t \ge T_1
\end{array} \right.$$
holds.

When $K = 1$, there hold ${X_{L}}\equiv X_{L}\left(0\right),\forall t \ge 0$ and \eqref{eq25} from the above proof of the first conclusion of Theorem \ref{thm2}. Using \eqref{eq28}, we have $-\mathscr{L}_F^{ - 1}{\mathscr{L}_{FL}} ={1_{N - 1}}$. This, together with \eqref{eq25}, leads to
\begin{equation*}
\left\{ \begin{array}{l}
{\lim _{t \to T_1}}{X_F} = {1_{N - 1}}{X_L}\left( 0 \right),\\
{X_F} = {1_{N - 1}}{X_L}\left( 0 \right),\forall t \geq T_1.
\end{array} \right.
\end{equation*}
Therefore, it follows that
\begin{equation*}
\left\{ \begin{array}{l}
{\lim _{t \to T_1}}{X} = {1_{N}}{X_L}\left( 0 \right),\\
{X_F} = {1_{N}}{X}\left( 0 \right),\forall t \geq T_1.
\end{array} \right.
\end{equation*}
Based on the above proof, we can see that the prescribed-time consensus of the nominal CAN \eqref{eq3} is achieved in the pre-specified finite time $T_1$. This completes the proof.
\end{proof}

\begin{remark}
 Corollary \ref{corol4} holds for arbitrary quasi-strongly connected traditional unsigned digraph $\mathscr{G}$, regardless of whether it
contains one root or more than one root. In addition, Corollary \ref{corol4} contains as a special case the classical prescribed-time consensus tracking result in Reference \cite{wang2018prescribed} for single-integrator traditional networks under quasi-strongly connected unsigned digraph having exactly one root.
\end{remark}

In what follows, we further consider the scenario where $\mathscr{G}$ is weakly connected. Notice that arbitrary weakly connected signed digraph $\mathscr{G}$ has more than one CSC. Without losing generality, suppose $\mathscr{G}$
has totally $m$ CSCs, say ${\mathscr{G}_1},{\mathscr{G}_2}, \ldots, {\mathscr{G}_m}$, and the nodes of ${\mathscr{G}_k}$ are given by ${\mathscr{V}_k} = \left\{ {v_{{\sum\nolimits_{\iota = 0}^{k-1}{N_\iota}} + 1}, v_{{\sum\nolimits_{\iota = 0}^{k-1}{N_\iota}} + 2}, \ldots, v_{\sum\nolimits_{\iota = 0}^k {{N_\iota}}}} \right\}, 1 \le k \le m$, where ${N_0} = 0$ and $\sum\nolimits_{\iota = 0}^m {{N_\iota}}= K$. Then it is easy to see that the leaders could be divided into $m$ $\left( {1 < m \le K} \right)$ separate subgroups, in which each subgroup corresponds to a CSC of $\mathscr{G}$. Thus, the Laplacian matrix $\mathscr{L}$ could be partitioned as
\begin{equation}\label{eq29}
\mathscr{L} = \left[ {\begin{array}{*{20}{c}}
{{\mathscr{L}_L}}&{{0}}\\
{{\mathscr{L}_{FL}}}&{{\mathscr{L}_F}}
\end{array}} \right]
\end{equation}
with ${\mathscr{L}_L} = {\rm{diag}}\left\{ {{\mathscr{L}_{L1}},{\mathscr{L}_{L2}}, \ldots ,{\mathscr{L}_{Lm}}} \right\}$ and ${\mathscr{L}_{FL}} = \left[ {{\mathscr{L}_{FL1}},{\mathscr{L}_{FL2}},\ldots, {\mathscr{L}_{FLm}}} \right]$, where $\mathscr{L}_F \in {\mathbb{R}^{\left( {N - K} \right) \times \left( {N - K} \right)}}$,  $\mathscr{L}_{Lk}\in {\mathbb{R}^{{N_k} \times {N_k}}}$, and ${\mathscr{L}_{FLk}} \in {\mathbb{R}^{\left( {N - K} \right) \times {N_k}}}$, ${1 \le k \le m}$. Denote ${X_{Lk}} = {\left[{x_{{\sum\nolimits_{\iota = 0}^{k-1}{N_\iota}} + 1}^{\rm{T}}, \ldots, x_{\sum\nolimits_{\iota = 0}^k {{N_\iota}} }^{\rm{T}}} \right]^{\rm{T}}}$ and ${E_{Lk}} = {\left[{e_{{\sum\nolimits_{\iota = 0}^{k-1}{N_\iota}} + 1}^{\rm{T}}, \ldots, e_{\sum\nolimits_{\iota= 0}^k {{N_\iota}} }^{\rm{T}}} \right]^{\rm{T}}}$, $1 \le k \le m$. Obviously, ${X_L} = {\left[ {X_{L1}^{\rm{T}}, \cdots, X_{Lm}^{\rm{T}}} \right]^{\rm{T}}}$ and ${E_L} = {\left[ {E_{L1}^{\rm{T}}, \cdots, E_{Lm}^{\rm{T}}} \right]^{\rm{T}}}$. By using \eqref{eq5} and \eqref{eq29}, one can get
\begin{align}
&{E_{Lk}} =  - {\mathscr{L}_{Lk}}{X_{Lk}},1 \le k \le m \label{eq30}\\
&{E_F} =  - {\mathscr{L}_F}{X_F} - \sum\limits_{k = 1}^m {{\mathscr{L}_{FLk}}{X_{Lk}}},\label{eq31}
\end{align}
where ${X_F}$ and ${E_F}$ are as previously defined. Substituting \eqref{eq29} into \eqref{eq8}, we have
\begin{align}
&{{\dot E}_{Lk}} =  - \left( {{\rho_1}  + {\rho_2} \frac{{\dot \varphi\left(t, T_1\right) }}{\varphi\left(t, T_1\right) }} \right){\mathscr{L}_{Lk}}{E_{Lk}},1 \le k \le m \label{eq32}\\
&{{\dot E}_F} =  - \left( {{\rho_1}  + {\rho_2} \frac{{\dot \varphi\left(t, T_1\right) }}{\varphi\left(t, T_1\right)}}\right)\left({\mathscr{L}_F}{E_F}+ \sum\limits_{k = 1}^m {{\mathscr{L}_{FLk}}{E_{Lk}}}\right). \label{eq33}
\end{align}

To proceed, we need the following lemma.

\begin{lemma}\cite{zhu2020observer}\label{lemma9}
Let $\mathscr{G}$ be weakly connected with its Laplacian matrix $\mathscr{L}$ given by \eqref{eq29}.
\begin{enumerate}
  \item ${\mathscr{L}_F}$ is a nonsingular {\it{H}}-matrix. In addition, there exists some positive diagonal matrix ${{\it{\Theta}}_{F}}\in {\mathbb{R}^{\left( {N - K} \right) \times \left( {N - K} \right)}}$ such that ${{{\it{\Theta}}_{F}}{\mathscr{L}_F} + {\mathscr{L}_F^{\rm{T}}}{{\it{\Theta}}_{F}}}\succ0$.
  \item Define $\varpi =\left[ {{\varpi_1},\ldots, {\varpi_m}} \right]$, where ${\varpi_k} = -{\mathscr{L}_F ^{-1}}{\mathscr{L}_{FLk}}{\mathbf{G}_k}{1_{{N_k}}}$ with ${\mathbf{G}_k} \in {\mathbb{G}_{{N_k}}}$, then
    $\sum\nolimits_{k = 1}^m {\left|{{\varpi_{jk}}}\right|}\le 1$ for all $1\leq j\leq N - K$.
\end{enumerate}
\end{lemma}

The following theorem shows that for the CAN \eqref{eq3} under any weakly connected signed graph $\mathscr{G}$, the protocol \eqref{eq4} can guarantee the CAN to reach prescribed-time bipartite containment (resp., stability) in the pre-specified finite time $T_1$, if at least one CSC of $\mathscr{G}$ is structurally balanced (resp., all the CSCs of $\mathscr{G}$ are structurally unbalanced).

\begin{theorem}\label{thm3}
Consider the nominal CAN \eqref{eq3} under the control protocol \eqref{eq4}, and let $\mathscr{G}$ be weakly connected. Then, the nominal CAN \eqref{eq3} achieves
\begin{enumerate}
  \item  prescribed-time bipartite containment in the pre-specified finite time $T_1$, if $\mathscr{G}$ has at least one structurally balanced CSC.
  \item prescribed-time stability in the pre-specified finite time $T_1$, if all the CSCs of $\mathscr{G}$ are structurally unbalanced.
\end{enumerate}
\end{theorem}

\begin{proof}
1) Since $\mathscr{G}$ has at least one structurally balanced CSC, without losing generality, suppose ${\mathscr{G}_k}, {1 \le k \le d} $ are structurally balanced, whereas ${\mathscr{G}_k}, {d+1 \le k \le m} $ are structurally unbalanced. Noting that a graph composed of one node is strongly connected and structurally balanced, we further assume that ${\mathscr{G}_k}$ has one node, ${1 \le k \le h} $, while ${\mathscr{G}_k}$ has at least two nodes, ${h+1 \le k \le d} $, that
is, ${N_k}=1$ for ${1 \le k \le h}$ and ${N_k}\geq 2$ for ${h+1 \le k \le d}$. It thus follows that ${\mathscr{L}_{Lk}}=0$, ${1 \le k \le h}$. This, together with \eqref{eq30}, yields
\begin{equation}\label{eq34}
{E_{Lk}} = 0, {1 \le k \le h}.
\end{equation}
By Lemma \ref{lemma1}, we have that there exits ${G_k} \in {\mathbb{G}_{{N_k}}}$ such that ${G_k}{\mathscr{L}_{Lk}}{G_k} = \mathscr{M}\left( {{\mathscr{L}_{Lk}}} \right),{h+1 \le k \le d}$. From Lemma \ref{lemma2}, we know that there are positive vectors ${\eta_k} = {\left[ {{\eta_{k1}},\ldots,{\eta_{k{N_k}}}} \right]^{\rm{T}}}$ satisfying $\eta_k^{\rm{T}}{1_{{N_k}}} = 1$ and $\eta_k^{\rm{T}}\mathscr{M}\left({\mathscr{L}_{Lk}}\right)= 0, {h+1 \le k \le d}$. Noting that $G_k$ is nonsingular, we can derive ${\eta_k^{\rm{T}}}{G_k}{\mathscr{L}_{Lk}}=0, {h+1 \le k \le d}$,
which, together with \eqref{eq30}, implies ${\eta_k^{\rm{T}}}{G_k}{{\rm\it{{E}}}_{Lk}}=0, {h+1 \le k \le d}$. Denote ${\bar{\mathscr{L}}_{Lk}} = {{\left( {{H_k}{\mathscr{L}_{Lk}} + \mathscr{L}_{Lk}^{\rm{T}}{H_k}} \right)} \mathord{\left/{\vphantom {{\left( {{H_k}{\mathscr{L}_{Lk}} + \mathscr{L}_{Lk}^{\rm{T}}{H_k}} \right)} 2}} \right.
 \kern-\nulldelimiterspace} 2}, {h+1 \le k \le d}$, where ${H_k}={\rm{diag}}\left\{\eta_{k1}, \ldots,\eta_{k{N_k}}\right\}$. Then, by using Lemma \ref{lemma3}, we can get
\begin{equation}\label{eq35}
E_{Lk}^{\rm{T}}{{\bar{\mathscr{L}}}_{Lk}}{E_{Lk}} \ge a\left( {{\mathscr{L}_{Lk}}} \right)E_{Lk}^{\rm{T}}{H_k}{E_{Lk}},{h+1 \le k \le d},
\end{equation}
where $a\left({{\mathscr{L}_{Lk}}}\right)= \mathop {\min_{{{E}_{Lk}^{\rm{T}}}{G}_{Lk}{\eta_{Lk}} = 0,{{E}_{Lk}} \ne 0}} \frac{{{{E}_{Lk}^{\rm{T}}}{\bar{\mathscr{L}_{Lk}}}{{E}_{Lk}}}}{{{{E}_{Lk}^{\rm{T}}}{H_{k}}{{E}_{Lk}}}} > 0$. Because ${\mathscr{G}_k}, {d+1\le k \le m} $ are structurally unbalanced, it follows from Lemma \ref{lemma3} that there exist positive diagonal matrices ${H_k}$ so that $\bar{\mathscr{L}}_{Lk} \buildrel \Delta \over = {{H_k}{\mathscr{L}_{{Lk}}} + \mathscr{L}_{{Lk}}^{\rm{T}}{H_k}}\succ0, {d+1\le k\le m}$, which implies
\begin{equation}\label{eq36}
{E_{Lk}^{\rm{T}}}{\bar{\mathscr{L}}_{Lk}}{E_{Lk}}\ge {{\lambda_{{\min}}}\left({\bar{\mathscr{L}}_{Lk}}\right)}{E_{Lk}^{\rm{T}}}{E_{Lk}}, d+1\le k\le m.
\end{equation}
From Lemma \ref{lemma9}, we have that there is a positive diagonal matrix ${H_{F}}$ such that $\bar{\mathscr{L}}_{F} \buildrel \Delta \over = {{H_{F}}{\mathscr{L}_{{F}}}+\mathscr{L}_{{F}}^{\rm{T}}{H_{F}}}\succ0$. Define ${\chi_1} = {{\min} _{h + 1 \le k \le d}}\left\{2{a\left( {{\mathscr{L}_{Lk}}} \right){\lambda_{{\min}}}\left({H_{k}}\right)} \right\}$ and ${\chi_2} = {\min_{d + 1 \le k \le m}}\left\{{{\lambda_{{\min}}}\left({{{\bar{\mathscr{L}}}_{Lk}}}\right)}\right\}$. Let $H ={\rm{diag}}\left\{ {{H_{h+1}},\ldots,{H_{m}},\bar{\rho}{H_{F}}}\right\}$, where $0 < \bar\rho < \frac{{{\min} \left\{{{\chi _1},{\chi_2}} \right\}}}{{\max \left( {\bar{\mathscr{L}}_{FL}^{\rm{T}}{H_F}\bar{\mathscr{L}}_F^{-1}{H_F}{\bar{\mathscr{L}}_{FL}}} \right)}}$ with ${ \bar{{\mathscr{L}}}_{FL}} = \left[ {{\mathscr{L}_{LF\left( {h + 1} \right)}}, \ldots ,{\mathscr{L}_{LFm}}} \right]$. Further, denote ${\bar{E}} = {\left[ {E_{h + 1}^{\rm{T}}, \ldots ,E_m^{\rm{T}},E_{F}^{\rm{T}}} \right]^{\rm{T}}} \buildrel \Delta \over = \left[ {{e_{h + 1}}, \ldots ,{e_N}} \right]^{\rm{T}}$ and $ {\it{Y}}= {\left( {H{\it{\Lambda}} + {{\it{\Lambda}}^{\rm{T}}}H}\right)}$, in which
$${\it{\Lambda}} = \left[ {\begin{array}{*{20}{c}}
{{\mathscr{L}_{{L{\left(h + 1\right)}}}}}& \cdots &0&0\\
 \vdots & \ddots & \vdots & \vdots \\
0& \cdots &{{\mathscr{L}_{{Lm}}}}&0\\
{{\mathscr{L}_{L{F{\left(h + 1\right)}}}}}& \cdots &{{\mathscr{L}_{LFm}}}&{{\mathscr{L}_F}}
\end{array}} \right].$$
It then follows from \eqref{eq32}-\eqref{eq34} that
$$\dot{\bar{E}} = - \left( {{\rho_1} + {\rho_2} \frac{{\dot \varphi\left(t, T_1\right)}}{\varphi\left(t, T_1\right)}} \right){\it{\Lambda}}{\bar E},$$
The Lyapunov function is selected as
$$\bar{V}_3 = {\bar{E}^{\rm{T}}}H{\bar{E}}.$$
The derivative of $\bar{V}_3$ w.r.t. $t$ satisfies
\begin{equation}\label{eq37}
\begin{split}
{{\bar V}_3} &= 2{{\bar E}^{\rm{T}}}H{\dot{\bar E}}\\
&=  - 2\left( {{\rho_1} + {\rho_2} \frac{{\dot \varphi\left(t, T_1\right)}}{\varphi\left(t, T_1\right) }} \right){{\bar E}^{\rm{T}}}H{\it{\Lambda}}{\bar E}\\
&=  - \left( {{\rho_1} + {\rho_2} \frac{{\dot \varphi\left(t, T_1\right) }}{\varphi\left(t, T_1\right) }} \right){{\bar E}^{\rm{T}}}{\it{Y}} {\bar E}\\
&\le  - \left( {{\rho_1} + {\rho_2} \frac{{\dot \varphi\left(t, T_1\right) }}{\varphi\left(t, T_1\right) }} \right){{\bar E}^{\rm{T}}}\Pi {\bar E},
\end{split}
\end{equation}
where $\Pi  = \left[ {\begin{array}{*{20}{c}}
\Gamma &{\bar \rho \bar{\mathscr{L}}_{FL}^{\rm{T}}{H_F}}\\
{\bar \rho {H_F}{{\bar{\mathscr{L}}}_{FL}}}&{\bar \rho {{\bar{\mathscr{L}}}_F}}
\end{array}} \right]$ with $\Gamma  = {\rm{diag}}\left\{{2a\left( {{\mathscr{L}_{L\left( {h + 1} \right)}}} \right){\lambda_{{\min}}}\left({H_{h+1}}\right){I_{N_{h+1}}}, \ldots, 2a\left( {{\mathscr{L}_{Ld}}} \right){\lambda_{{\min}}}\left({H_{d}}\right){I_{N_{d}}},}\right.$ $\left. {{\lambda_{{\min}}}\left( {{{\bar{\mathscr{L}}}_{L\left( {d + 1} \right)}}} \right){I_{N_{d+1}}},\ldots,{\lambda_{{\min}}}\left( {{{\bar{\mathscr{L}}}_{Lm}}}\right){I_{N_{m}}}} \right\}$. Noting that $\bar{\rho}{\bar{{\mathscr{L}}}_{F}}\succ 0$ and $\Gamma - \bar{\rho}\bar{\mathscr{L}}_{FL}^{\rm{T}}{H_{F}}{\bar{\mathscr{L}}}_{F}^{-1}{H_{F}}{\bar{\mathscr{L}}_{FL}} \succ 0$, we have from Lemma \ref{lemma5} that $\Pi \succ0$. It then follows from \eqref{eq37} that
\begin{equation*}
\begin{split}
{{\bar V}_3} &\le  - {\lambda_{{\min}}}\left({\Pi}\right)\left( {{\rho_1} + {\rho_2} \frac{{\dot \varphi\left(t, T_1\right) }}{\varphi\left(t, T_1\right) }} \right){\bar{E}^{\rm{T}}}\bar{E}\\
&\le  - {\rho_1} \frac{{{\lambda_{{\min}}}\left({\Pi}\right)}}{{{\lambda_{{\max}}}\left( {H} \right)}} {{\bar V}_3} - {\rho_2}\frac{{{\lambda_{{\min}}}\left({\Pi}\right)}}{{{\lambda_{{\max}}}\left( H \right)}} \frac{{\dot \varphi\left(t, T_1\right)}}{\varphi\left(t, T_1\right)}{{\bar V}_3},
\end{split}
\end{equation*}
which, by Lemma \ref{lemma6}, gives
\begin{equation*}
{\bar V_3}\left( t \right)\left\{ \begin{array}{l}
 \le {\varphi\left(t, T_1\right) ^{-{\rho_2} \frac{{{\lambda_{{\min}}}\left({\Pi}\right)}}{{{\lambda_{{\max}}}\left(H\right)}}}}{\exp ^{-{\rho_1}{\frac{{{\lambda_{{\min}}}\left({\Pi}\right)}}{{{\lambda_{{\max}}}\left(H\right)}}}t}}{\bar V_3}\left( 0 \right),{\rm{ }}t \in \left[ {0,T_1} \right)\\
 \equiv 0,  \;\;\;\;\;\;\;\;\;\;\;\;\;\;\;\;\;\;\;\;\;\;\;\;\;\;\;\;\;\;\; t \in \left[ {T_1, + \infty } \right).
\end{array} \right.
\end{equation*}
This implies that $\left\|\bar{E}\right\|^2 \le \frac{{{\lambda_{{\max}}}\left( H \right)}}{{{\lambda_{{\min}}}\left( H\right)}}{\varphi\left(t, T_1\right)^{-{\rho_2} \frac{{{\lambda_{{\min}}}\left({\Pi}\right)}}{{{\lambda_{{\max}}}\left( H\right)}}}}{\exp ^{-{\rho_1} \frac{{{\lambda_{{\min}}}\left({\Pi}\right)}}{{{\lambda_{{\max}}}\left(H\right)}}t}}\left\| {\bar{E}\left( 0 \right)} \right\|^2$ on $\left[{0,T_1} \right)$, and $\bar{E} \equiv 0$ on $\left[ {T_1, + \infty }\right)$. Due to ${\lim_{t \to {T_1^-}}}{\varphi\left(t, T_1\right)^{-{\rho_2}\frac{{{\lambda_{{\min}}}\left({\Pi} \right)}}{{{\lambda_{{\max}}}\left( H\right)}}}} = 0$, we have
${\lim_{t\to {T_1^-}}}{\left\|\bar{E} \right\|^2} = 0$, which implies ${\lim_{t\to {T_1^-}}}{\bar{E}} = 0$. It follows that ${\lim_{t \to {T_1}}}{\bar{E}} = 0$ and ${\bar{E}} = 0,\forall t \geq T_1$.
Thus, we have that
\begin{equation}\label{eq38}
\left\{ \begin{array}{l}
{\lim _{t \to T_1}}{{E}_{F}} = 0\\
{{E}_{F}} = 0,\forall t \ge T_1,
\end{array} \right.
\end{equation}
and
\begin{equation}\label{eq39}
\left\{ \begin{array}{l}
{\lim _{t \to T_1}}{{E}_{Lk}} = 0\\
{{E}_{Lk}} = 0,\forall t \ge T_1,
\end{array} \right.
\end{equation}
where $h + 1\le k\le m$. In view of the structural balance and strong connectivity of ${\mathscr{G}_{k}}$, one can obtain from Lemma \ref{lemma2} that the null space of ${\mathscr{L}_{Lk}}$ is spanned by $G_k{1_{N_k}}, {h+1 \le k \le d}$, which, by virtue of \eqref{eq30} and \eqref{eq39}, implies that
\begin{equation}\label{eq40}
\mathop {\lim}\limits_{t\to{T_1}}{X_{Lk}} ={G_k}{1_{{N_k}}}x_k^*\;\text{and}\;{X_{Lk}} = {G_k}{1_{{N_k}}}x_k^*,\forall t \ge{T_1},
\end{equation}
where ${x_k^*}\in \mathbb{R}$ and $h + 1 \le k \le d$. Since ${\mathscr{G}_k}, {d+1\le k \le m} $ are structurally unbalanced and strongly connected, we know from Lemma \ref{lemma2} that ${\mathscr{L}_{Lk}}$, ${d+1 \le k \le m}$ are nonsingular. Then, using \eqref{eq30} and \eqref{eq39}, we have
\begin{equation}\label{eq41}
\mathop {\lim}\limits_{t\to{T_1}}{X_{Lk}} = 0 \;\text{and}\; {X_{Lk}} = 0,\forall t \ge {T_1},d + 1 \le k \le m.
\end{equation}
By Lemma \ref{lemma9}, we have that ${\mathscr{L}_F}$ is nonsingular, which, in virtue of \eqref{eq31}, implies ${X_F} = -\mathscr{L}_F^{ - 1}{E_{F}} - \mathscr{L}_F^{-1}\sum\nolimits_{k = 1}^m{{\mathscr{L}_{FLk}}{X_{Lk}}}$. This, together with \eqref{eq38}, yields
\begin{equation}\label{eq42}
\begin{aligned}
\mathop {\lim }\limits_{t \to {T_1}} {X_F}& =  -{\mathscr{L}_F ^{-1}}\sum\limits_{k = 1}^m {{\mathscr{L}_{FLk}}\mathop {\lim }\limits_{t \to {T}} {X_{Lk}}},\\
{X_F} = -&{\mathscr{L}_F ^{-1}}\sum\limits_{k = 1}^m {{\mathscr{L}_{FLk}}{X_{Lk}}},\forall t \ge {T_1}.
\end{aligned}
\end{equation}
Since ${\mathscr{L}_{Lk}}=0$, ${1 \le k \le h}$, we have
\begin{equation}\label{eq43}
{X_{Lk}} = {x_k}\left( 0 \right), 1 \le k \le h.
\end{equation}
Denote ${G_k} = I_{N_k}$ and $x_k^ *  = {x_k}\left( 0 \right),1 \le k \le h.$ Then, using \eqref{eq40}-\eqref{eq43}, we get
\begin{equation}\label{eq44}
\begin{aligned}
\mathop {\lim }\limits_{t \to {T_1}} {X_F}&= -{\mathscr{L}_F ^{-1}}\sum\limits_{k = 1}^d {{\mathscr{L}_{FLk}}{G_k}} {1_{{N_k}}}x_k^ *,\\
{X_F} =-&{\mathscr{L}_F ^{-1}}\sum\limits_{k = 1}^d {{{\mathscr{L}_{FLk}}{G_k}} {1_{{N_k}}}x_k^*},\forall t \ge {T_1}.
\end{aligned}
\end{equation}
Moreover, let us denote $x_k^ * = 0$ and ${G_k} = {I_{N_k}}, d+1 \le k \le m.$
We can rewrite \eqref{eq44} as
\begin{equation}\label{eq45}
\begin{aligned}
\mathop {\lim }\limits_{t \to {T_1}} {X_F}& =-{\mathscr{L}_F ^{-1}}\sum\limits_{k = 1}^m {{\mathscr{L}_{FLk}}{G_k}} {1_{{N_k}}}x_k^ *,\\
{X_F} = -&{\mathscr{L}_F ^{-1}}\sum\limits_{k = 1}^m {{{\mathscr{L}_{FLk}}{G_k}} {1_{{N_k}}}x_k^*},t \ge {T_1}.
\end{aligned}
\end{equation}
By letting ${\mathbf{G}_k}= {G_k}, 1 \le k \le m$, \eqref{eq45} can be rewritten as
\begin{equation}\label{eq46}
\begin{aligned}
\mathop {\lim }\limits_{t \to {T_1}} {X_F}&= -{\mathscr{L}_F ^{-1}}\sum\limits_{k = 1}^m {{\mathscr{L}_{FLk}}{\mathbf{G}_k}} {1_{{N_k}}}x_k^*,\\
{X_F} =-&{\mathscr{L}_F ^{-1}}\sum\limits_{k = 1}^m {{{\mathscr{L}_{FLk}}{\mathbf{G}_k}} {1_{{N_k}}}x_k^*},t \ge {T_1}.
\end{aligned}
\end{equation}
In view of \eqref{eq46} and using Lemma \ref{lemma9}, we further have
\begin{equation*}
\begin{split}
\mathop {\lim }\limits_{t \to {T_1}} \left| {{x_l}} \right| = &\left| {\sum\limits_{k = 1}^m {{\varpi _{lk}}} x_k^ * } \right|\\
& \le \left( {\sum\limits_{k = 1}^m {\left| {{\varpi _{lk}}} \right|} } \right)\mathop {\max }\limits_{1 \le k \le m} \left| {x_k^ * } \right|\\
& = \left( {\sum\limits_{k = 1}^m {\left| {{\varpi _{lk}}} \right|}} \right)\mathop {\max }\limits_{k \in \mathcal{L}} \mathop {\lim }\limits_{t \to {T_1}} \left| {{x_k}} \right|\\
&\le \mathop {\max }\limits_{k \in \mathcal{L}} \mathop {\lim }\limits_{t \to {T_1}} \left| {{x_k}} \right|,\;{l \in \mathcal{F}},
\end{split}
\end{equation*}
and
$$\left|{{x_l}}\right|\le\mathop{\max}\limits_{k\in\mathcal{L}}\left|{{x_k}}\right|,t\ge {T_1},l\in \mathcal{F}.$$
Therefore, prescribed-time bipartite containment is reached for the CAN (\ref{eq3}) within the pre-specified finite time $T_1$.

2) Since ${\mathscr{G}_k}, {1 \le k \le m} $ are structurally unbalanced, it follows that ${N_k} > 1, 1 \le k \le m$. In virtue of Lemma \ref{lemma3}, we know that there exist positive diagonal matrices ${{\it{\Theta}}_{Lk}}$ so that ${\tilde{\mathscr{L}}_{Lk}} \buildrel \Delta \over ={{\it{\Theta}}_{Lk}}{\mathscr{L}_{Lk}} + \mathscr{L}_{Lk}^{\rm{T}}{\it{\Theta}_{Lk}} \succ 0, {1 \le k \le m}$. By Lemma \ref{lemma9}, we have ${\tilde{\mathscr{L}}_{F}} \buildrel \Delta \over ={{\it{\Theta}}_{F}}{\mathscr{L}_{F}} + \mathscr{L}_{F}^{\rm{T}}{{\it{\Theta}}_{F}} \succ 0$. Denote ${\it{\Theta}} = {\rm{diag}}\{{{\it{\Theta}}_{L1}}, \ldots, {\it{\Theta}_{Lm}},\tilde{\rho}{{\it{\Theta}}_{{F}}}\}$, where $0 < \tilde{\rho} < \frac{{{\lambda_{{\min}}}\left({{\rm{diag}}\left\{ {{\tilde{\mathscr{L}}_{L1}}, \ldots ,{{\tilde{\mathscr{L}}_{Lm}}}} \right\}} \right)}}{{{\lambda_{{\max}}}\left( {\mathscr{L}_{FL}^{\rm{T}}{{\it{\Theta}}_{F}}{\tilde{\mathscr{L}}_{F}}^{-1}{{\it{\Theta}}_{F}}{\mathscr{L}_{FL}}} \right)}}$. Further, let ${W} = {{\it{\Theta}}}{\mathscr{L}} + \mathscr{L}^{\rm{T}}{{\it{\Theta}}}$. Since $\tilde{\rho}{\tilde{\mathscr{L}}_{F}}\succ 0$ and ${\rm{diag}}\left\{{{\tilde{\mathscr{L}}_{L1}}, \ldots ,{\tilde{\mathscr{L}}_{Lm}}} \right\} - \tilde{\rho} \mathscr{L}_{FL}^{\rm{T}}{{\it{\Theta}}_{F}}\tilde{\mathscr{L}}_{F}^{-1}{{\it{\Theta}}_{F}}{\mathscr{L}_{FL}} \succ 0$, we have ${W}\succ0$ by Lemma \ref{lemma5}.
Consider the following Lyapunov function candidate:
$${\tilde{V_3}} = {E^{\rm{T}}}{\it{\Theta}} E.$$
Its time derivative satisfies
\begin{equation*}
\begin{split}
{\dot{\tilde{V}}_3} &= 2{E^{\rm{T}}}{\it{\Theta}}\dot E\\
& = - \left({{\rho_1}  + {\rho_2} \frac{{\dot \varphi\left(t, T_1\right) }}{\varphi\left(t, T_1\right) }} \right){E^{\rm{T}}}\left( {{\it{\Theta}}\mathscr{L} + {\mathscr{L}^{\rm{T}}}{\it{\Theta}}} \right)E\\
& =  - \left( {{\rho_1}  + {\rho_2} \frac{{\dot \varphi\left(t, T_1\right) }}{\varphi\left(t, T_1\right)}} \right){E^{\rm{T}}}{W}E\\
& \le  - {\lambda_{{\min}}}\left({W}\right)\left( {{\rho_1} + {\rho_2} \frac{{\dot \varphi\left(t, T_1\right)}}{\varphi\left(t, T_1\right)}} \right){E^{\rm{T}}}E \\
& \le - {\rho_1} \frac{{{\lambda_{{\min}}}\left({W}\right)}}{{\lambda_{{\max}}}\left( {{\it{\Theta}}} \right)} {{\tilde V}_3} - {\rho_2}\frac{{{\lambda_{{\min}}}\left({W}\right)}}{{{\lambda_{{\max}}}\left({\it{\Theta}} \right)}} \frac{{\dot \varphi\left(t, T_1\right)}}{\varphi\left(t, T_1\right)}{{\tilde V}_3},
\end{split}
\end{equation*}
which, by Lemma \ref{lemma6}, yields
\begin{equation*}
{\tilde V_3}\left( t \right)\left\{ \begin{array}{l}
 \le \varphi\left(t, T_1\right)^{-{\rho_2} \frac{{\lambda_{{\min}}}\left({W}\right)}{{\lambda_{{\max}}}\left({\it{\Theta}} \right)}}{\exp ^{-{\rho_1}{\frac{{{\lambda_{{\min}}}\left({W}\right)}}{{{\lambda_{{\max}}}\left({\it{\Theta}}\right)}}}t}}{{\tilde V}_3}\left( 0 \right),{\rm{ }}t \in \left[ {0,T_1} \right)\\
 \equiv 0,  \;\;\;\;\;\;\;\;\;\;\;\;\;\;\;\;\;\;\;\;\;\;\;\;\;\;\;\;\;\;\; t \in \left[ {T_1, + \infty } \right).
\end{array} \right.
\end{equation*}
This in turn implies that $\left\|E\right\|^2 \le \frac{{{\lambda_{{\max}}}\left({\it{\Theta}} \right)}}{{{\lambda_{{\min}}}\left({\it{\Theta}}\right)}}{\varphi\left(t, T_1\right)^{-{\rho_2} \frac{{{\lambda_{{\min}}}\left({W}\right)}}{{{\lambda_{{\max}}}\left( {\it{\Theta}}\right)}}}}{\exp ^{-{\rho_1} \frac{{{\lambda_{{\min}}}\left({W}\right)}}{{{\lambda_{{\max}}}\left({\it{\Theta}}\right)}}t}}\left\| {E\left( 0 \right)} \right\|^2$ on $\left[{0,T_1} \right)$, and $E \equiv 0$ on $\left[ {T_1, + \infty }\right)$. Furthermore, noting that ${\lim_{t \to {T_1^-}}}\varphi\left(t, T_1\right)^{-{\rho_2}\frac{{\lambda_{{\min}}}\left({W}\right)} {{\lambda_{{\max}}}\left( {\it{\Theta}}\right)}} = 0$, we have ${\lim_{t\to {T_1^-}}}{\left\|E \right\|^2} = 0$, and hence ${\lim_{t\to {T_1^-}}}{E} = 0$. Thus, \eqref{eq12} holds. By Lemmas \ref{lemma3} and \ref{lemma9}, we know that matrices $\mathscr{L}_{Lk}$, ${1 \le k \le m}$ and $\mathscr{L}_{F}$ are nonsingular. Then, using \eqref{eq30} and \eqref{eq31}, we have
\begin{equation*}
\begin{split}
&{X_{Lk}} = - \mathscr{L}_{Lk}^{-1}{E_{Lk}}, {1 \le k \le m} \\
&{X_F}= -\mathscr{L}_F^{-1}{E_F}-\mathscr{L}_F^{-1}\sum\nolimits_{k = 1}^m{{\mathscr{L}_{FLk}}{X_{Lk}}}.
\end{split}
\end{equation*}
This, together with \eqref{eq12}, implies that ${\lim_{t\to {T_1}}}{X} = 0$ and $X = 0,\forall t\geq T_1$. The proof is completed.
\end{proof}

In the special case where $\mathscr{G}$ has no negative edge weights, the CSCs of $\mathscr{G}$ are all structurally balanced. In view of the above Theorem \ref{thm3}, the following result can be given.

\begin{corollary}\label{corol5}
Consider the nominal CAN \eqref{eq3} under the control protocol \eqref{eq4}. Let $\mathscr{G}$ be weakly connected and all of its edge weights are positive. Then, the leaders in every CSC reach prescribed-time consensus and the followers converge towards the convex hull spanned by all the leaders' states within the preset finite time $T_1$.
\end{corollary}

\begin{remark}
Note that if every CSC of $\mathscr{G}$ has exactly one node, the weak connectivity topology condition in Corollary \ref{corol5} collapses into the classical spanning forest condition commonly employed in the containment control literature. Moreover, the result of Corollary \ref{corol5} contains as a special case the standard prescribed-time containment control result derived in Reference \cite{wang2018prescribed} for first-order traditional networks under unsigned digraph topology which contains a spanning forest.
\end{remark}

\subsection{Prescribed-time coordination control with disturbances}
In this subsection, a novel prescribed-time sliding mode control protocol is developed to settle the robust PTCC problems of the CAN \eqref{eq2}. The following sliding variables are proposed:
\begin{equation}\label{eq47}
\begin{split}
{\sigma_k} &= {x_k} + {\varsigma_k},\\
{{\dot \varsigma}_k} & = - u_k^{nom},\;\;\;\; k \in {\mathscr{I}_N},
\end{split}
\end{equation}
with $u_k^{nom} = \left( {{\rho_1}  + {\rho_2} \frac{{\dot \varphi\left(t, T\right) }}{\varphi\left(t,T\right) }} \right)\sum\nolimits_{l = 1}^N {{w_{kl}}\left[ {{x_l} - {\rm{sign}}\left( {{w_{kl}}} \right){x_k}} \right]}$, where $T= T_r + T_s$ with $T_r>0$ and $T_s>0$, and other variables are as stated previously. The sliding mode control protocol is designed as
\begin{equation}\label{eq48}
{u_k} = u_k^{dis} + u_k^{nom},\;\;\;\; k \in {\mathscr{I}_N},
\end{equation}
where $u_k^{dis}  = - {\mu_1}{\rm{sign}}\left( {{\sigma_k}} \right) - \left( {{\mu_2}  + {\mu_3} \frac{{\dot \varphi\left(t, T\right) }}{\varphi\left(t, T\right)}}\right){\sigma_k}$ with $T = {T_r}$, ${\mu_1}> \delta$, ${\mu_2}>0$, and ${\mu_3}>0$.

\begin{remark}
The first component of $u_k^{dis}$ is designed for compensating the bounded perturbation ${d_k}$, while the second term $-\left( {{\mu_2}  + {\mu_3} \frac{{\dot \varphi\left(t, T_r\right) }}{\varphi\left(t, T_r\right) }} \right){\sigma_k}$ is responsible for driving the sliding variables to zero in the prescribed finite time $T_r$.
\end{remark}

The following proposition shows that the sliding mode control protocol \eqref{eq48} can drive the states of all agents in \eqref{eq2} with any initial values towards the sliding manifold $\mathcal{S}\triangleq \left\{{\left( {{x_1}, \ldots ,{x_N}} \right)\left| {{\sigma_k} = 0}, k \in {\mathscr{I}_N}\right.} \right\}$ in the prescribed finite time $T_r$ and keeps them on it thereafter in spite of the disturbances $d_k, k \in {\mathscr{I}_N}$.

\begin{proposition}\label{prop1}
Consider the CAN \eqref{eq2} under the sliding mode control protocol \eqref{eq48} with its sliding variables as designed in \eqref{eq47}. Then, the closed-loop CAN \eqref{eq2} can reach the sliding manifold $\mathcal{S}$ in the prescribed finite time $T_r$ for arbitrary initial conditions.
\end{proposition}

\begin{proof}
Using \eqref{eq2}, \eqref{eq47}, and \eqref{eq48}, one can obtain
\begin{equation}\label{eq49}
{\dot \sigma_k} =  - {\mu_1}{\rm{sign}}\left( {{\sigma_k}} \right) - \left( {{\mu_2}  + {\mu_3} \frac{{\dot \varphi\left(t, T_r\right) }}{\varphi\left(t, T_r\right) }} \right){\sigma_k} + {d_k},\;\; k \in {\mathscr{I}_N}.
\end{equation}
Select the Lyapunov function candidate as $V = \frac{1}{2}{\sigma^{\rm{T}}}\sigma$, where $\sigma = {\left[ {{\sigma_1}, \ldots ,{\sigma_N}} \right]^{\rm{T}}}$. In view of ${\mu_1}> \delta $ and \eqref{eq49}, we have
\begin{equation}\label{eq50}
\begin{split}
\dot{V} = & {\sigma^{\rm{T}}}\dot \sigma\\
= & - {\mu_1}\sum\nolimits_{k = 1}^N {{\sigma_k}{\rm{sign}}\left( {{\sigma_k}} \right)}  + \sum\nolimits_{k = 1}^N {{\sigma_k}{d_k}} \\
&- \left( {{\mu_2}  + {\mu_3} \frac{{\dot \varphi\left(t, T_r\right) }}{\varphi\left(t, T_r\right) }} \right)\sum\nolimits_{k = 1}^N {{\sigma_k^2}} \\
= & - \left( {{\mu_1} - \delta } \right){\sum\nolimits_{k = 1}^N {\left| {{\sigma_k}} \right|}} - 2\left( {{\mu_2}  + {\mu_3} \frac{{\dot \varphi\left(t, T_r\right) }}{\varphi\left(t, T_r\right) }} \right){\sigma^{\rm{T}}}\sigma\\
\le &- 2{\mu_2} V - 2{\mu_3} \frac{{\dot \varphi\left(t, T_r\right) }}{\varphi\left(t, T_r\right) }V,
\end{split}
\end{equation}
which, by Lemma \ref{lemma6}, yields
\begin{equation*}
{V}\left( t \right)\left\{ \begin{array}{l}
 \le {{\varphi\left(t, T_r\right)}^{-2{\mu_3}}}{\mathrm{exp}^{- 2{\mu_2} t}}V\left(0 \right),{\rm{ }}t \in \left[ {0,T_r} \right)\\
 \equiv 0,  \;\;\;\;\;\;\;\;\;\;\;\;\;\;\;\;\;\;\;\;\;\;\;\; t \in \left[ {T_r, + \infty } \right).
\end{array} \right.
\end{equation*}
This implies that $\left\|\sigma\right\| \le {\varphi\left(t, T_r\right)^{-{\mu_3}}}{\exp ^{-{\mu_2} t}}\left\| {\sigma\left( 0 \right)} \right\|$ on $\left[{0,T_r} \right)$, and $\sigma \equiv 0$ on $\left[ {T_r, + \infty }\right)$. By noting that ${\lim_{t \to {T_r^-}}}\varphi\left(t, T_r\right)^{-{\mu_3}} = 0$, we then have ${\lim_{t\to {T_r^-}}}{\left\|\sigma \right\|} = 0$, which gives
${\lim_{t\to {T_r^-}}}{\sigma} = 0$. It
thus follows that ${\lim_{t\to {T_r}}}{\sigma} = 0$ and $\sigma \equiv 0,\forall t\geq T_r $. Therefore, the states of agents reach the sliding manifold $\mathcal{S}$ in the prescribed finite time $T_r$ and remain on the sliding manifold thereafter. The proof is complete.
\end{proof}

From the above Proposition \ref{prop1}, we can obtain ${\sigma}= {\dot{\sigma}} = 0, \forall t \geq T_r$. This, together with \eqref{eq47}, obviously implies
\begin{equation*}
{{\dot{x}}_k} = \left( {{\rho_1} + {\rho_2}\frac{{\dot \varphi\left(t, T_r + T_s\right) }}{\varphi\left(t, T_r + T_s\right) }} \right)\sum\limits_{l \in {{\mathcal{N}}_k}} {{w_{kl}}\left[ {{x_l} - {\rm{sign}}\left( {{w_{kl}}} \right){x_k}} \right]},
\end{equation*}
$k\in {\mathscr{I}_N}$ for all $t \geq T_r$. With this fact and based on Theorems \ref{thm1}-\ref{thm3}, one can readily derive the following results.

\begin{theorem}\label{thm4}
Consider the CAN \eqref{eq2} under the sliding mode control protocol \eqref{eq48} with its sliding variables as designed in \eqref{eq47}, and let $\mathscr{G}$ be strongly connected. Then, the CAN \eqref{eq2} reaches
\begin{enumerate}
  \item prescribed-time bipartite consensus in the pre-specified finite time ${T_r} + {T_s}$, if $\mathscr{G}$ is structurally balanced.
  \item  prescribed-time stability in the pre-specified finite time ${T_r} + {T_s}$, if $\mathscr{G}$ is structurally unbalanced.
\end{enumerate}
\end{theorem}

\begin{remark}
Notice that Theorems \ref{thm1} and \ref{thm4} are valid for any strongly connected signed digraph $\mathscr{G}$, irrespective of whether it is sign-symmetric or not.
\end{remark}

\begin{theorem}\label{thm5}
Consider the CAN \eqref{eq2} under the sliding mode control protocol \eqref{eq48} with its sliding variables as designed in \eqref{eq47}, and let $\mathscr{G}$ be quasi-strongly connected. Then, the CAN \eqref{eq2}  accomplishes
\begin{enumerate}
  \item  prescribed-time interval bipartite consensus in the pre-specified finite time ${T_r} + {T_s}$, if the subgraph $\mathscr{G}_{L}$ of $\mathscr{G}$ is structurally balanced.
  \item prescribed-time stability in the pre-specified finite time ${T_r} + {T_s}$, if the subgraph $\mathscr{G}_{L}$ of $\mathscr{G}$ is structurally unbalanced.
\end{enumerate}
\end{theorem}

\begin{remark}
Notice that Theorems \ref{thm2} and \ref{thm5} hold for any quasi-strongly signed digraph $\mathscr{G}$, irrespective of its sign patterns. It is worth noticing also that with the quasi-strong connectivity of $\mathscr{G}$, Theorems \ref{thm2} and \ref{thm5} hold irrespective of whether $\mathscr{G}_{F}$ is structurally unbalanced or balanced and can fulfill any connectivity condition or not.
\end{remark}

\begin{theorem}\label{thm6}
Consider the CAN \eqref{eq2} under the sliding mode control protocol \eqref{eq48} with its sliding variables as designed in \eqref{eq47}, and let $\mathscr{G}$ be weakly connected. Then, the CAN \eqref{eq2} achieves
\begin{enumerate}
  \item prescribed-time bipartite containment in the pre-specified finite time ${T_r} + {T_s}$, if at least one CSC of $\mathscr{G}$ is structurally balanced.
  \item prescribed-time stability in the pre-specified finite time ${T_r} + {T_s}$, if  all the CSCs of $\mathscr{G}$ are structurally unbalanced.
\end{enumerate}
\end{theorem}

\begin{remark}
Note that Theorems \ref{thm3} and \ref{thm6} hold for any weakly connected signed digraph $\mathscr{G}$, irrespective of its sign pattern.
\end{remark}

\begin{remark}
With Theorems \ref{thm1}-\ref{thm6}, we establish general prescribed-time coordination control results for single-integrator CANs without and with external disturbances under arbitrary fixed signed digraphs. Our results significantly improve and extend the results of \cite{gong2020distributed}, where the problems of prescribed-time bipartite consensus and interval bipartite consensus for nominal single-integrator CANs over sign-symmetric signed digraphs were investigated, in four aspects: 1) the sign-symmetry requirement upon signed digraphs is relaxed; 2) the limitation that the design of control protocols relies upon the global information of the interaction topology is removed; 3) the prescribed-time stability and bipartite containment control problems are addressed; and 4) the external disturbances are taken into account.
\end{remark}

Parallel to Corollaries \ref{corol1}-\ref{corol5}, we have

\begin{corollary}\label{corol6}
Consider the CAN \eqref{eq2} under the sliding mode control protocol \eqref{eq48} with its sliding variables as designed in \eqref{eq47}, and let $\mathscr{G}$ be undirected and connected. Then, the  CAN \eqref{eq2} achieves
\begin{enumerate}
  \item prescribed-time signed-average consensus in the pre-specified finite time ${T_r} + {T_s}$, if $\mathscr{G}$ is structurally balanced.
  \item prescribed-time stability in the pre-specified finite time ${T_r} + {T_s}$, if $\mathscr{G}$ is structurally unbalanced.
\end{enumerate}
\end{corollary}

\begin{corollary}\label{corol7}
Consider the CAN \eqref{eq2} under the sliding mode control protocol \eqref{eq48} with its sliding variables as designed in \eqref{eq47}. Let $\mathscr{G}$ be strongly connected with its edge weights all positive. Then, the CAN \eqref{eq2} reaches prescribed-time consensus in the pre-specified finite time ${T_r} + {T_s}$.
\end{corollary}

\begin{corollary}\label{corol8}
Consider the CAN \eqref{eq2} under the sliding mode control protocol \eqref{eq48} with its sliding variables as designed in \eqref{eq47}. Let $\mathscr{G}$ be quasi-strongly connected and structurally balanced. Then, the CAN \eqref{eq2} achieves prescribed-time bipartite consensus in the pre-specified finite time ${T_r} + {T_s}$.
\end{corollary}

\begin{corollary}\label{corol9}
Consider the CAN \eqref{eq2} under the sliding mode control protocol \eqref{eq48} with its sliding variables as designed in \eqref{eq47}. Let $\mathscr{G}$ be quasi-strongly connected with its edge weights all positive. Then, the CAN \eqref{eq2} achieves prescribed-time consensus in the pre-specified finite time ${T_r} + {T_s}$.
\end{corollary}

\begin{corollary}\label{corol10}
Consider the CAN \eqref{eq2} under the sliding mode control protocol \eqref{eq48} with its sliding variables as designed in \eqref{eq47}. Let $\mathscr{G}$ be weakly connected with its edge weights all positive. Then, the leaders in each CSC achieve consensus and the followers converge towards the convex hull spanned by all the leaders' states in the pre-specified finite time ${T_r} + {T_s}$.
\end{corollary}

\begin{remark}
Note that the settling time ${T_r}$ to reach the sliding manifold $\mathcal{S}$ and the overall settling time ${T_r} + {T_s}$ for the CAN \eqref{eq2} to achieve coordination control are both independent of initial conditions, design parameters, and topology structure among agents, and can be pre-specified explicitly.
\end{remark}

\section{Numerical Simulations}\label{sec4}
 Numerical examples are provided in this section to validate the effectiveness of the proposed results.

\begin{figure}
\subfigure[]{
  \begin{minipage}{0.23\textwidth}
  \centering
  \includegraphics[scale=0.5]{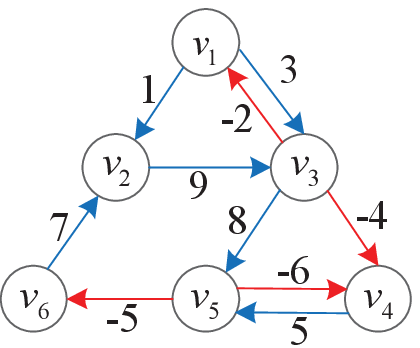}
   \end{minipage}}
\subfigure[]{
  \begin{minipage}{0.23\textwidth}
  \centering
  \includegraphics[scale=0.5]{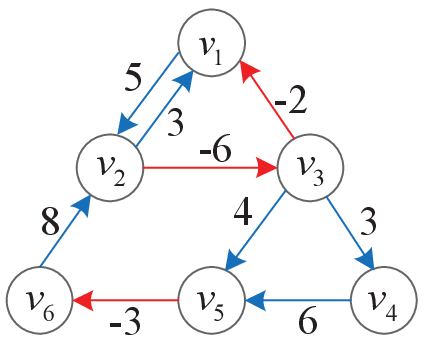}
  \end{minipage}}
\caption{Two strongly connected signed digraphs.\label{Fig1}}
\end{figure}

Example 1. Consider the CAN \eqref{eq3} under the strongly connected signed digraphs in Figure \ref{Fig1}. Obviously, the signed digraph of Figure 1(a) is structurally unbalanced, whereas the signed digraph of Figure 1(b) is structurally balanced. The states of agents are initially set as $X\left( 0 \right) = {\left[ { 5,2, -4, 3, -2, 1} \right]^{\rm{T}}}$. The design parameters of the control protocol \eqref{eq4} are taken as ${T_1}= 0.6s, {\rho_1} = 0.1$,${\rho_2} = 0.3$,$\kappa = 1$. Figure 2 depicts the simulation results. Evidently, stability and bipartite consensus  are, respectively, accomplished within the pre-specified finite time $T_1  = 0.6s$ under the signed digraphs of Figure 1(a) and (b), which illustrates the results of Theorem \ref{thm1}.

\begin{figure}
\subfigure[]{
  \begin{minipage}{0.23\textwidth}
  \centering
  \includegraphics[scale=0.3]{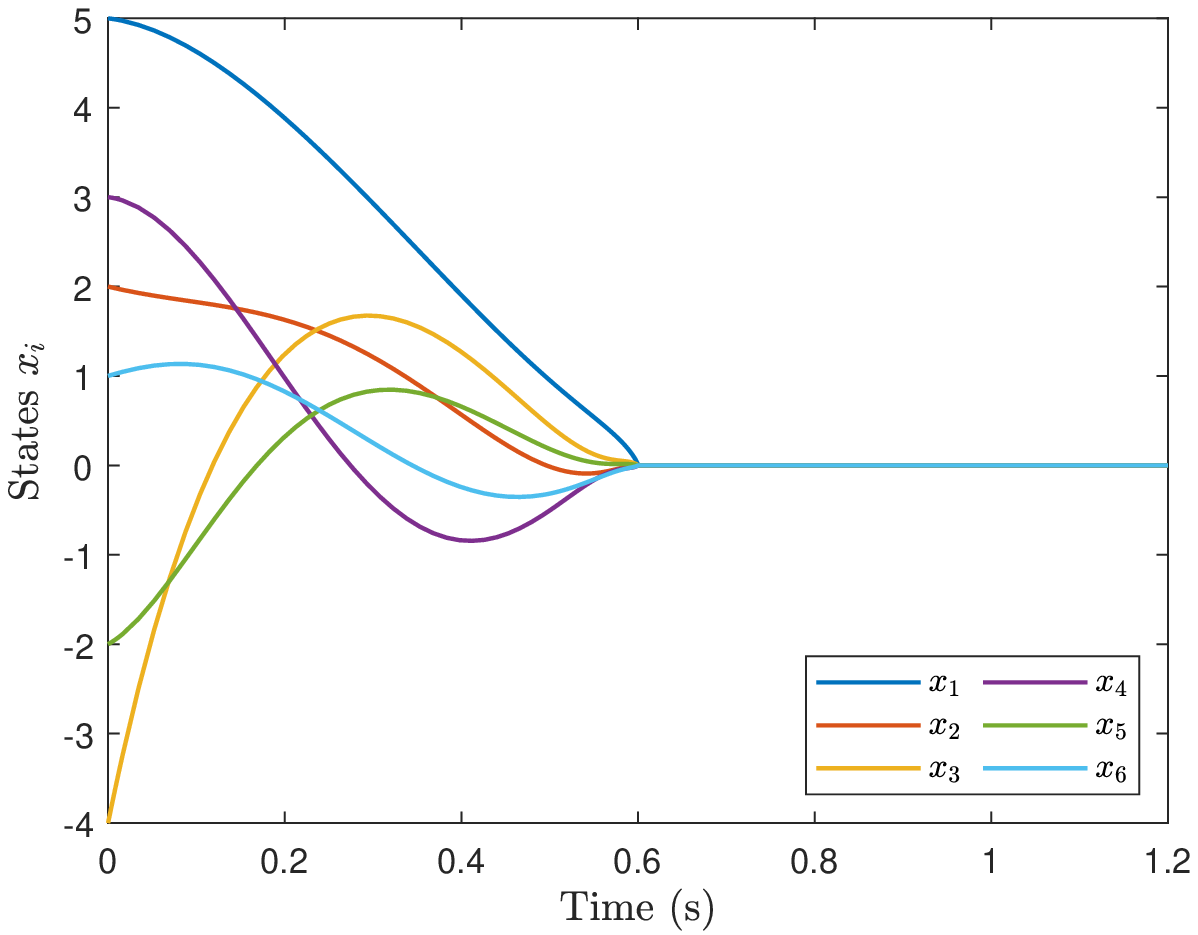}
   \end{minipage}}
\subfigure[]{
  \begin{minipage}{0.23\textwidth}
  \centering
  \includegraphics[scale=0.3]{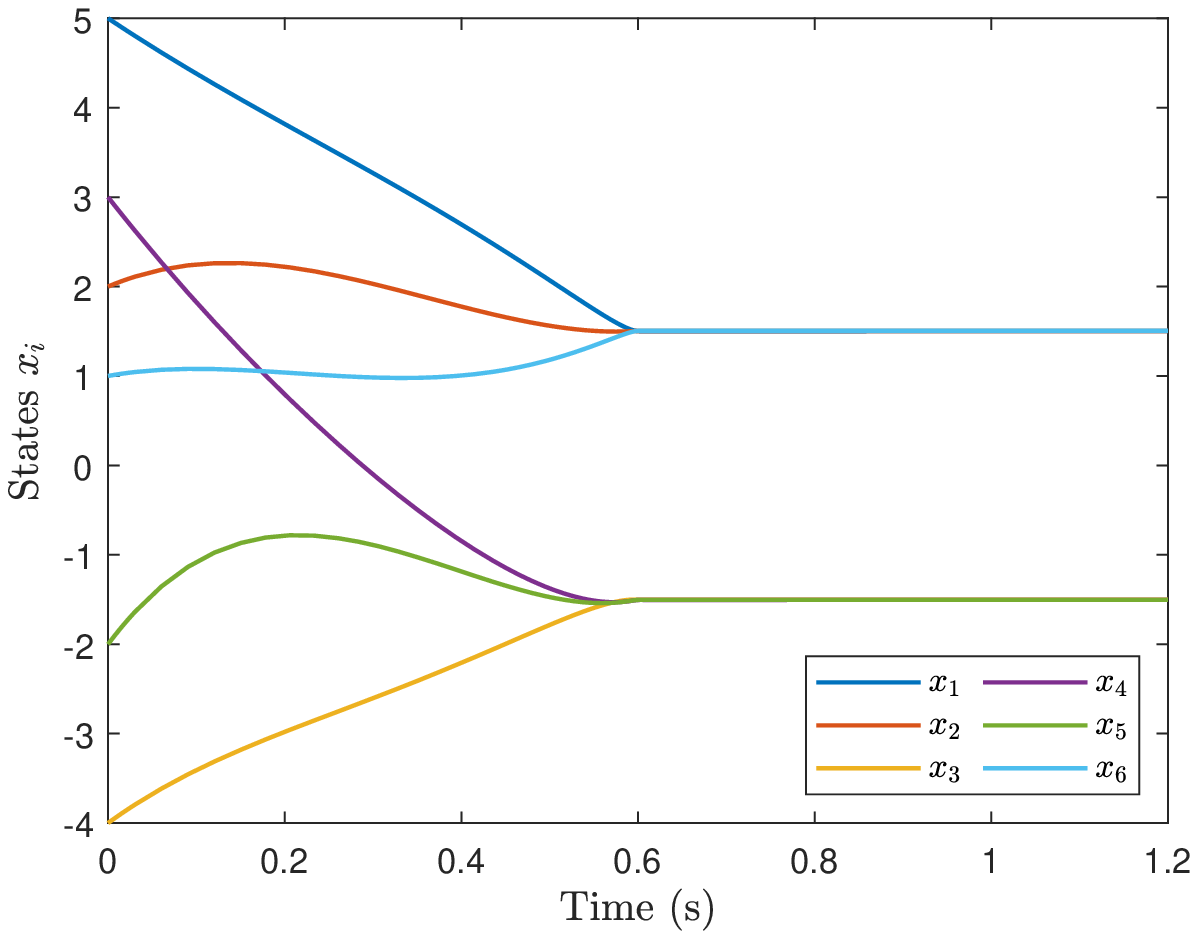}
  \end{minipage}}
\caption{State evolution of the CAN (3) under the signed digraphs of Figure 1. (a) Under Figure 1 (a). (b) Under Figure 1 (b).\label{Fig2}}
\end{figure}

Example 2. Consider the CAN \eqref{eq3} under the quasi-strongly connected signed digraphs in Figure \ref{Fig3}. It is clearly seen that the agents $k$, $1\leq k \leq 3$ are leaders, and the subgraphs of the signed digraphs in Figure \ref{Fig3}(a) and (b) associated with the leaders are structurally balanced and unbalanced, respectively. The design parameters of the control protocol \eqref{eq4} are taken as ${T_1}= 0.6s, {\rho_1} = 0.2$,${\rho_2} = 0.5$,$\kappa = 1$. Set the initial states of agents as $X\left( 0 \right) = {\left[ { - 4,3, - 1,2, - 2,5} \right]^{\rm{T}}}$. The simulation results are presented in Figure 4. One can observe from Figure 4 (a) and (b) that the CAN \eqref{eq3} respectively achieves interval bipartite consensus and stability within the pre-specified finite time $T_1  = 0.6$ under the signed digraphs of Figure 3(a) and (b), which verifies the results of Theorem \ref{thm2}.

\begin{figure}
\subfigure[]{
  \begin{minipage}{0.23\textwidth}
  \centering
  \includegraphics[scale=0.5]{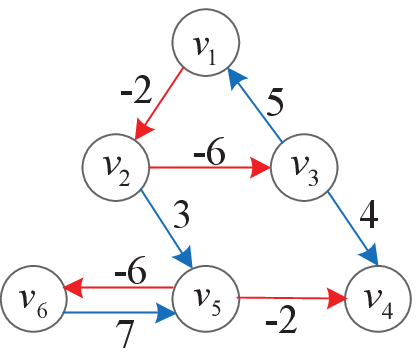}
   \end{minipage}}
\subfigure[]{
  \begin{minipage}{0.23\textwidth}
  \centering
  \includegraphics[scale=0.5]{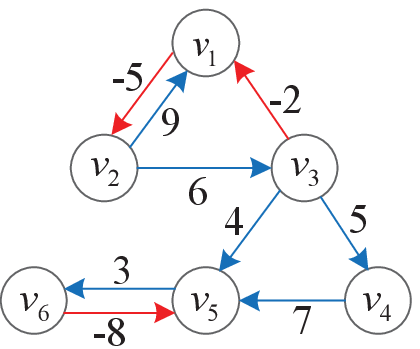}
  \end{minipage}}
\caption{Two quasi-strongly connected signed digraphs.\label{Fig3}}
\end{figure}

\begin{figure}
\subfigure[]{
  \begin{minipage}{0.23\textwidth}
  \centering
  \includegraphics[scale=0.3]{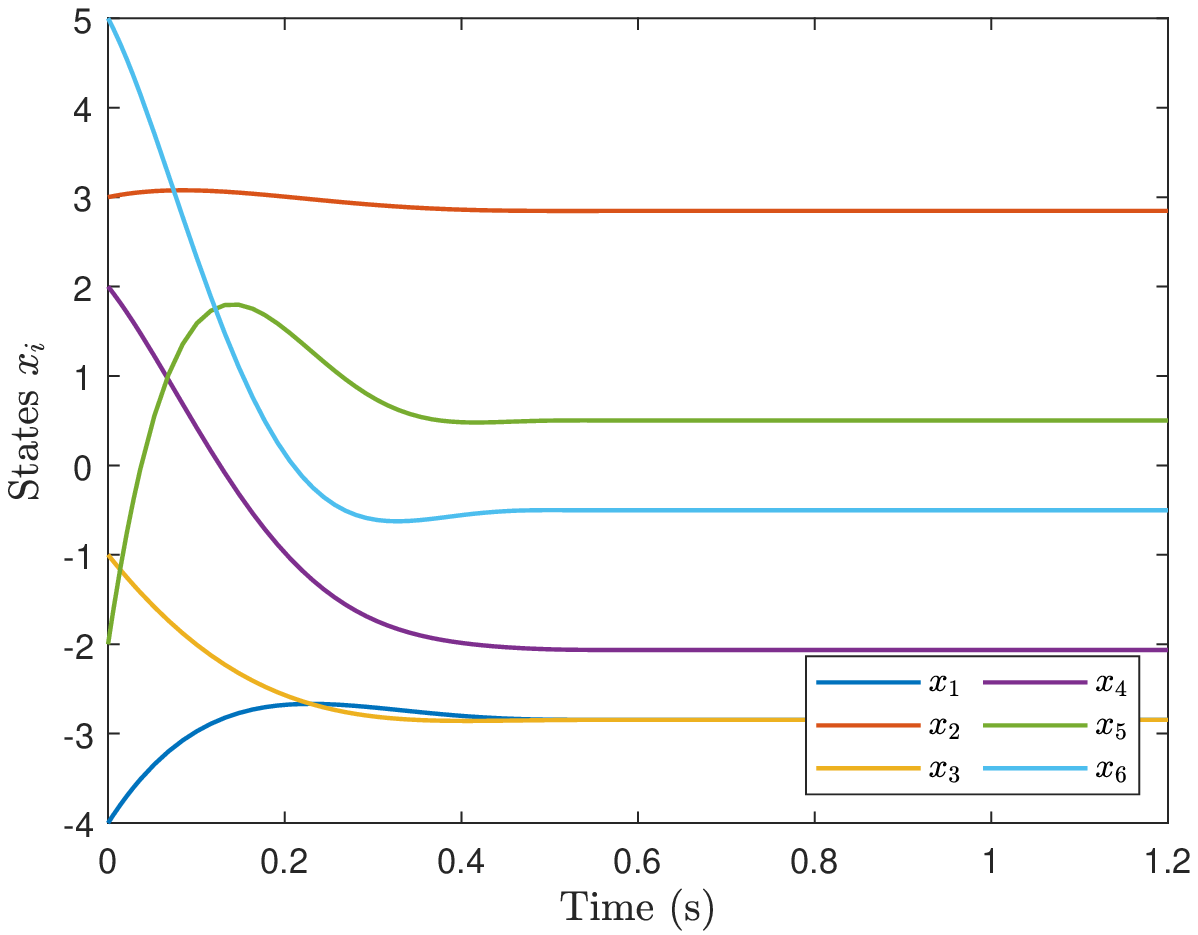}
   \end{minipage}}
\subfigure[]{
  \begin{minipage}{0.23\textwidth}
  \centering
  \includegraphics[scale=0.3]{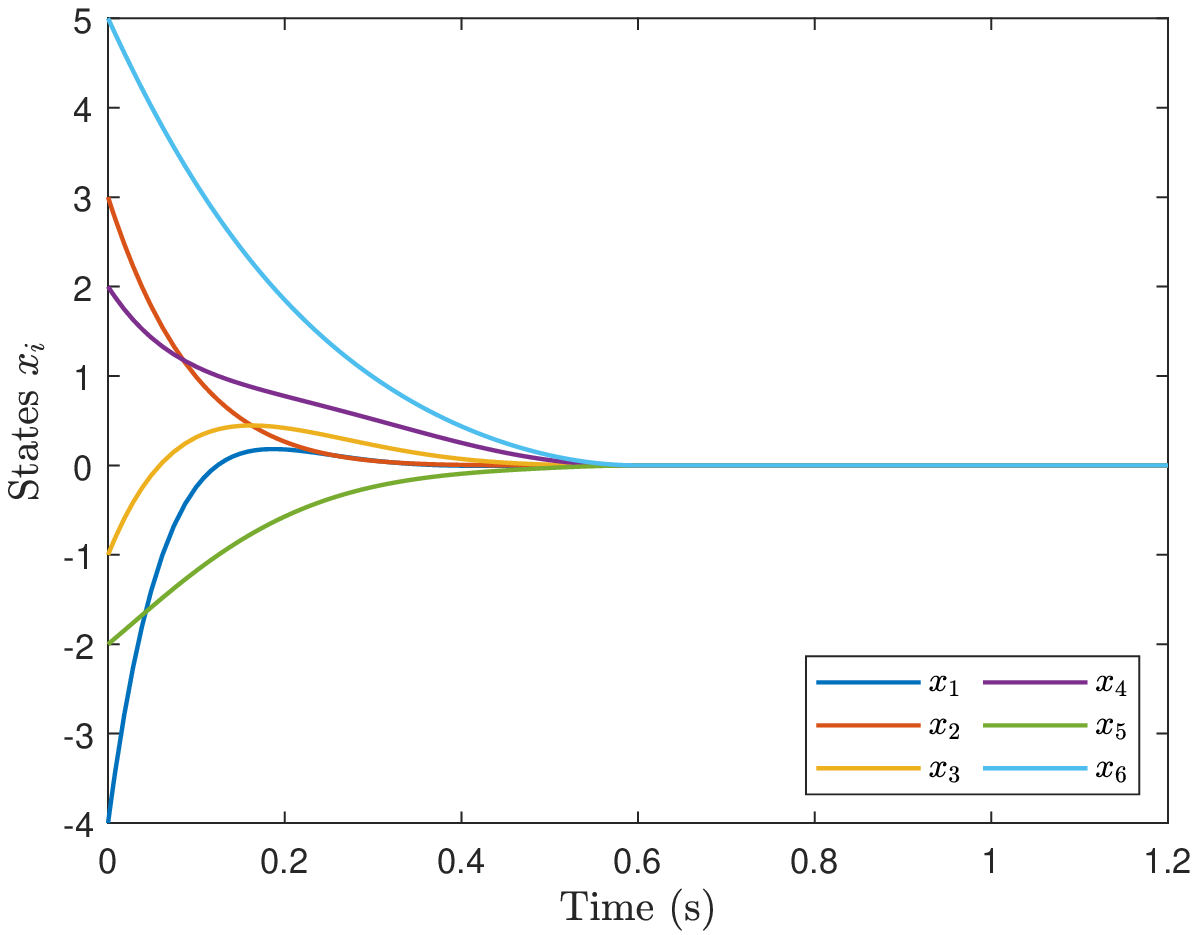}
  \end{minipage}}
\caption{State evolution of the CAN \eqref{eq3} under the signed digraphs of Figure 3. (a) Under Figure 3 (a). (b) Under Figure 3(b).\label{Fig4}}
\end{figure}

Example 3. Consider the CAN \eqref{eq3} under the two weakly connected signed digraphs in Figure \ref{Fig5}. It is easily obtained that the agents $k$, $1\leq k \leq 6$ are leaders and there are two groups of leaders: $\{1, 2, 3\}$ and $\{4, 5, 6\}$. Moreover, the CSCs of the signed digraph in Fig.\ref{Fig5}(a) are not all structurally unbalanced (the CSC of the signed digraph in Fig.\ref{Fig5}(a) associated with the leader group $\{1, 2, 3\}$ is structurally balanced), whereas those of the signed digraph in Fig.\ref{Fig5}(b) are all structurally unbalanced. We select the design parameters of the control protocol \eqref{eq4} as ${T_1}= 0.2s, {\rho_1} = 0.1$,${\rho_2} = 0.3$,$\kappa = 1$. With $X\left( 0 \right) = {\left[ {-6, 4, 5, -7, 8, -5,-3, 7, -5, 6, 4, 2, -5, 3, -8, 1}\right]^{\rm{T}}}$, we conduct simulations and display the simulation results in Figure \ref{Fig6}. From Figure \ref{Fig6}, one can see that bipartite containment and stability are, respectively, achieved under the signed digraphs of Figure \ref{Fig5}(a) and (b) within the pre-specified finite time $T_1 = 0.2$, which confirms Theorem \ref{thm3}.

\begin{figure}
\subfigure[]{
  \begin{minipage}{0.23\textwidth}
  \centering
  \includegraphics[scale=0.5]{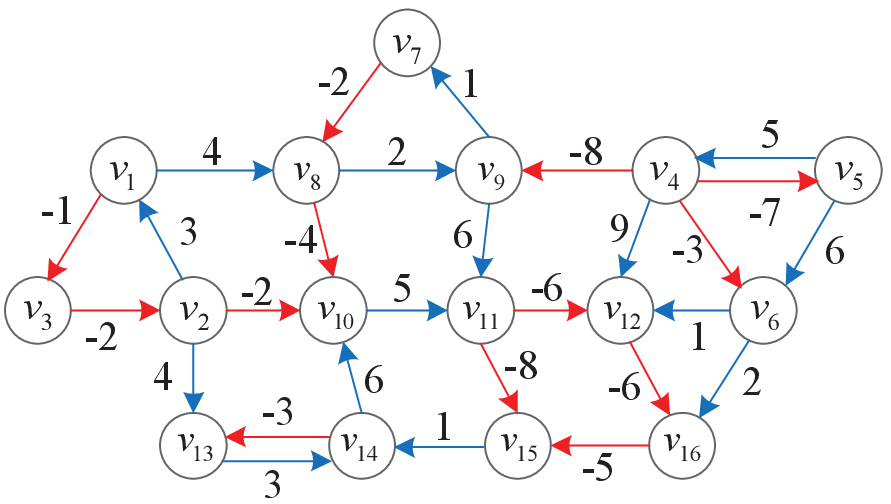}
   \end{minipage}}
\subfigure[]{
  \begin{minipage}{0.23\textwidth}
  \centering
  \includegraphics[scale=0.5]{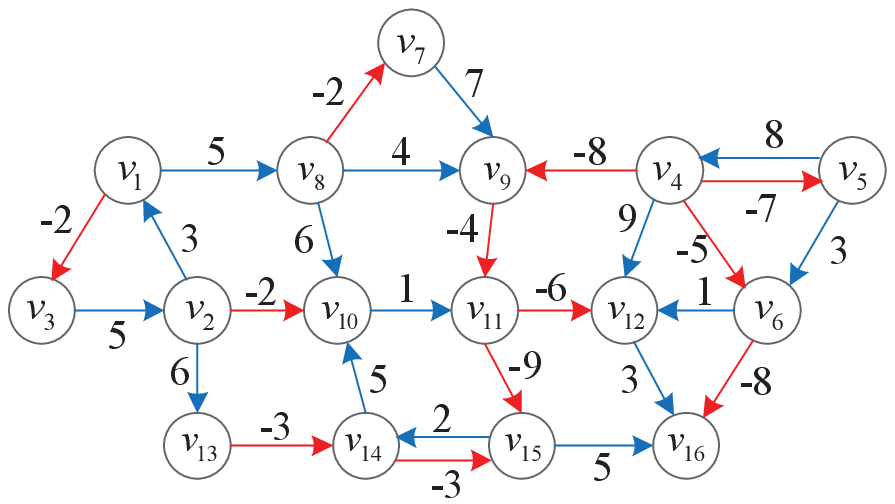}
  \end{minipage}}
\caption{Two weakly connected signed digraphs.\label{Fig5}}
\end{figure}

\begin{figure}
\subfigure[]{
  \begin{minipage}{0.23\textwidth}
  \centering
  \includegraphics[scale=0.3]{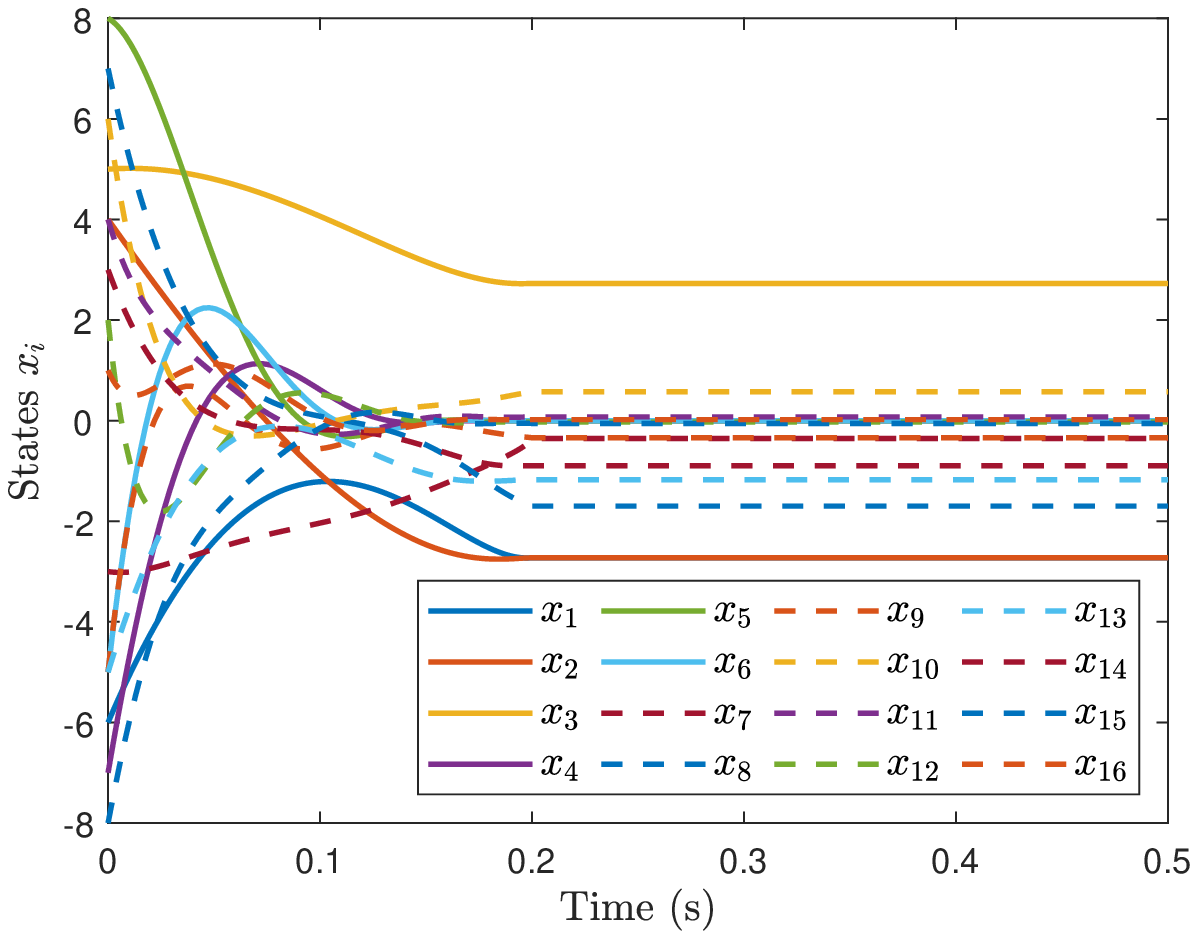}
   \end{minipage}}
\subfigure[]{
  \begin{minipage}{0.23\textwidth}
  \centering
  \includegraphics[scale=0.3]{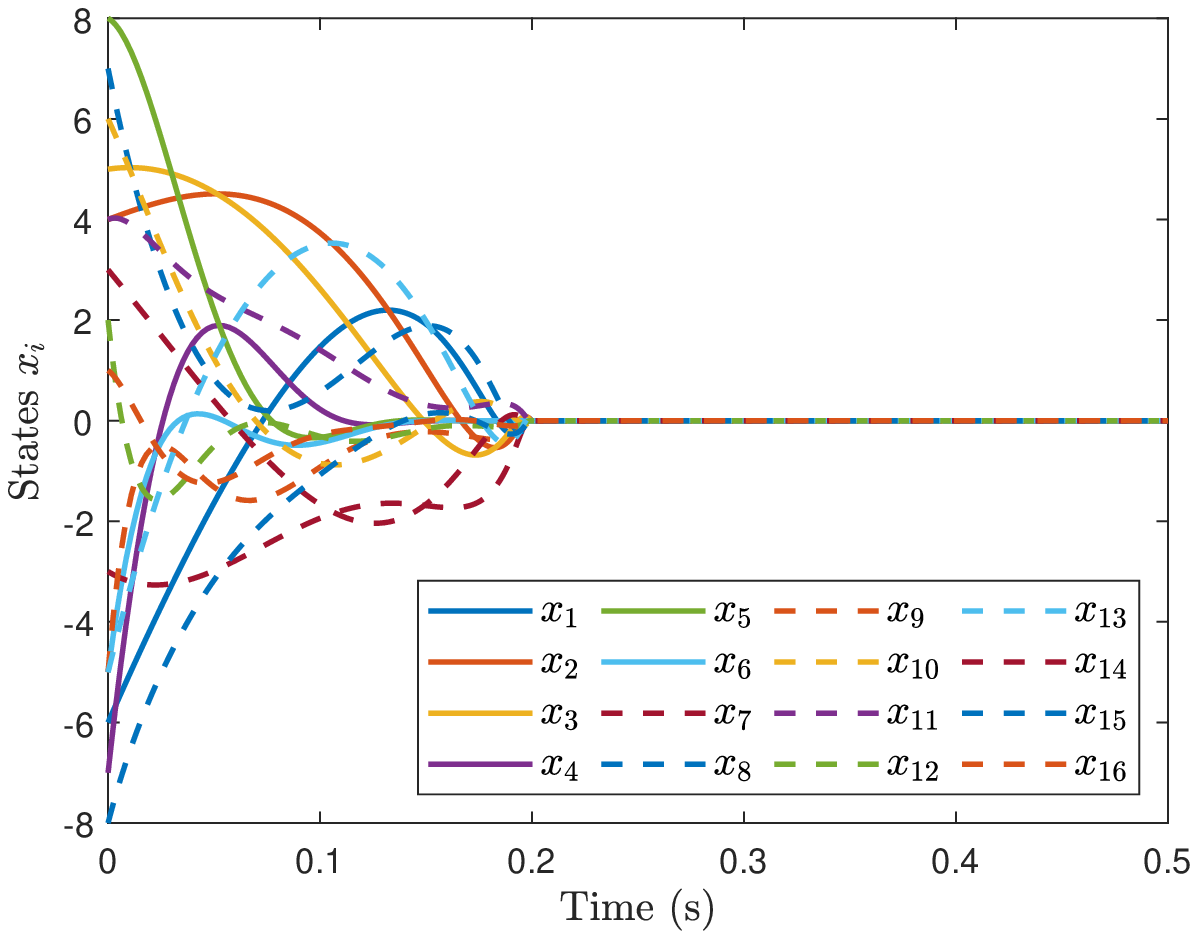}
  \end{minipage}}
\caption{State evolution of the CAN \eqref{eq3} under signed digraphs of Figure 5. (a) Under Figure 5 (a). (b) Under Figure 5 (b).\label{Fig6}}
\end{figure}

Example 4. Let us consider the CAN \eqref{eq2} with ${d_k}= \sin \left( {2kt + \frac{\pi}{3}} \right)$ under the same signed digraphs as considered in Example 1. Since $\left| {{d_k}} \right| \le 1$, we can choose ${\mu_1} = 1.2$. The other control parameters in \eqref{eq47} and \eqref{eq48} are chosen as $\kappa = 2,T_r = 0.5, T_s=1,{\rho_1} = 0.1,{\rho_2} = 0.3,{\mu_2} = 0.6,{\mu_3} = 0.9$.  In the simulations, we set $X\left( 0 \right) = {\left[ {-4,3, -1, 2, -2, 5}\right]^{\rm{T}}}$ and $\sigma\left( 0 \right) = {\left[ {-9,1,  -5,   8,  -4 , 6 }\right]^{\rm{T}}}$. Figures 7 and 8, respectively, show the evolution of the sliding variables and the agents' states under the signed digraphs of Figure 1. From Figure 7, we see that the sliding variables converge to zero within the predefined finite time $T_r = 0.5$. From Figure 8, we see that the CAN \eqref{eq2} reaches stability and bipartite consensus in the predefined finite time ${T_r}+{T_s} = 1.5$ under the signed digraphs of Figure 1(a) and (b), respectively. This demonstrates the results of Theorem \ref{thm4}.

\begin{figure}
\subfigure[]{
  \begin{minipage}{0.23\textwidth}
  \centering
  \includegraphics[scale=0.3]{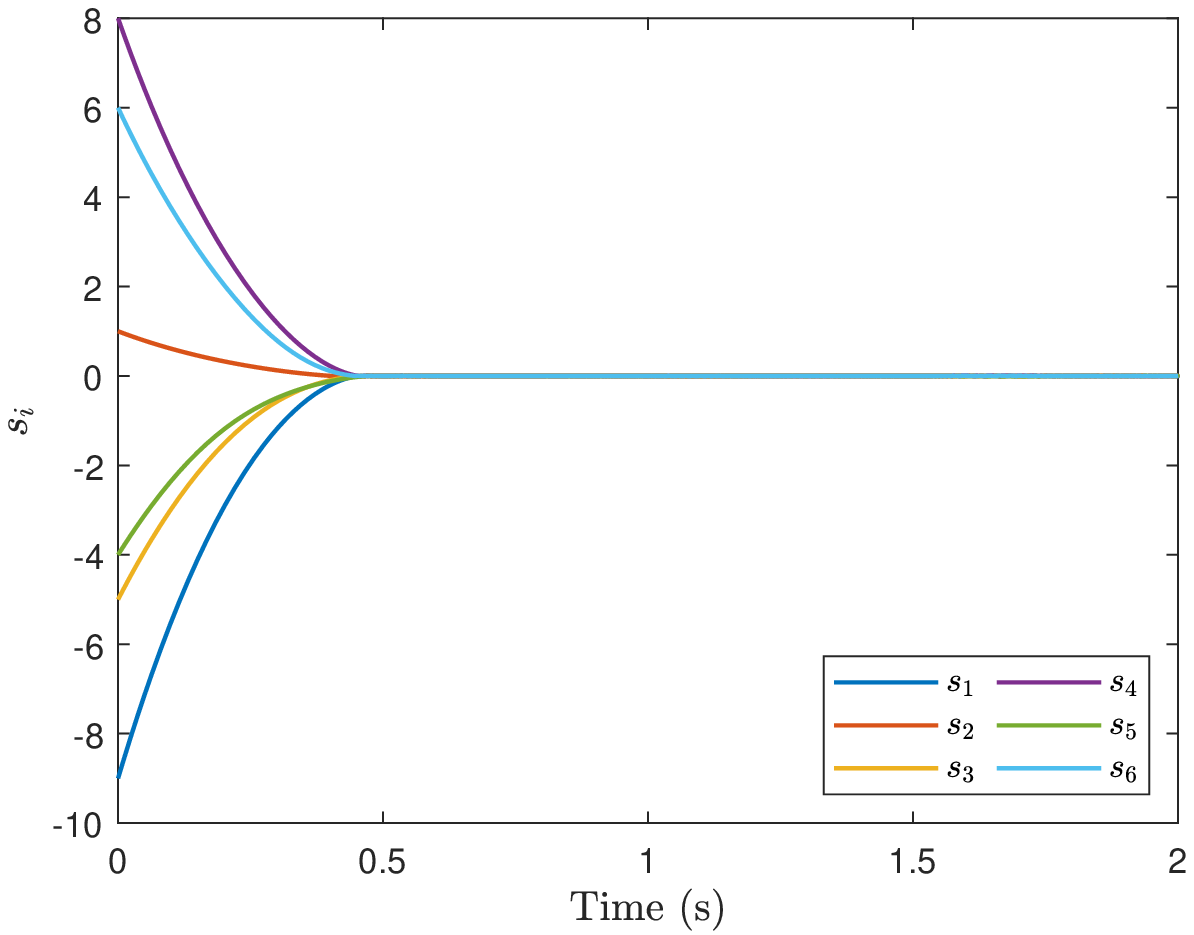}
   \end{minipage}}
\subfigure[]{
  \begin{minipage}{0.23\textwidth}
  \centering
  \includegraphics[scale=0.3]{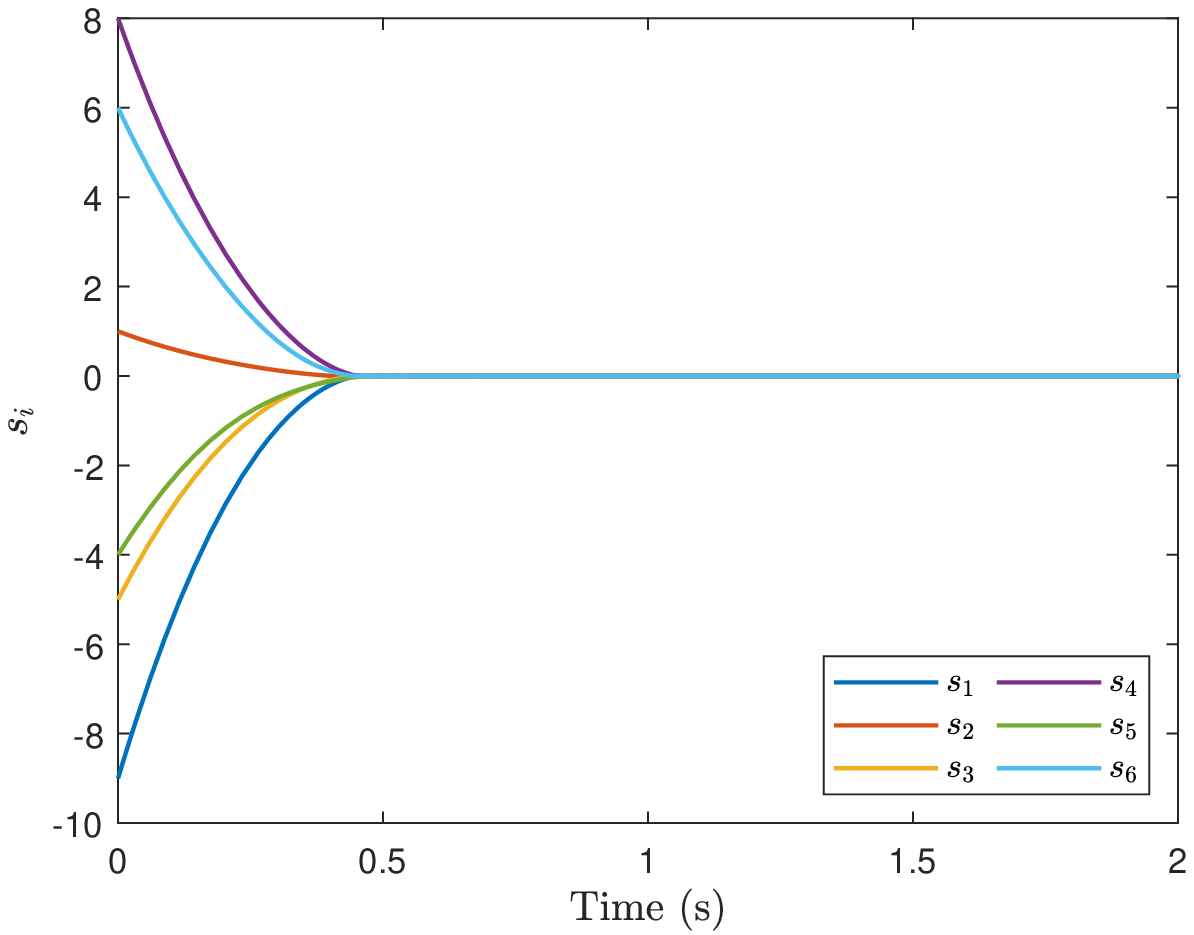}
  \end{minipage}}
\caption{Evolution of the sliding variables under the signed digraphs of Figure 1. (a) Under Figure 1 (a). (b) Under Figure 1 (b)}
\end{figure}\label{Fig7}

\begin{figure}
\subfigure[]{
  \begin{minipage}{0.23\textwidth}
  \centering
  \includegraphics[scale=0.3]{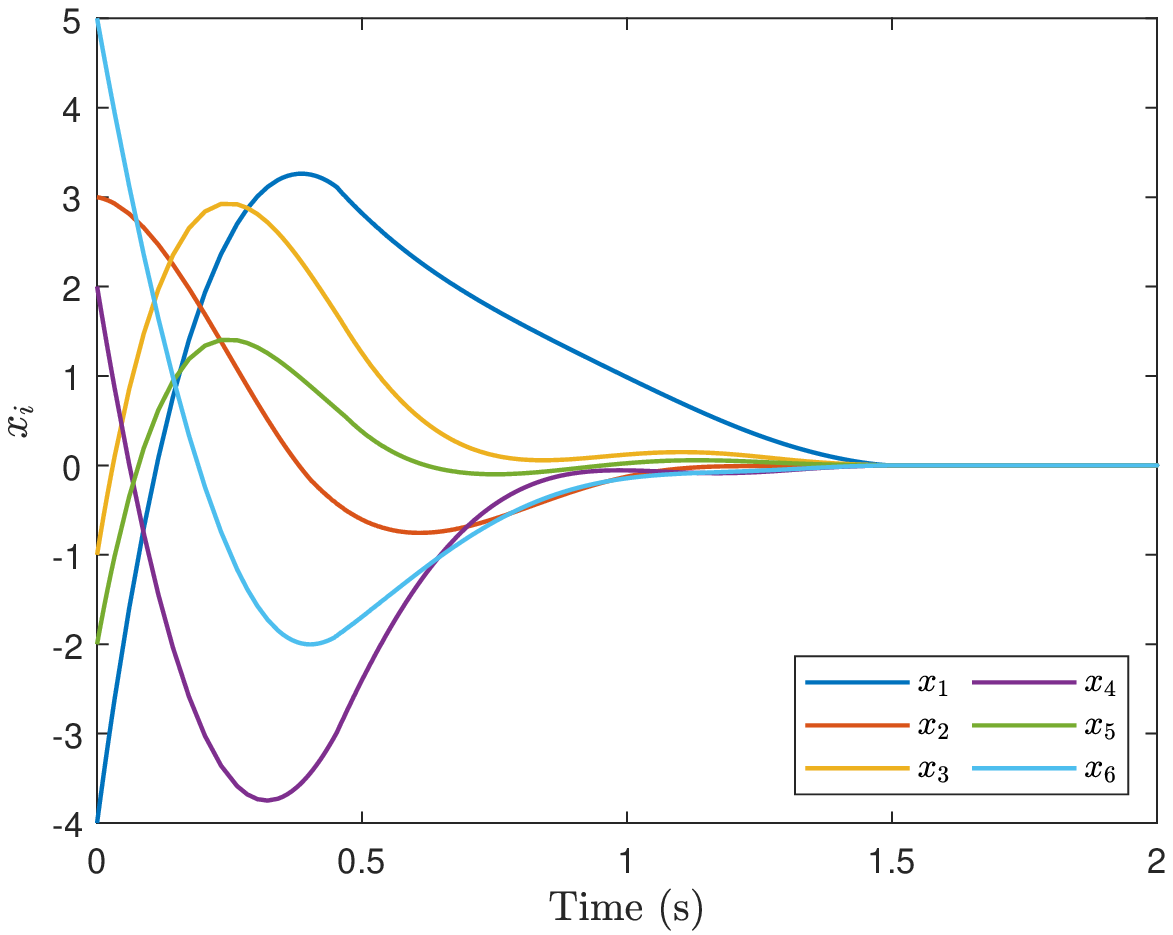}
   \end{minipage}}
\subfigure[]{
  \begin{minipage}{0.23\textwidth}
  \centering
  \includegraphics[scale=0.3]{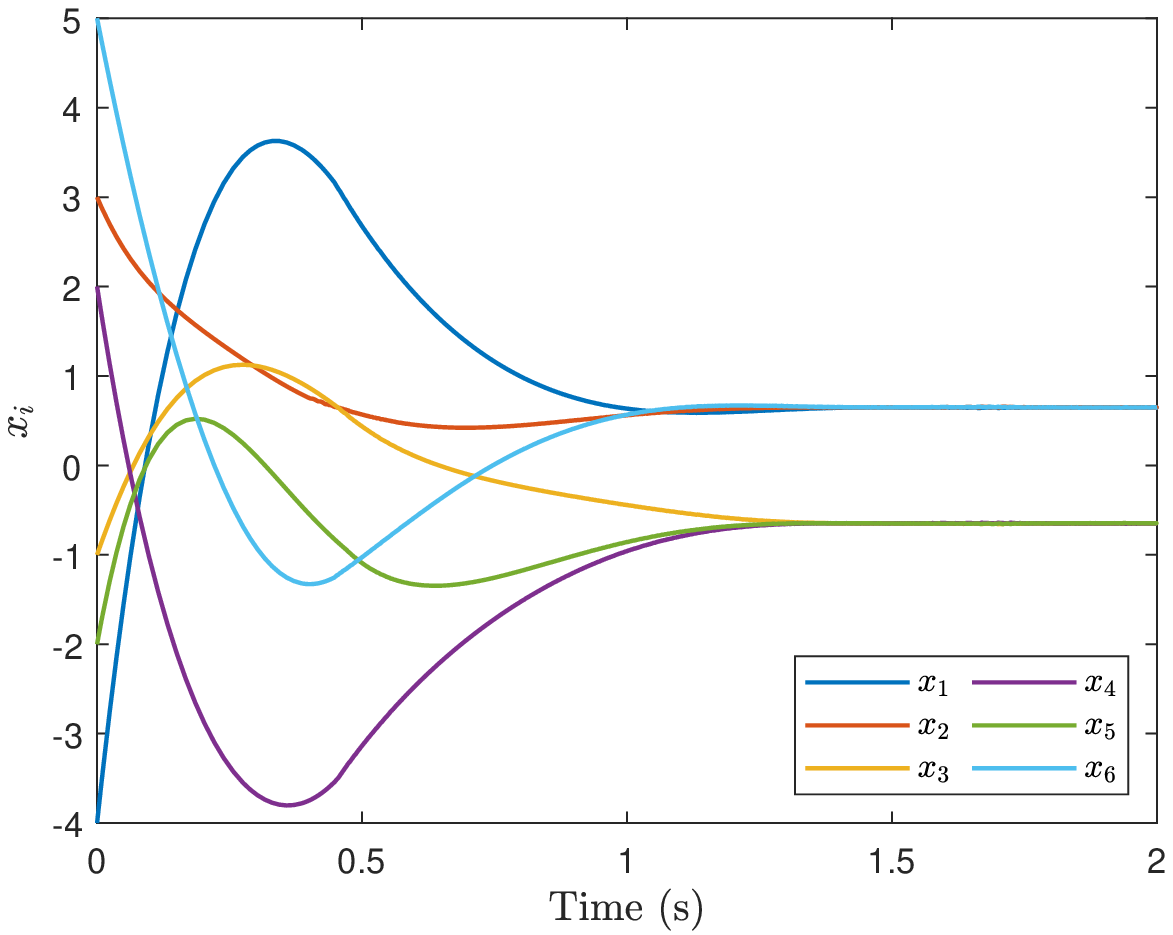}
  \end{minipage}}
\caption{State evolution of the CAN(2) under the signed digraphs of Figure 1. (a) Under Figure 1 (a). (b) Under Figure 1 (b)}
\end{figure}\label{Fig8}

Example 5. Let us consider the CAN \eqref{eq2} with ${d_k}= \sin \left( {2kt + \frac{\pi}{2}} \right)$ under the signed digraphs of Figure 3. The control parameters in \eqref{eq47} and \eqref{eq48} are selected as $\kappa = 3,T_r = 1, T_s= 0.5,{\rho_1} = 0.25,{\rho_2} = 0.3,{\mu_1} = 2,{\mu_2} = 0.4,{\mu_3} = 0.5$.  Let $X\left( 0 \right) = {\left[ {-4,4,5,-7,8,1}\right]^{\rm{T}}}$ and $\sigma\left( 0 \right) = {\left[ {-10, 10, 9, -5, 5  4}\right]^{\rm{T}}}$. Figures 9 and 10,respectively, show the trajectories of sliding variables and the agents under the signed digraphs of Figure 3. From Figure 9, we see that the the convergence of the sliding variables towards zero is attained within the predefined finite time $T_r = 1$. From Figure 10, one can observe that interval bipartite consensus and stability are, respectively, attained within the predefined finite time $T_r + T_s = 1.5$.

\begin{figure}
\subfigure[]{
  \begin{minipage}{0.23\textwidth}
  \centering
  \includegraphics[scale=0.3]{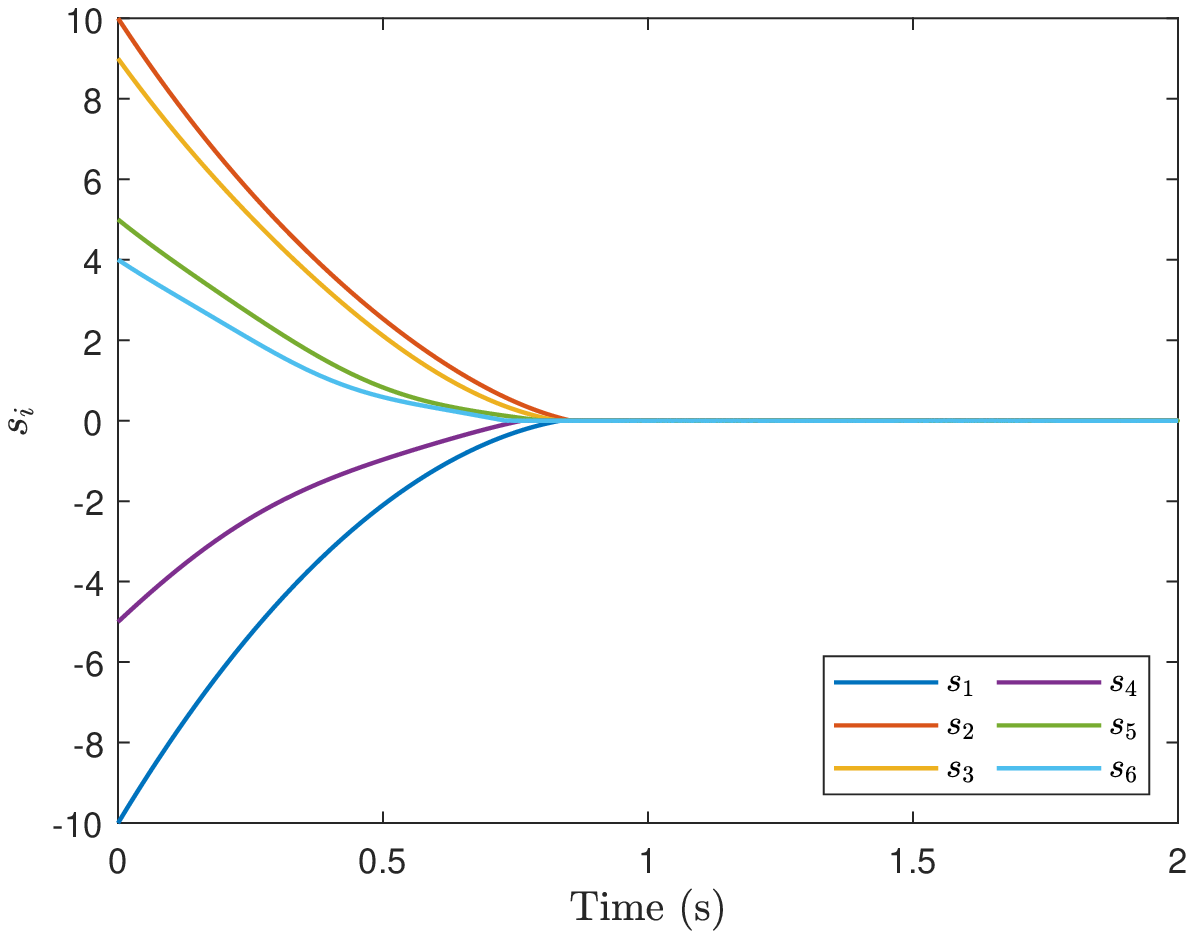}
   \end{minipage}}
\subfigure[]{
  \begin{minipage}{0.23\textwidth}
  \centering
  \includegraphics[scale=0.3]{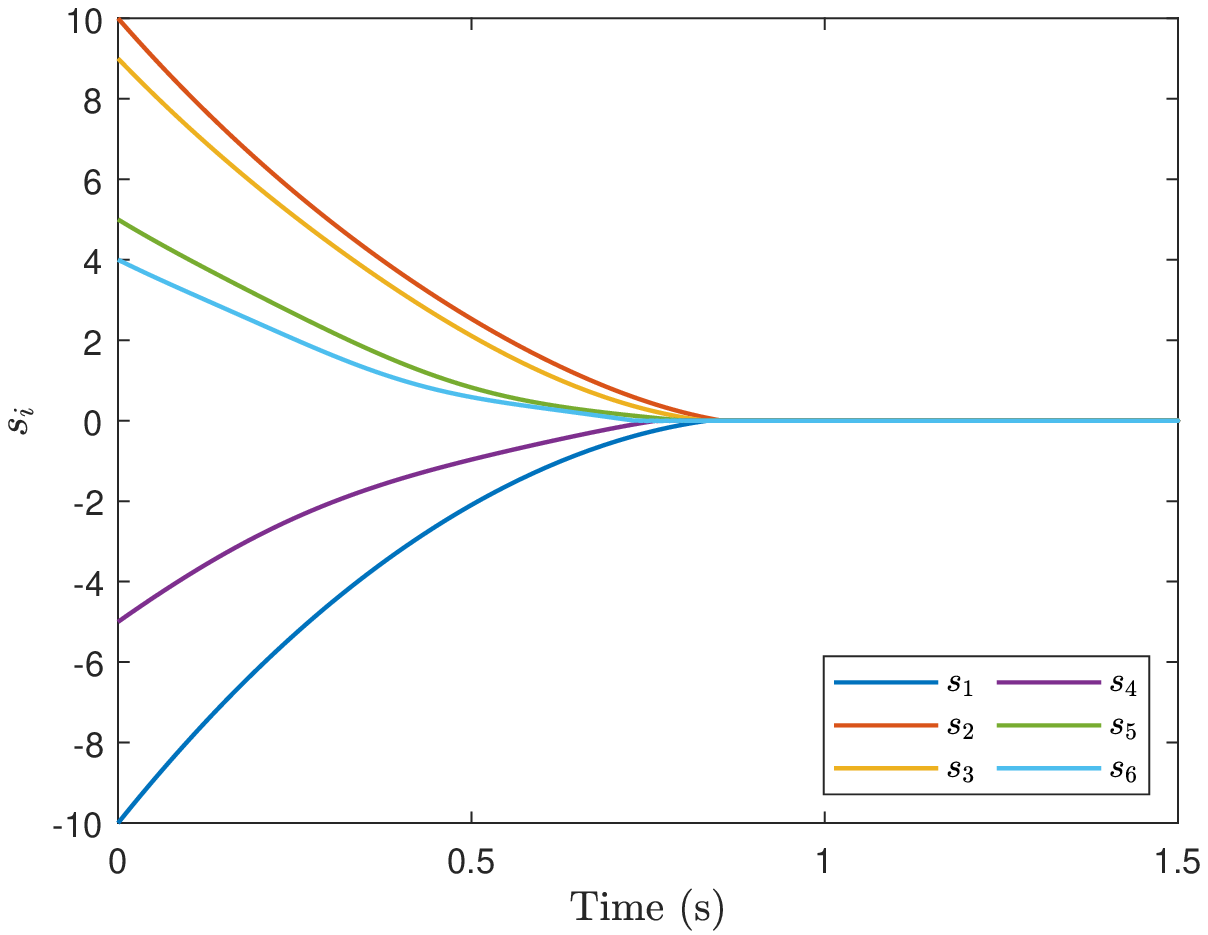}
  \end{minipage}}
\caption{Evolution of the sliding variables  under the signed digraphs of figure 3. (a) Under figure 3 (a). (b) Under figure 3 (b)}
\end{figure}\label{fig9}

\begin{figure}
\subfigure[]{
  \begin{minipage}{0.23\textwidth}
  \centering
  \includegraphics[scale=0.3]{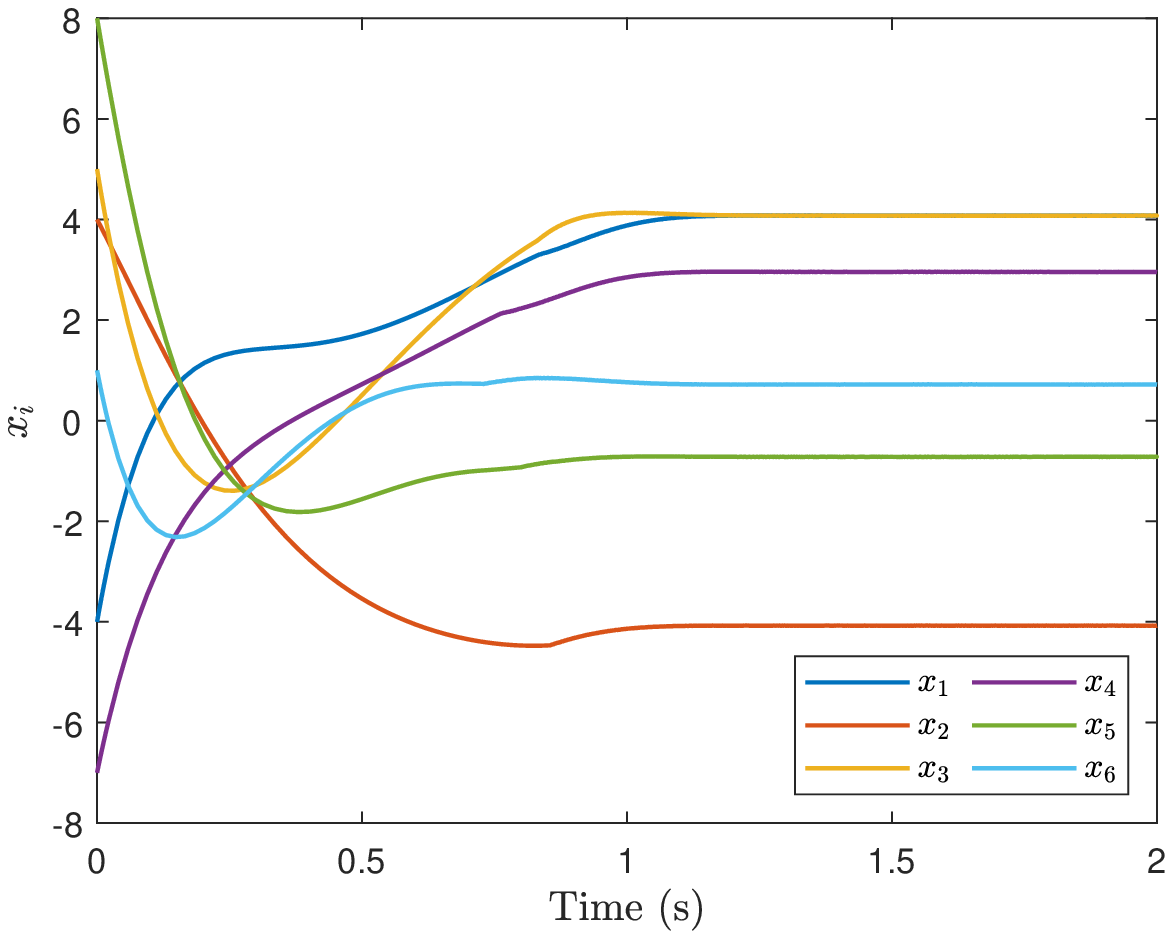}
   \end{minipage}}
\subfigure[]{
  \begin{minipage}{0.23\textwidth}
  \centering
  \includegraphics[scale=0.3]{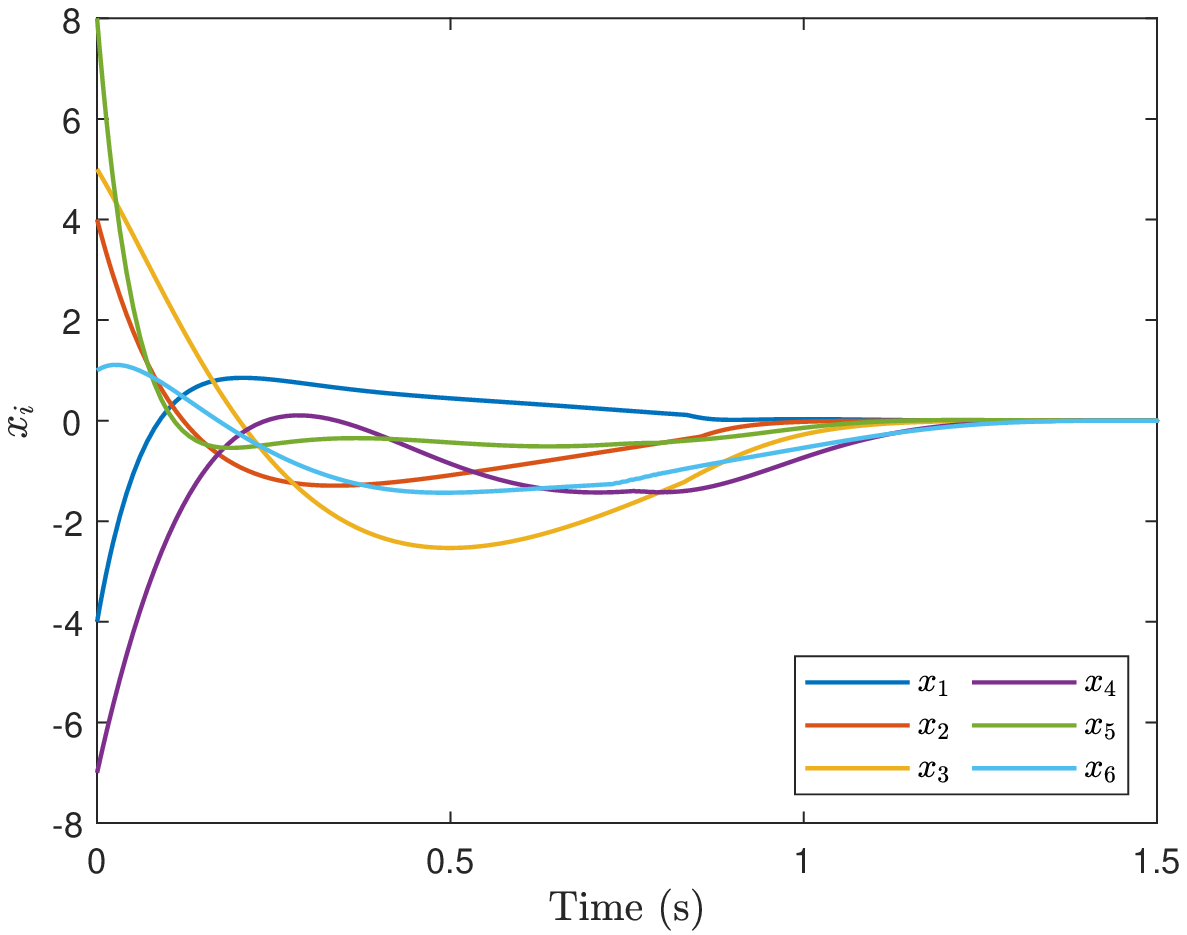}
  \end{minipage}}
\caption{State evolution of the CAN(2) under the signed digraphs of Figure 3. (a) Under Figure 3 (a). (b) Under Figure 3 (b)}
\end{figure}\label{Fig10}

Example 6. Let us consider the CAN \eqref{eq2} with ${d_k}= \cos \left( {kt - \frac{\pi}{3}} \right)$ under the signed digraphs in Figure 5. The control parameters of \eqref{eq47} and \eqref{eq48} are chosen as $\kappa = 3,T_r = 0.4, T_s=0.6,{\rho_1} = 0.2,{\rho_2} = 0.1,{\mu_1} = 2,{\mu_2} = 0.1,{\mu_3} = 0.5$. Set $X\left( 0 \right) = \left[2.6,-1.2,-1.2,-1,-0.2,0.9,-2.9,2,0.3,2.1,-1,-0.3, -2.7,\right.\\
\left.-2,1,-1\right]^{\rm{T}}$ and $\sigma\left( 0 \right) = \left[  2.9,-3,0.5,0,0.75,-0.8,-1.5,\right.\\
\left.3.8,2.3,3.6,0.2,-0.27,-4,-2.3,-0.5,-2.9\right]^{\rm{T}}$. The trajectories of the sliding variables and the agents under the signed digraphs of Figure 5 are depicted in Figures 11 and 12, respectively. From Figure 11, one see that the sliding variables converge to zero within the predefined finite time $T_r = 0.4s$. From Figure 12, one observe that bipartite containment and stability are, respectively, accomplished within the predefined finite time $T_r = 1s$ under the Figure 5(a) and (b), which validates Theorem \ref{thm6}.

\begin{figure}
\subfigure[]{
  \begin{minipage}{0.23\textwidth}
  \centering
  \includegraphics[scale=0.3]{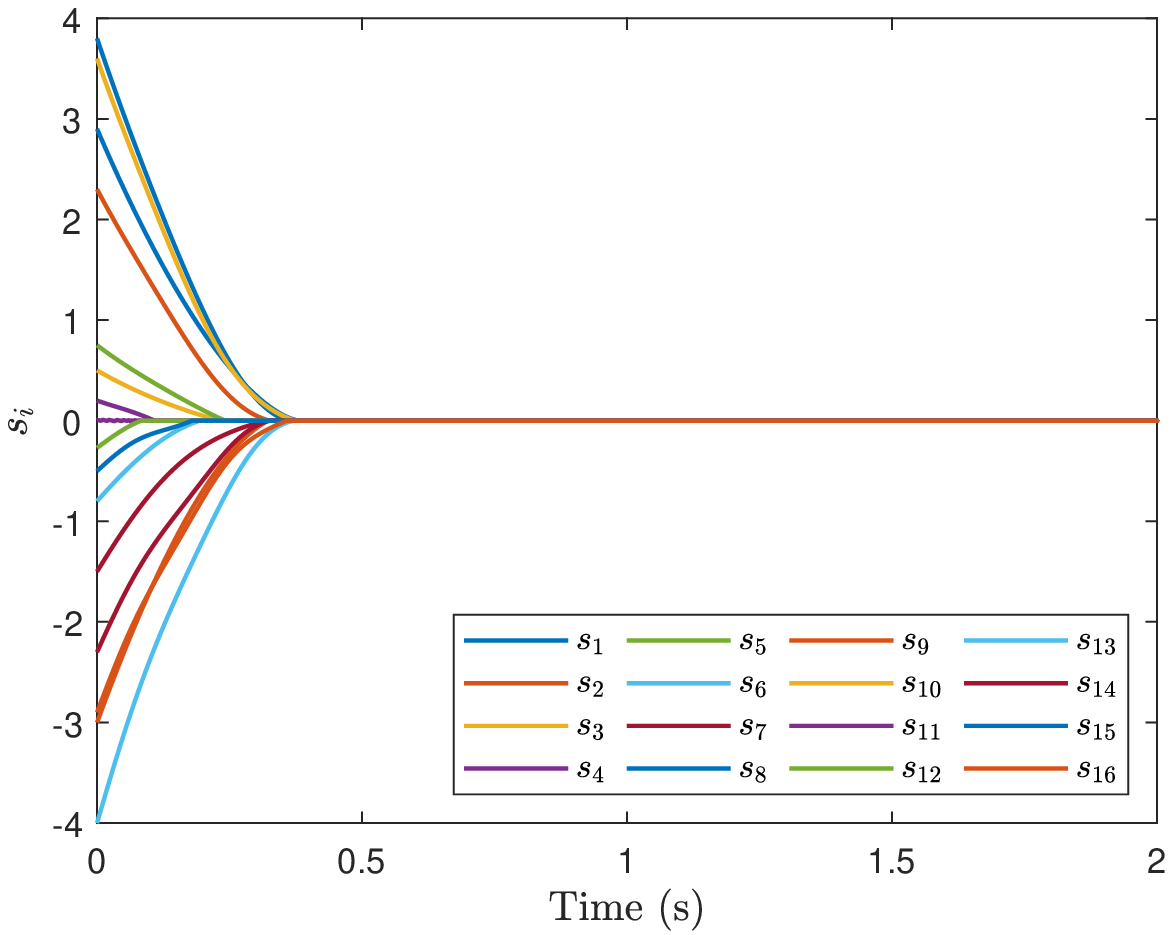}
   \end{minipage}}
\subfigure[]{
  \begin{minipage}{0.23\textwidth}
  \centering
  \includegraphics[scale=0.3]{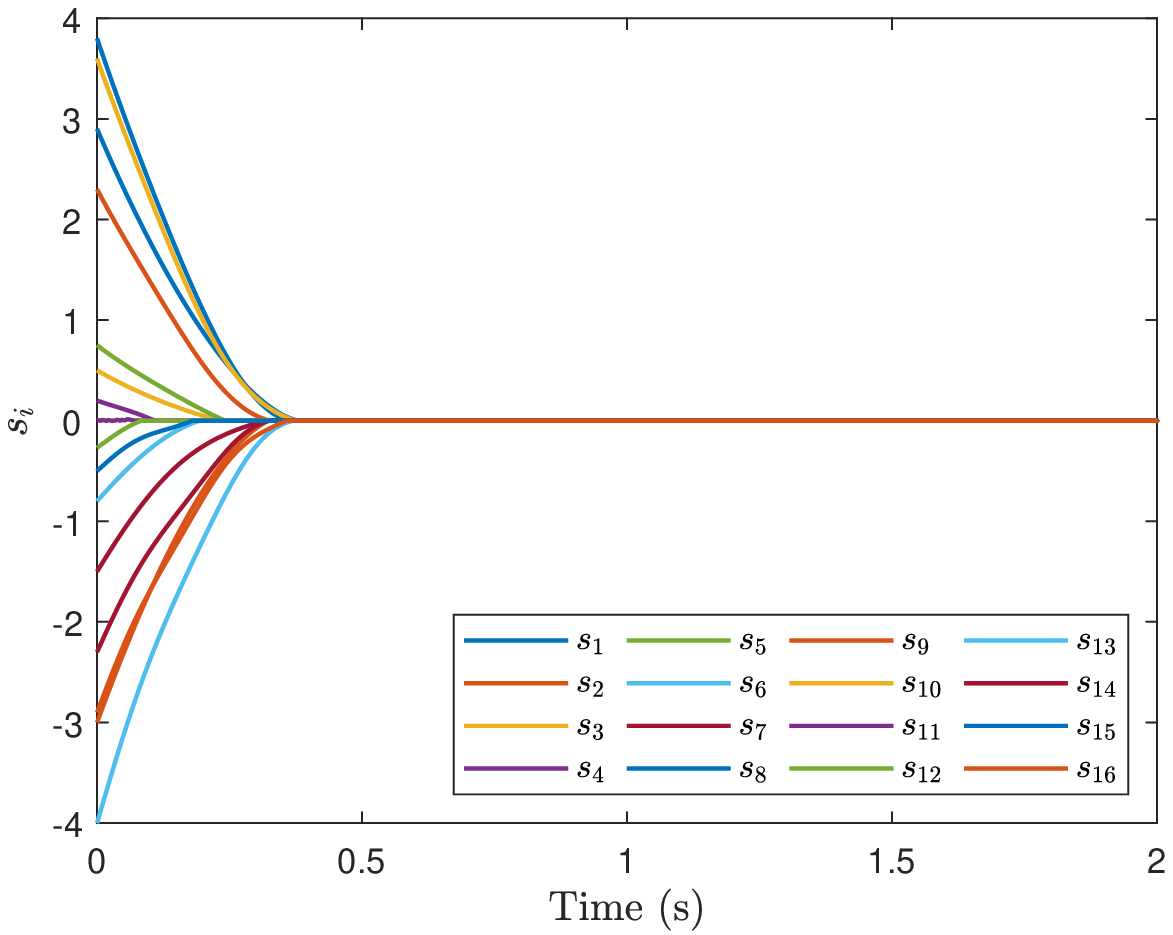}
  \end{minipage}}
\caption{Evolution of the sliding variables  under the signed digraphs of Figure 5. (a) Under Figure 5 (a). (b) Under Figure 5 (b). }
\end{figure}\label{Fig11}

\begin{figure}[!t]
\subfigure[]{
  \begin{minipage}{0.23\textwidth}
  \centering
  \includegraphics[scale=0.3]{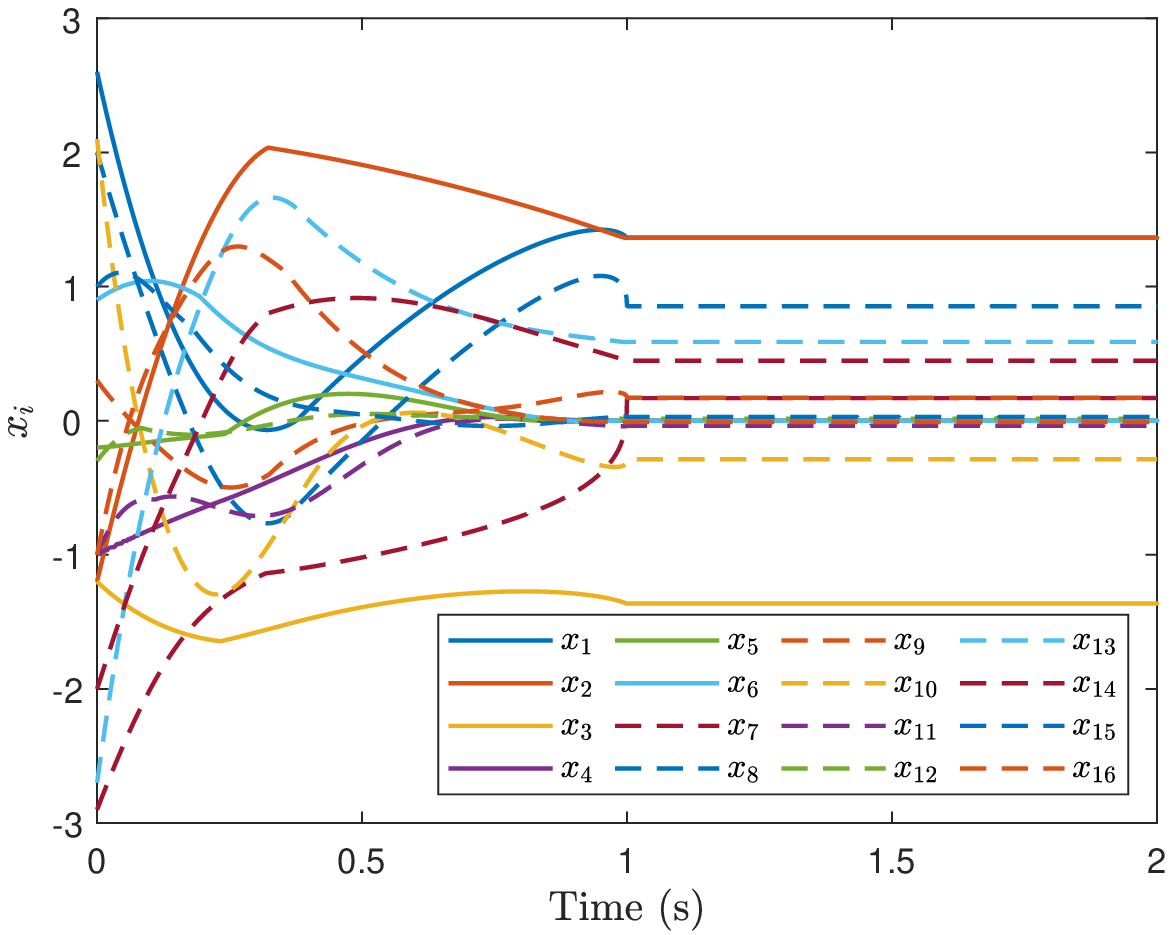}
   \end{minipage}}
\subfigure[]{
  \begin{minipage}{0.23\textwidth}
  \centering
  \includegraphics[scale=0.3]{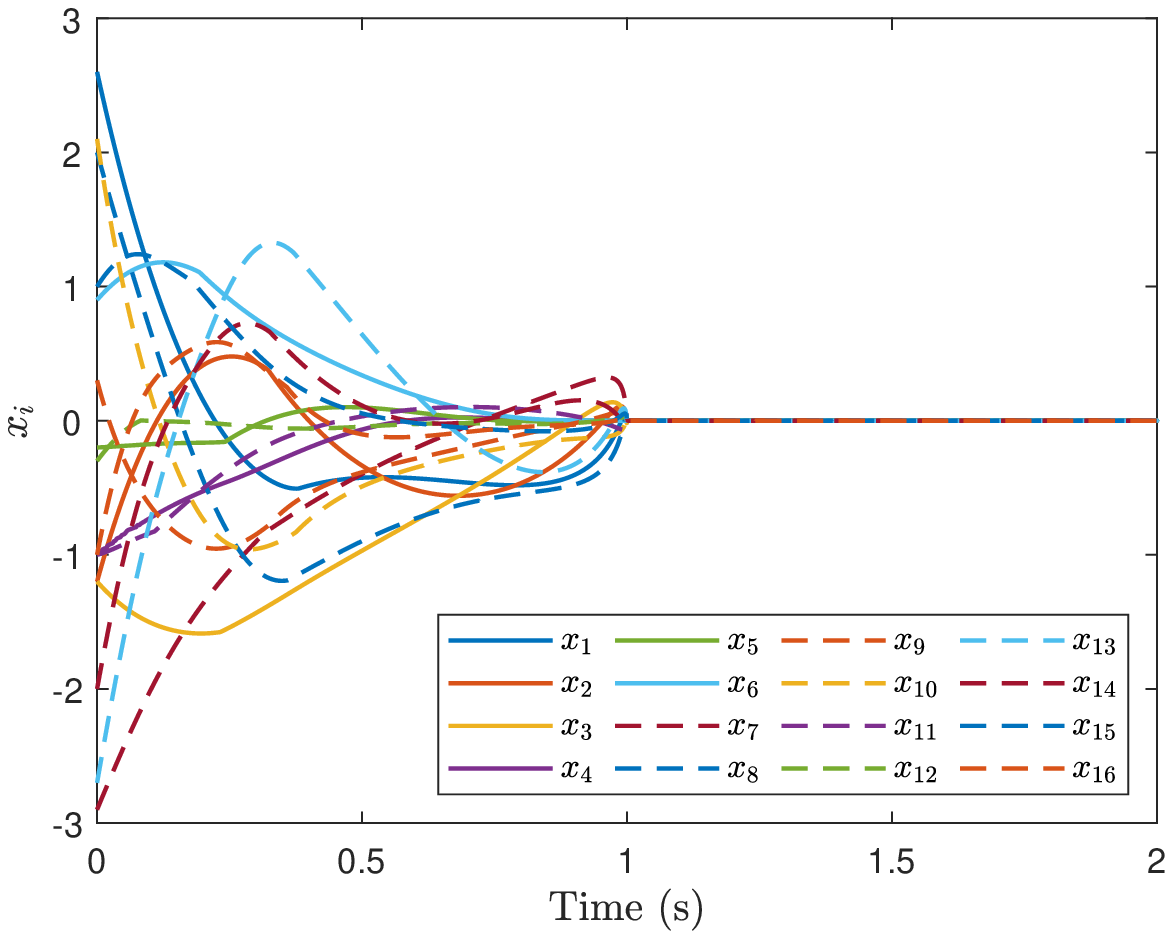}
  \end{minipage}}
\caption{State evolution of the CAN(2) under the signed digraphs of Figure 5.(a) Under Figure 5 (a). (b) Under Figure 5 (b).}
\end{figure}\label{Fig12}

\section{Conclusion}\label{sec5}

The problems of PTCC, including stability, bipartite consensus, interval bipartite consensus, and bipartite containment control, for single-integrator CANs with and without external disturbances over general signed digraphs without any structure restriction have been addressed, for which a unified analysis and design framework has been provided. By using the relative states of neighboring agents, a prescribed-time control protocol with time-invariant and time-varying gains has been firstly developed to handle the PTCC problems for CANs without disturbances. Then, based on the control protocol for the disturbance-free nominal CANs, a novel prescribed-time sliding mode control protocol is developed to address the PTCC problems for CANs subject to external disturbances. In particular, it has been demonstrated that with the proposed control protocols, a CAN can reach bipartite containment (respectively, stability) in prescribed finite time if the underlying signed digraph is weakly connected and has at least one structurally balanced CSC (respectively, all its CSCs are structurally unbalanced). The conditions on sign-symmetry, structural balance, and connectivity of signed digraphs typically assumed in the existing literature are removed. Finally, simulation examples have been provided to validate the derived results.

We have only considered the PTCC problems for CANs with single-integrator dynamics. Possible extensions of the established results include addressing CANs with general linear dynamics, nonlinear dynamics, switching topologies and communication delays.


%

%



\ifCLASSOPTIONcaptionsoff
  \newpage
\fi



\bibliographystyle{IEEEtran}
\bibliography{mybibfile}
\end{document}